\documentclass[authorcolumns,numberwithinsect,a4paper]{no-lipics}

\usepackage{amssymb}
\usepackage{mathtools}
\usepackage{bm}
\usepackage{stmaryrd}
\usepackage{booktabs}
\usepackage[T1]{fontenc}
\usetikzlibrary{arrows,arrows.meta,shapes.misc, decorations.pathreplacing,decorations.pathmorphing,calligraphy, patterns.meta}
\usetikzlibrary{calc}
\tikzset{dotmark/.style={circle,fill,inner sep=1.5pt}}
\tikzset{emptymark/.style={circle,draw,fill=white,inner sep=1.5pt}}
\tikzset{crossmark/.style={thick,inner sep=1.5pt}}

\def\twoheadleadsto{\tikz[baseline=(a.base)]{%
        \node at (.2,0) {\(\leadsto\)};%
        \fill[white] (0,-.1+.012) -- (.25,-.1+.012) -- (.335,0.012) -- (.25,.113) --
        (0,.113) --cycle;
        \node at (.125,0) {\(\leadsto\)};%
        \node (a) at (.4/2,-.0) {\phantom{\(\leadsto\)}};%
}}

\def\ponto#1{\twoheadleadsto\kern-.3em{}_{#1}\kern.3em}

\def\setn#1{\bm{[}\,#1\,\bm{]}}

\makeatletter
\renewenvironment{cases}{%
    \matrix@check\cases\env@cases
}{%
    \endarray\right.%
}
\def\env@cases{%
    \let\@ifnextchar\new@ifnextchar
    \left\lbrace
        \def\arraystretch{1.1}%
        \array{@{\;}c@{\quad}l@{}}%
}
\makeatother

\newcommand{\nat}{\mathbb{N}}
\newcommand{\Z}{\mathbb{Z}}

\newcommand{\CoutDeg}[2]{\mathrm{d}^+(#1, #2)}  
\newcommand{\inDeg}[2]{\mathrm{d}^-_{#1}(#2)}   
\newcommand{\outDeg}[2]{\mathrm{d}^+_{#1}(#2)}  

\DeclareMathOperator{\id}{id}

\def\NUM{\text{\#}}
\def\W#1{\ensuremath { W{{[}\,#1\,{]}}}}
\def\W#1{\ensuremath { W{{[}#1{]}}}}
\def\w{\NUM\W1} 
\def\Dw{\W1}    
\newcommand{\fpt}{\leq_{\mathrm{T}}^{\mathrm{fpt}}}

\DeclareMathOperator{\symdifOperator}{\triangle}
\newcommand{\sym}[1]{\mathfrak{S}_{#1}} 
\newcommand{\flip}[2]{#1 \symdifOperator #2} 
\newcommand{\symdif}[2]{#1 \symdifOperator #2} 

\newcommand{\outE}{\mathcal{O}}
\newcommand{\inE}{\mathcal{I}}
\newcommand{\intE}{\mathcal{B}}

\newcommand{\dSet}{{D}}
\newcommand{\dCard}{{d}}
\newcommand{\zSet}{{Z}}
\newcommand{\zCard}{{z}}

\DeclareMathOperator{\IS}{IS} 
\DeclareMathOperator{\tw}{tw} 

\newcommand{\graphs}[1]{\mathcal{G}_{#1}} 
\newcommand{\auts}[1]{\ensuremath{\mathrm{Aut}(#1)}}
\newcommand{\indGraph}[2]{#1[#2]} 
\newcommand{\edgesub}[2]{#1\{#2\}} 

\newcommand{\tranTour}{\text{\sc T\kern-.325emt}}

\DeclareMathOperator{\sig}{sig}
\newcommand{\DE}[1]{E(#1)} 

\newcommand{\Tperm}[2]{#1^{#2}} 

\newcommand{\Miso}[2]{\sigma_{#2}^{#1}} 
\newcommand{\MTour}[2]{#1(#2)} 
\newcommand{\Mperm}[2]{#1^{#2}} 
\newcommand{\matchOrder}[1]{\mathcal{M}_#1} 
\newcommand{\matchSet}[1]{\tilde{\mathcal{M}_#1}} 

\newcommand{\tourG}[2]{#1^{(#2)}} 
\newcommand{\match}[1]{{M}_{#1}} 
\newcommand{\cmatch}[1]{\overline{M}_{#1}} 
\newcommand{\GraphT}[2]{#1_{(#2)}}

\newcommand{\TourSymb}{\mathrm{To}}

\newcommand{\tourProb}{\text{\ref{prob:indsub:tour}}}
\newcommand{\tours}[2]{\mbox{\NUM{}\ensuremath{\mathrm{IndSub}(#1 \to #2)}}}

\newcommand{\subProb}{\NUM{}\ensuremath{\textsc{Sub}}}
\newcommand{\subs}[2]{\mbox{\NUM{}\ensuremath{\mathrm{Sub}(#1 \to #2)}}}

\newcommand{\indsubProb}{\NUM{}\ensuremath{\textsc{IndSub}}}
\newcommand{\indsubs}[2]{\mbox{\NUM{}\ensuremath{\mathrm{IndSub}(#1 \to #2)}}}

\newcommand{\homProb}{\NUM{}\ensuremath{\textsc{Hom}}}

\newcommand{\CPsubProb}{\text{\ref{prob:cp:sub}}}
\newcommand{\CPsubs}[2]{\mbox{\NUM{}\ensuremath{\mathrm{cp}\text{-}\mathrm{Sub}(#1 \to #2)}}}

\newcommand{\CPindsubProb}{\NUM{}\ensuremath{\textsc{cp}\text{-}\textsc{IndSub}}}
\newcommand{\CPindsubs}[2]{\mbox{\NUM{}\ensuremath{\mathrm{cp}\text{-}\mathrm{IndSub}(#1 \to #2)}}}

\newcommand{\CFsubProb}{\text{\ref{prob:cf:IndSub}}}
\newcommand{\CFsubs}[2]{\mbox{\NUM{}\ensuremath{\mathrm{cf}\text{-}\mathrm{IndSub}(#1 \to #2)}}}

\newcommand{\GraphCFsubProb}{\NUM{}\ensuremath{\textsc{cf}\text{-}\textsc{Sub}}}
\newcommand{\GraphCFsubs}[2]{\mbox{\NUM{}\ensuremath{\mathrm{cf}\text{-}\mathrm{Sub}(#1 \to #2)}}}

\newcommand{\cliqueProb}{\text{\ref{prob:clique}}}
\newcommand{\clique}[2]{\ensuremath{\NUM{}\mathrm{Clique}_{#1}(#2)}}

\newcommand{\CFcliqueProb}{\text{\ref{prob:cf:clique}}}
\newcommand{\CFclique}[2]{\ensuremath{\NUM{}\mathrm{cf}\text{-}\mathrm{Clique}_{#1}(#2)}}

\newcommand{\DecTourProb}{\ensuremath{\textsc{Dec}\text{-}\textsc{IndSub}_{\TourSymb}}}
\newcommand{\DecTours}[2]{\mbox{\ensuremath{\mathrm{Dec}\text{-}\mathrm{IndSub}(#1 \to #2)}}}

\newcommand{\DecSubProb}{\ensuremath{\textsc{Dec}\text{-}\textsc{Sub}}}

\newcommand{\DecCPsubProb}{\ensuremath{\textsc{Dec}\text{-}\textsc{cp}\text{-}\textsc{Sub}}}

\newcommand{\DecCFsubProb}{\ensuremath{\textsc{Dec}\text{-}\textsc{cf}\text{-}\textsc{IndSub}_{\TourSymb}}}
\newcommand{\DecCFsubs}[2]{\mbox{\ensuremath{\mathrm{Dec}\text{-}\mathrm{cf}\text{-}\mathrm{IndSub}(#1 \to #2)}}}

\newcommand{\DecCliqueProb}{\ensuremath{\textsc{Dec}\text{-}\textsc{Clique}}}
\newcommand{\DecClique}[2]{\ensuremath{\mathrm{Dec}\text{-}\mathrm{Clique}_{#1}(#2)}}

\newcommand{\DecCFcliqueProb}{\ensuremath{\textsc{Dec}\text{-}\textsc{cf}\text{-}\textsc{Clique}}}

\newcommand{\sat}[1]{\ensuremath{#1\text{-}\textsc{SAT}}}

\newcommand{\Exp}[1]{\ensuremath{\mathrm{cx}(#1)}} 
\newcommand{\cExp}[1]{\Exp{\cliqueProb_{#1}}} 
\newcommand{\tExp}[1]{\Exp{\tourProb(\{#1\})}} 

\newcommand{\DeccExp}[1]{\Exp{\DecCliqueProb_{#1}}} 
\newcommand{\DectExp}[1]{\Exp{\DecTourProb(\{#1\})}} 

\newcommand{\DecCFcExp}[1]{\Exp{\DecCFcliqueProb_{#1}}} 
\newcommand{\DecCFtExp}[1]{\Exp{\DecCFsubProb(\{#1\})}} 

\def\aename#1{\ensuremath \widehat{#1}}
\def\ae#1#2{\ensuremath \aename{#1}(#2)}

\newcommand{\core}[1]{\mathrm{sl}(#1)}
\newcommand{\tri}{\mathrm{Tri}}

\newcommand{\TT}[1]{\alpha(#1)}
\newcommand{\Tout}[1]{N^-(#1)}
\newcommand{\Tvec}[2]{\Tout{#1} \cap #2}

\title{The Complexity of Finding and Counting Subtournaments}

\author{Simon Döring}{Max Planck Institute for Informatics and\\Saarbrücken Graduate School
of Computer Science\\Saarland Informatics Campus\\Saarbrücken, Germany}{sdoering@mpi-inf.mpg.de}{https://orcid.org/0009-0002-6667-5257}{}
\author{Sarah Houdaigoui}{National Institute of
Informatics,\\The~Graduate~University~for~Advanced~Studies, SOKENDAI\\Tokyo, Japan}{shoudaigoui@nii.ac.jp}{https://orcid.org/0009-0003-5490-4806}{}
\author{Lucas Picasarri-Arrieta}{National Institute of
Informatics,\\Tokyo, Japan}{lpicasarr@nii.ac.jp}{https://orcid.org/0000-0003-0414-8136}{}
\author{Philip Wellnitz}{National Institute of
Informatics,\\The~Graduate~University~for~Advanced~Studies, SOKENDAI\\Tokyo, Japan}{wellnitz@nii.ac.jp}{https://orcid.org/0000-0002-6482-8478}{}
\authorrunning{S. Döring, S. Houdaigoui, L. Picasarri-Arrieta, and P. Wellnitz}

\acknowledgements{We thank Radu Curticapean for insightful discussions.}

\begin{document}
\pagenumbering{roman}
\maketitle
\begin{abstract}
We study the complexity of counting and finding small tournament patterns
inside large tournaments. Given a fixed tournament $T$ of order
$k$, we write $\text{\#}\textsc{IndSub}_{\mathrm{To}}(\{T\})$ for the problem whose
input is a tournament $G$ and the task is to compute the number of subtournaments
of $G$ that are isomorphic to $T$. Previously, Yuster [Yus25] obtained
that $\text{\#}\textsc{IndSub}_{\mathrm{To}}(\{T\})$ is hard to compute for random
tournaments $T$. We consider a new approach that uses linear combinations
of subgraph-counts [CDM17] to obtain a finer analysis of the complexity
of $\text{\#}\textsc{IndSub}_{\mathrm{To}}(\{T\})$.

We show that for all tournaments $T$ of order $k$ the problem
$\text{\#}\textsc{IndSub}_{\mathrm{To}}(\{T\})$ is always at least as hard
as counting $\lfloor 3k/4 \rfloor$-cliques. This immediately yields
tight bounds under ETH. Further, we consider
the parameterized version of $\text{\#}\textsc{IndSub}_{\mathrm{To}}(\mathcal{T})$
where we only consider patterns $T \in \mathcal{T}$ and that is parameterized
by the pattern size $|V(T)|$. We show that $\text{\#}\textsc{IndSub}_{\mathrm{To}}(\mathcal{T})$ is
$\text{\#}W[1]$-hard as long as $\mathcal{T}$ contains infinitely many
tournaments.

However, the situation drastically changes when we consider the decision version
of the problem $\textsc{Dec-IndSub}_{\mathrm{To}}(\{T\})$. Here, we have to decide whether
an input tournament $G$ contains a subtournament that is isomorphic to $T$.
According to a famous theorem by Erd\H{o}s and Moser [EM64] the problem
$\textsc{Dec-IndSub}_{\mathrm{To}}(\{T\})$ is easy to solve for transitive tournaments.
In a first step, we extend this result and present other kinds of tournaments
for which the parameterized version of $\textsc{Dec-IndSub}_{\mathrm{To}}(\mathcal{T})$
is FPT.

In a next step, we show that certain structures inside a tournament $T$
can be exploited to show that $\textsc{Dec-IndSub}_{\mathrm{To}}(\{T\})$
is at least as hard as finding large cliques. We show that almost all tournaments
have this specific structure. Hence, $\textsc{Dec-IndSub}_{\mathrm{To}}(\{T\})$
is hard for almost all tournaments.

Lastly, we combine this result
with our FPT result to construct, for each constant $c$, a class of tournaments
$\mathcal{T}_c$ for which $\textsc{Dec-IndSub}_{\mathrm{To}}(\mathcal{T}_c)$ is FPT
but which cannot be solved in time $O(f(k) \cdot n^{\alpha c})$,
unless ETH fails. Here, $\alpha > 0$ is a global constant independent of $c$.
\end{abstract}

\clearpage
\thispagestyle{plain}
\tableofcontents
\clearpage
\pagenumbering{arabic}

\section{Introduction}

Detecting and counting small \emph{patterns} inside large \emph{host} structures (like graphs) is one of the
oldest problems in computer science and has many
applications in other scientific fields such as
statistical physics~\cite{10.1007/978-3-031-92935-9_18}, database theory~\cite{10.1145/800105.803397,
dell_et_al:LIPIcs.ICALP.2019.113,10.1145/3517804.3526231, 10.1145/380752.380867},
network science~\cite{doi:10.1126/science.298.5594.824,doi:10.1126/science.1089167},
and computable biology~\cite{10.1093/bioinformatics/btn163, 10.1007/11599128_7,
10.1007/978-3-319-21233-3_5}, to name but a few examples.

Formally, given a (small) pattern graph \(H\) and a (large) host graph \(G\),
we write $\indsubs{H}{G}$ for the number of induced occurrences of \(H\) in \(G\), that
is, the number of subgraphs of \(G\) that are isomorphic to \(H\).
For a class of graphs $\mathcal{H}$, in the problem $\indsubProb(\mathcal{H})$ we are
given a pattern graph $H \in \mathcal{H}$
and a host graph $G$ and the task is to compute $\indsubs{H}{G}$.
Due to its importance, many researchers investigated $\indsubProb(\mathcal{H})$ and its
computational complexity---in particular from a \emph{parameterized} point of view where
we parameterize by \(|V(H)|\).

In full generality, $\indsubProb(\mathcal{H})$ is \w-hard if and only if the set $\mathcal{H}$
is infinite~\cite{DBLP:conf/icalp/ChenTW08} and there are tight lower bounds
under the Exponential Time Hypothesis (ETH)~\cite{topo}.
Similar results exist for counting \emph{directed}
subgraphs~\cite{BLR23}.
Having understood the general problem, the research interest thus shifted toward
understanding the complexity of $\indsubProb(\mathcal{H})$ for more specific classes of
host graphs \(G\).
For instance in~\cite{Planar_subgraph}, Eppstein shows that $\indsubProb(\mathcal{H})$
has a linear time algorithm for fixed patterns when the host graph
is planar.%
\footnote{Also see~\cite{bodlaender, Planar_subgraph2} for more recent results.}
Similarly, in the realm of directed graphs, many researchers studied graphs of bounded
outdegree~\cite{bera_et_al:LIPIcs.ITCS.2020.38, BR22, Pena, BLR23}.
Still, we are missing a comprehensive understanding which pattern and host graph
combinations allow for efficient algorithms---and which do not.

\begin{center}
    \textit{Which classes of patterns allow for efficient counting algorithms in which
    classes of host graphs?}
\end{center}

We answer the above question in the directed setting for tournaments, that is,
directed graphs with exactly one directed edge between any pair of vertices.

\begin{problem}{$\NUM{}\textsc{IndSub}_{\TourSymb}$}
    \label{prob:indsub:tour}
    \PInput{A directed graph \(T \in \mathcal{T}\) and a tournament \(G\).}
    \POutput{$\tours{T}{G}$; that is, the number of sets $A\subseteq V(G)$ such that $\indGraph{G}{A}$
    is isomorphic to $T$.}
    \PParameter{\(k \coloneqq |V(T)|\).}
\end{problem}

\begin{restatable*}{mtheorem}{thmsubpara}
    \dglabel{maintheorem:sub:para}[lem:subs:to:color,lem:CFsubs:CFclique,lem:sig:lower:bound,lem:CFclique:w1]($\tourProb(\mathcal{T})$
    is \w-hard)
    Write $\mathcal{T}$ for a recursively enumerable class
    of directed graphs.
    The problem $\tourProb(\mathcal{T})$ is \w-hard if $\mathcal{T}$ contains infinitely many tournaments
    and FPT otherwise.
\end{restatable*}

Complementing our first result, we also obtain fine-grained lower bounds under
ETH.%
\footnote{Observe that $\tourProb(\{T\})$ is equivalent to the problem
of computing $\tours{T}{\star}$ for a fixed tournament $T$.}

\begin{restatable*}[{Fine-grained lower bounds for
    $\tourProb(\{T\})$}]{mtheorem}{thmsub}\dglabel{maintheorem:sub}[def:match,
    maintheorem:cfsub,lem:cmatch:to:clique]
    For all tournaments $T$ of order $k$,  assume that there is
        an algorithm that reads the whole input and computes $\tourProb(\{T\})$
        for any tournament of order $n$ in time $O(n^\gamma)$. Then there is an algorithm that
        solves $\cliqueProb_{\lfloor 3k/4 \rfloor}$ for any graph of order $n$
        in time $O(n^{\gamma})$.

    Further, assuming ETH, there is a global constant $\beta > 0$, such that
        no algorithm that reads the whole input computes $\tourProb(\{T\})$
        for any graph of order $n$ in time $O(n^{\beta k})$.
\end{restatable*}

Our results continue the long line of research on tournaments in
general~(see e.g.~\cite{DBLP:conf/ijcai/KimW15,grohe_et_al:LIPIcs.ICALP.2024.78}), and in particular significantly improve upon
a recent work of Yuster~\cite{Yuster25} that contains hardness results for $\tourProb(\mathcal{T})$
for specific patterns, but not yet a complete characterization for
all patterns.%
\footnote{Yuster~\cite{Yuster25} considers counting \emph{subgraphs} where
    both host and pattern are tournaments. However, the number of induced subgraphs is
    equal to the number of subgraphs if both $H$ and $G$ are tournaments.
    If $G$ is a tournament and $H$ is not a tournament then $\tours{H}{G} = 0$
    and \(\tours{H}{G}\) is therefore easy to compute.}

As an example, \cite[Theorem~1.4]{Yuster25} gives an $O(n^\omega)$-time algorithm to count
any specific tournament pattern with four vertices,
where $\omega < 2.3713$ is the matrix multiplication exponent~\cite{Clique_counting}.
With \cref{maintheorem:sub}, we obtain a matching (conditional) lower
bound (see \cref{remark:sub}).

Naturally, lower bounds for counting problems would also follow directly from lower bounds
for corresponding decision problems. Hence, of specific interest are problems with
decision versions but hard counting versions.
In the realm of tournaments, famously, the Erd\H{o}s-Moser-Theorem~\cite{ErMo64}
ensures that every tournament $T$ with $k$ vertices
contains a \emph{transitive tournament} of logarithmic size---thereby rendering
easy-to-solve the decision problem $\DecTourProb(\{\tranTour_k\})$ of detecting a transitive
tournament of size \(k\).
Now, \cref{maintheorem:sub:para,maintheorem:sub} show that the counting version
$\tourProb(\{\tranTour_k\})$ is indeed hard.

We extend the above argument using the Erd\H{o}s-Moser-Theorem to tournament patterns that
consist in a large transitive tournament (a \emph{spine}) \(S\) and two sets \(R_+\) and
\(R_-\) (the \emph{ribs}) such that all edges between \(S\) and \(R_+\) are directed
toward \(S\) and all edges between \(S\) and \(R_-\) are directed toward \(R_-\).
For a given tournament \(T\), we write \(\core{T}\) for the largest possible spine of a
decomposition of the above shape---consul \cref{def:core} for the formal definition.
In particular, we show how to detect a tournament pattern \(T\) with \(c \coloneqq |V(T)|
- \core{T}\) in time $O( f(k) \cdot n^{c + 2})$ for some computable \(f\)---which we also
essentially match with a corresponding conditional lower bound.

\begin{restatable*}[For all $c > 0$, there is a $\mathcal{T}_c$ for which
    $\DecTourProb(\mathcal{T}_c)$ is in time $f(k) \, n^{\Theta(c)}$]{mtheorem}{thmdecsubpara}
    \dglabel{maintheorem:Dec:sub:para}[theo:core:easy,maintheorem:Dec:sub]
    Assuming ETH, there is a global constant $\alpha > 0$ such that all of
    the following hold.

    \begin{itemize}
        \item For any constant $c > 0$ there is a class of infinitely many tournaments $\mathcal{T}_c$
            such that $|V(T)| - \core{T} \leq c$ for all $T \in \mathcal{T}_c$.
            Thus, the problem $\DecTourProb(\mathcal{T}_c)$ is FPT and in time $O( f(k) \cdot n^{c + 2})$
            for some computable function $f$.
        \item Further, there is a tournament $T \in \mathcal{T}_c$
            that has a TT-unique partition $(\dSet, \zSet)$ with $|\zSet| \geq c$. Hence,
            no algorithm that reads the whole input and solves $\DecTourProb(\mathcal{T}_c)$ in time $O(f(k) \cdot n^{\alpha c})$
            for any computable function $f$. Here $k$ is the order of the pattern tournament (parameter)
            and $n$ is the order of the host tournament.
        \qedhere
    \end{itemize}
\end{restatable*}

Observe that \cref{maintheorem:Dec:sub:para} offers a smooth trade-off from transitive
tournament patterns (\(c = 0\)) where the decision version is much easier to solve compared to the
counting version and (somewhat) general patterns (around \(c \ge \beta k\)) where both variants
are essentially equally hard.

\subsection{Related Work}

The study of the complexity of pattern detection and pattern counting in graphs is an
active area of research.
In general, detecting (directed) subgraphs is a classical NP-hard problem~\cite{Cook71,Ullmann76}, as it
generalizes finding cliques and (directed) Hamiltonian cycles \cite{DBLP:books/fm/GareyJ79}.
Hence, long lines of research are concerned with restricted variations of the
general problem, where either the set of allowed host graphs or the set of allowed pattern
graphs (or both) is restricted.

When restricting the possible host graphs, results include
algorithms to detect (induced) subgraphs in trees~\cite{MATULA197891,10.1145/3093239},
to detect and count patterns in planar host graphs~\cite{Planar_subgraph,bodlaender,Planar_subgraph2},
to detect induced subgraphs that satisfy a certain property in restricted host graphs~\cite{DBLP:conf/fct/EppsteinGH21},
and to count patterns in somewhere dense host graphs~\cite{BGMR24}, to name but a few examples.
Recently, researchers also started to study counting for directed host graphs with bounded
outdegree \cite{Pena,bera_et_al:LIPIcs.ITCS.2020.38, BR22, BLR23}.%
\footnote{
    Most results are about counting induced subgraphs for undirected
    host graphs with bounded degeneracy. However, the authors
    reformulate the problem by orientating the graphs such that they have a
    bounded outdegree.
}

When restricting the possible pattern graphs, results include algorithms for finding
specific, fixed patterns, such as cliques~\cite{n1985, VASSILEVSKA2009254},
cycles~\cite{DBLP:journals/siamcomp/ItaiR78, DBLP:journals/algorithmica/AlonYZ97},
paths~\cite{DBLP:conf/icalp/Koutis08, DBLP:journals/ipl/Williams09},
or bounded treewidth graphs~\cite{DBLP:journals/jacm/AlonYZ95}, among others. Many
of these results carry over to counting subgraphs in undirected/directed graphs (see~\cite{DBLP:journals/algorithmica/AlonYZ97}
for paths and cycles and \cite{Clique_counting} for cycles).

In this setting, researchers also study pattern detection/counting problems from a
parameterized point of view. In the simplest setting (as is also done in this paper), one
considers the size of the pattern to be the parameter (with the assumption that this
parameter is somewhat small).
Unfortunately, the general parameterized problem $\indsubProb(\mathcal{H})$
is \w-hard~\cite{DBLP:conf/icalp/ChenTW08} for any infinite class of pattern graphs $\mathcal{H}$
and has tight bounds under ETH~\cite{DBLP:conf/stoc/CurticapeanDM17, topo}. Both results carry over to the decision case
and to detecting/counting directed subgraphs. Hence researchers are interested in understanding the
complexity for special classes of pattern graphs and (linear) combinations of such
counts~\cite{DBLP:journals/jcss/JerrumM15, DBLP:journals/combinatorica/JerrumM17,
DBLP:conf/stoc/FockeR22, rsw24,DMW24, CN24}.

Many results in this area rely on a direct
link between counting graph homomorphisms and counting (induced)
subgraphs~\cite{DBLP:conf/stoc/CurticapeanDM17}---thereby connecting the diverse
complexity landscapes of the individual problem families:
For restricted classes of pattern graphs $\mathcal{H}$, counting graph homomorphisms
$\homProb(\mathcal{H})$ is \w-hard whenever $\mathcal{H}$ has unbounded treewidth~\cite{DBLP:journals/tcs/DalmauJ04}.
Further, the problem has almost tight bounds under ETH~\cite{DBLP:journals/toc/Marx10}.
The problem of counting/detecting homomorphism can be
generalized to other problems like {conjunctive queries}~\cite{10.1145/380752.380867, dell_et_al:LIPIcs.ICALP.2019.113, 10.1145/3517804.3526231}
or the constraint satisfaction problem~\cite{DBLP:journals/jacm/Bulatov13, DBLP:journals/siamcomp/BulatovM14}.

\section{Technical Overview}
\dglabel^{sec:techov}[theo:cfsub:hardest:term,lem:ae:clqiue,cor:cmatch:densest]

\subsection{The Complexity of Counting Tournaments}

\paragraph*{Exploring the Limits of the Existing Approach}

For the discussion of our techniques, we start with \cref{maintheorem:sub:para}, that is,
our \w-hardness result.
It is instructive to take a step back and reflect on the approach taken by
Yuster~\cite{Yuster25}, that---while also making progress into the same direction---falls
short of obtaining reductions that work for \emph{all} (classes of) tournaments and not
just \emph{almost all} of them.

\begin{theoremq}[{\cite[Theorem 1.11 (Counting Results)]{Yuster25}}]
    Fix an integer $k \geq 3$.
    Then, there is a tournament $T$ such that
    any algorithm that computes $\tourProb(\{T\})$ in time $O(n^{\gamma})$
    implies an algorithm that solves $\DecCliqueProb_{k - O(\log(k))}$ in
    time $O(n^{\gamma + \varepsilon})$, for all $\varepsilon > 0$.

    In fact, as $k$ goes to infinity, almost all tournaments $T$
    on $k$ vertices satisfy the property above.
\end{theoremq}

Due to the absence of useful tools to show the hardness of counting directed patterns in
directed graphs, Yuster's proof~\cite{Yuster25} defaults to the problem of counting
(undirected) cliques as the base of the hardness.
Naturally, while removing the directions of all edges of the given tournaments recovers
the clique counting problem, doing so loses too much information to be of any use in a
reduction.
Hence, as an intermediate step, Yuster first considers a \emph{colorful} variant
$\CFsubProb(\{T\})$, where the vertices of the pattern and host tournaments are assigned
colors and the task is to count just those induced occurrences that have each
color of the pattern exactly once (the colors of an occurrence might be distributed differently
compared to the pattern, though).

With the help of colors, Yuster is able to reduce to counting \emph{undirected}
cliques (and then finally to the corresponding uncolored variant).
For the reduction to colored cliques, Yuster defines and exploits a structure of
tournaments called \emph{signature} (see \cref{def:sig} or~\cite[Definition
2.2]{Yuster25})---the smaller the signature of the tournament, the tighter the reduction
becomes. Unfortunately, Yuster is able to show only that  \emph{almost} all tournaments have a small
signature. We overcome this minor deficit with a new, short and ad-hoc construction that shows that
the signature of any tournament of order \(k\) is of size at most \(k - \log(k) / 4\);
yielding \cref{maintheorem:sub:para} (see \cref{lem:sig:lower:bound}).

\thmsubpara*

\paragraph*{A New Approach for Better Hardness Results}

Now, while \cref{maintheorem:sub:para} yields  \w-hardness, the corresponding proof gives
only very weak lower bounds ruling out algorithms that are significantly faster than
\(n^{o(\log k)}\). Further, to rule out algorithms up to \(n^{o(k)}\) with the same
approach, we would need a much stronger version of \cref{lem:sig:lower:bound}, which seems
to be difficult.
Even then, Yuster's approach seemingly limits one to obtain lower bounds not higher than
\(O(n^{k/2})\) for certain tournaments: the signature of a transitive tournament
of order \(k\) is \(\lfloor k/2\rfloor\), thus giving a natural barrier for Yuster's approach
(see~\cite[Lemma~2.5]{Yuster25}).%
\footnote{
    See \cref{remark:yuster:limits} for a detailed discussion.
}

Our much stronger lower bounds in \cref{maintheorem:sub} thus require a completely new approach.

\thmsub*

\begin{remark}
    \dglabel^{remark:sub}
    By setting $k = 4$ in \cref{maintheorem:sub}, we obtain that for each tournament
    $T$ with four vertices, $\tourProb(\{T\})$ is at least as hard
    as $\cliqueProb_3$. It is believed (see \emph{clique conjecture} in~\cite[Page 3]{Clique_counting})
    that this problem cannot be solved faster than in time $O(n^\omega)$.
    Thus, we obtain a matching lower bound to the \(O(n^\omega)\)-time algorithm for $\tourProb(\{T\})$ of
    Yuster~\cite[Theorem~1.4]{Yuster25}.
\end{remark}

Again, let us take a step back and reflect on our overall goal: we wish to show that
\(\tourProb\) is powerful enough to count (small) undirected cliques---crucially in
an \emph{arbitrary} host graph.
Using a not-too-difficult construction, given a tournament \(T\) of order \(k\), we may
turn a colored (with \(k\) colors) graph \(G\) into a colored \emph{directed} graph
\(G^T\) such that the number of colorful
copies of \(T\) in \(G^T\) is precisely the number of colorful cliques in \(G\).%
\footnote{Create a directed graph $G^T$ such that $(u, v) \in E(G^T)$ if and only if
    $\{u, v\}$ is an edge of $G$ and $(c(u), c(v)) \in E(T)$ (that is, the colored edge
    appears in $T$).
    Now, each colorful $k$-clique $G$ is in a one-to-one correspondence
to a colorful copy of $T$ in $G^T$, hence $\CFsubs{T}{G^T}$ counts cliques.}

Clearly, the issue with the construction is that \(G^T\) is not necessarily a tournament.
Hence, Yuster~\cite{Yuster25} adds the missing edges to arrive at a tournament \(\tourG{G}{T}\) such that
\begin{itemize}
    \item an edge \(\{u, v \}\) in \(G\) results in an edge in \(\tourG{G}{T}\) that has
        the \emph{same} orientation as the edge \(\{c(u), c(v)\}\) in \(T\) (where \(c\) is the
        coloring of \(G\)); and
    \item a non-edge \(\{u, v\}\) results in an edge in \(\tourG{G}{T}\) that has
        the \emph{opposite} orientation as the edge \(\{c(u), c(v)\}\) in \(T\) (where \(c\) is the
        coloring of \(G\)).
\end{itemize}
Also consider \cref{def:TourG,fig1:to}.

\begin{figure}[tp]
    \centering
    \includegraphics{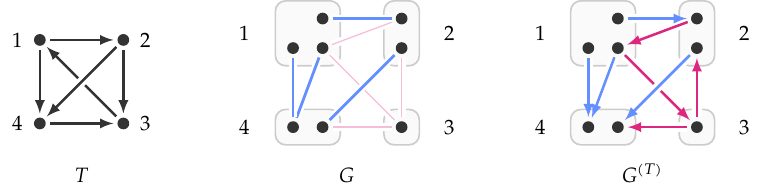}
    \caption{A tournament $T$, a graph $G$, and the tournament $\tourG{G}{T}$.
        We depict a subset of the (non-)edges of \(G\) and \(\tourG{G}{T}\), where
        blue edges are edges in $G$ and red edges are
        non-edges in $G$.}%
    \label{fig1:to}
\end{figure}

Taken for itself, this introduces a new problem---now counting induced colorful copies of \(T\)
in \(\tourG{G}{T}\) (henceforth \(\CFsubs{T}{\tourG{G}{T}}\))
also counts other colored graphs in $G$ that are not a clique; and
this is precisely the issue signatures help deal with---signatures allow to forcefully prevent
occurrences of non-\(T\) patterns.

We take a less forceful approach. We precisely understand which other patterns are
counted by \(\CFsubs{T}{\tourG{G}{T}}\): a linear combination of counts
of \emph{color-prescribed} occurrences%
\footnote{In contrast to \emph{colorful}, an occurrence is \emph{color-prescribed} if it
is colorful and the colors in the occurrence are distributed exactly as in the pattern.}
of patterns.

In itself, such a linear combination is not very useful---in fact, we initially obtain a
linear combination that contains the sought number of cliques (\cref{lem:tour:colindsub}), but non-trivial
cancellations still prevent us from using it directly.%
\footnote{One may view Yuster's approach~\cite{Yuster25} as a method to simplify the
resulting linear combination---that is, Yuster's approach~\cite{Yuster25} forces all
non-clique terms to become zero, thereby strictly simplifying the problem. Consult
\cref{7-15-3} for a detailed discussion.}

Interestingly, we observe that it is possible to rewrite the initial linear combination
into a different basis using a standard inclusion-exclusion argument
(\cref{lem:colindsub:colsub}) and, crucially, that the new basis is much more useful to
us.

Combining \cref{lem:tour:colindsub,lem:colindsub:colsub} we are able to
establish a connection to the \emph{alternating enumerator} (see
\cref{def:alternating:enumerator}) that is the main
ingredient in many recent breakthroughs concerning the complexity of counting
undirected induced subgraphs that satisfy a graph
property~\cite{DBLP:conf/soda/DoringMW25,DMW24,alge,CN24}.
In particular, we show that if the alternating enumerator $\ae{T}{H}$ of a
tournament \(T\) and a graph  \(H\) is non-zero, then our linear combination
contains a term that counts occurrences of \(H\).

\begin{restatable*}[{{{$\CFsubProb(\{T\})$ to
    $\CPsubProb$-basis}}}]{theorem}{sevenfifteenfour}
    \dglabel*{theo:tour:ae}[lem:tour:colindsub,lem:colindsub:colsub,def:alternating:enumerator,def:TourG]
    Let $T$ be a $k$-labeled tournament and $G$ be a $k$-colored graph. Then,
    \begin{equation}
        \label{7-15-5}
        \CFsubs{T}{\tourG{G}{T}} = \sum_{H \in \graphs{k}} \ae{T}{H} \cdot \CPsubs{H}{G}.
        \qedhere
    \end{equation}
\end{restatable*}

For our goal of proving \cref{maintheorem:sub} we need to work a little bit harder,
though.
First, we show that it is indeed possible to extract single terms of \cref{7-15-5} given
an oracle computing the value of the whole linear combination.

\begin{restatable*}[Complexity monotonicity of $\CPsubProb$-basis]{lemma}{sevenfifteensix}
    \dglabel{lem:colsub:monotone}[]
    Let $H_1, \ldots, H_m$ be a sequence of $m$ pairwise distinct $k$-labeled graphs. Let
    $\alpha_1, \dots, \alpha_m \in \mathbb{Q}$ be a sequence of coefficients with $\alpha_i \neq 0$ for
    all $i \in \setn{m}$ (we assume that we have access to the coefficients and graphs).

    Assume that
    there is an algorithm that computes for every $k$-colored graph $G$ of order $n$ the value
    \[f(G) = \sum_{i = 1}^m \alpha_i \cdot \CPsubs{H_i}{G}.\]
    Then for each $j \in \setn{m}$, there is an algorithm that computes $\CPsubProb(\{H_j\})$
    such that
    \begin{itemize}
        \item the algorithm calls $f(\star)$ at most $h(k)$ times for some computable function $h$,
        \item each call to $f(\star)$ is for a $k$-colored graph $G^\ast$ of order at most $n$, and
        \item each $k$-colored graph $G^\ast$ can be computed in time $O(n^2)$.
            \qedhere
    \end{itemize}
\end{restatable*}

While results similar to \cref{lem:colsub:monotone} have been obtained for linear
combinations of counts of \emph{graph homomorphisms} (see for
instance~\cite{DBLP:conf/stoc/CurticapeanDM17,rsw24})---to the best of our knowledge---this
is the first such result for color prescribed subgraph counts and might be of independent interest.

With complexity monotonicity at hand, the path seems to be clear: extract the term
corresponding to cliques from \cref{7-15-5} and complete the reduction.
Alas, understanding which terms are in fact present in the linear combination of
\cref{7-15-5} (that is, which graphs have a non-vanishing alternating enumerator) poses in
itself a non-trivial challenge---similar to the works that use the alternating enumerator
in the context of counting problems.

\begin{itemize}
    \item Unfortunately, the alternating enumerator of the clique is always
        zero, see \cref{lem:ae:clqiue}.
        In other words, our transformations successfully eliminated the target we were
        aiming at.
    \item Luckily, not all hope is lost, though. For our intended
        reduction we are fine with identifying non-vanishing terms in \cref{7-15-5} that
        correspond to graphs with large clique \emph{minors} (see \cref{cor:ae:reduction}).
        Standard tools then allow to reduce again to counting cliques---albeit at the cost
        of a slight running time overhead.
    \item In particular, we are able to show that for the \emph{anti-matching} \(\cmatch{k}\) (see
        \cref{def:match}) we have $\ae{T}{\cmatch{k}} \neq 0$ for any tournament \(T\)
        (see \cref{theo:cmatch:not:zero}).
        In other words, \(\cmatch{k}\) is always present in \cref{7-15-5}.
        Further, $\cmatch{k}$ contains $\lfloor 3k/4 \rfloor$-clique minor (see
        \cref{lem:cmatch:clique:minor}), which gives hardness.
\end{itemize}

In total, we obtain the following two results whose combination yields
\cref{maintheorem:sub}.

\begin{restatable*}[{$\tourProb(\{T\})$ is harder than
    $\CPsubProb(\{\cmatch{k}\})$}]{theorem}{sevenfifteentwo}\dglabel{maintheorem:cfsub}[lem:subs:to:color,cor:ae:reduction,
    theo:cmatch:not:zero]
    Fix a (pattern) tournament \(T\) of order \(k\) and assume that there is
    an algorithm that
    reads the whole input and
    computes $\tourProb(\{T\})$
    for any (host) tournament of order $n$ in time $O(n^\gamma)$.

    Then there is an
    algorithm that solves
    $\CPsubProb(\{\cmatch{k}\})$ for any $k$-colored graph of order $n$ in time $O(n^{\gamma})$.
\end{restatable*}

\begin{restatable*}[{$\CPsubProb(\{\cmatch{k}\})$ is hard}]{theorem}{sevenfifteenseven}
    \dglabel{lem:cmatch:to:clique}[lem:cmatch:clique:minor,lem:cpsub:minor,lem:cp:cf,lem:CFclique:clique]
    Fix $k \geq 1$ and assume that there is
    an algorithm that
    computes $\CPsubProb(\{\cmatch{k}\})$
    for any $k$-colored graph of order $n$ in time $O(n^\gamma)$. Then there is an algorithm that
    solves $\cliqueProb_{\lfloor 3k/4 \rfloor}$ for any graph of order $n$
    in time $O(n^{\gamma})$.
\end{restatable*}

\paragraph*{Surprising Additional Benefits of Our Approach}

While our new approach clearly has its merits with enabling a proof of
\cref{maintheorem:sub}, it turns out that is allows to unveil even more of the structure of
tournament counting.

As a first extra consequence, we improve Yuster's lower bound for counting transitive
tournaments~\cite{Yuster25} from \(n^{\lceil k/2 \rceil - o(1)}\) to \(n^{\lfloor 3k/4
\rfloor - o(1)}\).

\begin{corollaryq}
    Assume that there is
    an algorithm that
    computes $\tourProb(\{\tranTour_k\})$
    for any tournament of order $n$ in time $O(n^\gamma)$.  Then there is an algorithm that
    solves $\cliqueProb_{\lfloor 3k/4 \rfloor}$ for any graph of order $n$
    in time $O(n^{\gamma})$.
\end{corollaryq}

Next, we observe that our construction of \(\tourG{G}{T}\) is \emph{reversible}.
We are able to define a graph $\GraphT{G}{T}$ (see \cref{def:GraphT}) that allows us to
leverage the reversible nature of basis transformations to obtain efficient algorithms.

\begin{restatable*}[Efficient algorithms for \(\CFsubProb(\{T\})\) via the
    $\CPsubProb$-basis]{theorem}{sevenfifteeneight}
    \dglabel{theo:GraphT}[theo:tour:ae,def:GraphT]
    Given a $k$-labeled tournament $T$ and a $k$-colored tournament $G$ then
    \[\CFsubs{T}{G} = \sum_{H \in \graphs{k}} \ae{T}{H} \cdot \CPsubs{H}{\GraphT{G}{T}}.\]
    Further, assume that for
    each $H$ with $\ae{T}{H} \neq 0$ we have an algorithm that
    reads the whole input and
    computes $\CPsubProb(\{H\})$ in time $O(n^\gamma)$.
    Then there is an algorithm that computes $\CFsubProb(\{T\})$ in time $O(n^\gamma)$.
\end{restatable*}

As a direct corollary of \cref{cor:ae:reduction,theo:GraphT}, our new approach thus lets
us understand the complexity of $\CFsubProb$ \emph{precisely}.

\begin{restatable*}[Complexity of \(\CFsubProb(\{T\})\) is equal to
    hardest $\CPsubProb(\{H\})$ with $\ae{T}{H} \neq 0$]{theorem}{theocfsub}
    \dglabel{theo:cfsub:hardest:term}[cor:ae:reduction,theo:GraphT]
    Let $T$ be a $k$-labeled tournament then
    $\CFsubProb(\{T\})$ can be computed
    in time $O(n^\gamma)$
    if and only if for each $H$ with $\ae{T}{H} \neq 0$ the problem
    $\CPsubProb(\{H\})$ can be computed
    in time $O(n^\gamma)$.
\end{restatable*}

With \cref{theo:cfsub:hardest:term} at hand one may now wonder about the \emph{exact}
complexity of $\CFsubProb$.
By \cref{theo:cfsub:hardest:term}, we have to find a graph $H$
with $\ae{T}{H} \neq 0$ such that $\CPsubProb(\{H\})$ is as hard as possible.

We show that the anti-matching is in a sense optimal in this regard:
all supergraphs of $\cmatch{k}$ have an alternating enumerator that vanishes (again, in
particular the clique).

\begin{restatable*}[Anti-matchings are the densest graphs with $\ae{T}{H} \neq 0$]{theorem}{sevenfifteennine}
    \dglabel{cor:cmatch:densest}[lem:ae:apex]
    Let $T$ be a $k$-labeled tournament and $H$ be a $k$-labeled graph.
    If $|E(H)| > |E(\cmatch{k})|$, then $\ae{T}{H} = 0$.
\end{restatable*}

\Cref{cor:cmatch:densest} suggests that $\CPsubProb(\{\cmatch{k}\})$ is a
very good candidate to understand the true (fine-grained) complexity of \(\tourProb\)---our
non-tightness is in particular due to our somewhat crude applications of the standard
techniques to show hardness for $\CPsubProb(\{\cmatch{k}\})$.%
\footnote{See \cref{7-15-10} for a detailed discussion.}

\subsection{The Complexity of Finding Tournaments}

While we consider the results for the counting problems to be the  main technical
contribution of this work, our investigation of the corresponding decision problems yields
interesting insights that might be of independent interest.
Hence, we continue with a brief run-down of the techniques that we use to obtain
\cref{maintheorem:Dec:sub:para}.

\paragraph*{Efficiently-detectable Tournaments}

We start with a brief description of the proof of the algorithmic part of
\cref{maintheorem:Dec:sub:para}, which we encapsulate in \cref{theo:core:easy}.

\begin{restatable*}[\(\DecTourProb(\{T\})\) is easy for \(T\) of large spine length $\core{T}$]{theorem}{thmdecsubeasy}
    \dglabel{theo:core:easy}[def:core,thm:directed_ramsey]
    Fix a pattern tournament \(T\). There is an algorithm for $\DecTourProb(\{T\})$ that
    for host tournaments of order \(n\) runs in time $O(n^{|V(T)| - \core{T} + 2})$.

    Next, fix a class $\mathcal{T}$ of tournaments such that there is a
    constant $c$ with $|V(T)| - \core{T} \leq c$ for all $T \in \mathcal{T}$.
    There is an algorithm for $\DecTourProb(\mathcal{T})$
    for pattern tournaments of order \(k\) and host tournaments of order \(n\)
    that runs in time $O(f(k) \cdot n^{c+2})$ for some computable function $f$.
\end{restatable*}

Now, first observe that tournament patterns of constant size are detectable by
simple brute-force algorithms (whose running time depends on the complexity of the
pattern).
However, as we are interested in easily-detectable classes of \emph{infinitely many}
pattern graphs, constant-size patterns alone do not suffice.

Hence, let us also recall a famous result due to Erd\H{o}s and Moser~\cite{ErMo64}.
\begin{restatable*}[\cite{ErMo64}]{theoremq}{emdirto}
    \dglabel{thm:directed_ramsey}(Large tournaments contain large transitive
    subtournaments,~\cite{ErMo64})
    All tournaments of order $2^{k-1}$ contain a subtournament isomorphic to $\tranTour_k$.
\end{restatable*}
Observe that \cref{thm:directed_ramsey} yields a single class of easily-detectable
patterns---the transitive tournaments.

In our proof of \cref{theo:core:easy}, we combine (the tournaments of) both previous ideas to obtain classes of
infinitely many easily-detectable patterns, where the exact running time in
``easily-detectable'' depends on the structure of the graph class.
Formally, we choose an arbitrary but small tournament \(R_+ \cup R_-\) and a transitive
tournament \(S\), and connect by a forward edge every vertex of \(R_+\) with every vertex in
\(S\) and connect by a backward edge every vertex of \(R_-\) with every vertex in \(S\).
Phrased differently, we define a decomposition of a tournament into a transitive
tournament \emph{spine} \(S\) and some remaining \emph{ribs} \(R_+ \cup R_-\)---the formal definition follows;
also consult \cref{fig:neighborhood:decomposition} for a visualization of an example.

\begin{restatable*}[The spine decomposition of a tournament \(T\)]{definition}{defcore}
    \dglabel{def:core}
    For a tournament $T$ of order $k$, we say that  $(R_+, R_-, S)$ for $R_+ \uplus R_-
    \uplus S = V(T)$  is a \emph{spine decomposition of $T$} if
    \(\indGraph{T}{S}\) is a transitive tournament and
    \[S \coloneqq \left(\bigcap_{v \in R_+} N^{+}_T(v) \right) \; \cap \;
    \left(\bigcap_{v \in R_-} N^{-}_T(v) \right).\]
    We also call \(S\) the \emph{spine} of \((R_+, R_-, S)\), call \(R_+\) and \(R_-\)
    the \emph{ribs} of $(R_+, R_-, S)$, and say that the spine decomposition $(R_+, R_-, S)$
    has a spine length of \(|S|\).

    Further, we write $\core{T}$ for the largest spine length of any spine decomposition
    of \(T\).
\end{restatable*}

\begin{figure}[t]
    \centering
    \includegraphics[width=.45\textwidth]{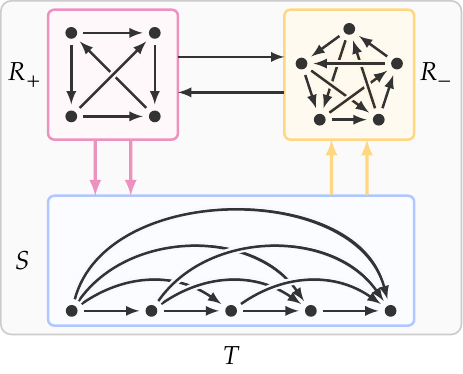}
    \caption{A spine decomposition of a tournament $T$.
        The spine  $S$ forms a transitive tournament. All vertices of the ribs $R_+$ have
        outgoing edges toward the spine and all vertices of the ribs $R_-$ have ingoing
        edges from the spine. Edges inside between ribs $R_+ \uplus R_-$ may be oriented arbitrarily.}
    \label{fig:neighborhood:decomposition}
\end{figure}

Now, when searching for a tournament pattern \(R_+ \cup R_-\cup S\) (we may compute such a
decomposition of the pattern in FPT time), we proceed in the natural
manner: given a host tournament \(G\), we try all possible choices for the (few) vertices
of \(R_+ \cup R_-\). For each such choice of vertices, we compute the corresponding
neighborhood of vertices that may potentially host a vertex of \(N\) (that is, we compute
a set \(N'\) of vertices of \(G\) analogous to the definition of \(N\) in
\cref{def:core}).

Now, if \(N'\) is small (smaller than $2^{|N|}$), we brute-force to find a copy of \(N\);
if \(N'\) is large (larger than $2^{|N|}$), we
apply \cref{thm:directed_ramsey}. In total, this yields \cref{theo:core:easy}.

\paragraph*{Not Efficiently-detectable Tournaments}

For the hardness part of \cref{maintheorem:Dec:sub:para}, we show that \emph{most}
tournaments are in fact hard to detect.

\begin{restatable*}[\(\DecTourProb(\{T\})\) is hard for random tournaments]{theorem}{thmdecsub}
    \dglabel{maintheorem:Dec:sub}[theo:TT:reduction,theo:prob:TT]
    Any tournament $T$ of order $k \geq 10^5$ that is chosen uniformly
    at random from all tournaments of order $k$ admits the following reduction with
    probability at least $1 - 3/k^3$.
    \begin{itemize}
        \item If there is an algorithm that reads the whole
        input and solves $\DecTourProb(\{T\})$ for any tournament of order $n$ in time $O(n^\gamma)$,
        then there exists an algorithm that solves
        $\DecCliqueProb_{\lfloor k/(9\log(k)) \rfloor}$ for any graph of order $n$
        in time $O(n^\gamma)$.
        \item Further, assuming ETH, there is a global constant $\beta > 0$ such that
        no algorithm that reads the whole input solves $\DecTourProb(\{T\})$
        for any graph of order $n$ in time $O(n^{\beta k/\log(k)})$.\qedhere
    \end{itemize}
\end{restatable*}

When counting all solutions, we may employ the colorful variant $\CFsubProb$ as an intermediate problem
to show the hardness of $\tourProb$---as does Yuster~\cite{Yuster25} for his hardness results.
However, such an approach is doomed in the decision case as there is no reduction from $\DecCFsubProb$
to $\DecTourProb$---see \cref{remark:no:reduction:dec:cf} for the technical details.

Hence, unable to use colors directly, we take an indirect route and aim to \emph{simulate}
colors via gadget constructions---since there are tournaments for which $\DecTourProb$ is easy,
our constructions work only for \emph{most} tournaments.
In particular, we show that the reduction of \cref{maintheorem:Dec:sub} works for all
tournaments that contain a specific substructure.
To this end, we write $\TT{T}$ for the largest transitive tournament inside $T$.

\begin{restatable*}[TT-unique]{definition}{TTunique}\dglabel{def:TT:unique}
    For a tournament $T$ of order $k$, we say that a partition of $V(T)$
    into $(\dSet, \zSet)$ is \emph{TT-unique with respect to $T$} if
    \begin{itemize}
        \item $\indGraph{T}{\dSet}$ has a trivial automorphism group,
        \item $\indGraph{T}{\dSet}$ appears exactly once in $T$ (that is, $\subs{\indGraph{T}{\dSet}}{T} = 1$), and
        \item for all $\dSet' \subseteq \dSet$ with $|\dSet'| \geq |\dSet| - \TT{T} \cdot |\zSet|$ and all
        $v \neq u \in V(T) \setminus \dSet'$, we have $\Tvec{v}{\dSet'} \neq \Tvec{u}{\dSet'}$. \qedhere
    \end{itemize}
\end{restatable*}

The intuition behind the first two properties of being TT-unique is that
$\indGraph{T}{\dSet}$ is \emph{distinguished in $T$}. This means that we can uniquely
identify $\indGraph{T}{\dSet}$ inside $T$.
The third property ensures that all vertices in $\zSet$ are distinguishable
by the intersection of their in-neighborhood with respect to $\dSet$.
The uniqueness of each vertex in $\zSet$ enables us to simulate colors.

\begin{restatable*}[Simulating colors via TT-uniqueness]{theorem}{ttReduction}
    \dglabel{theo:TT:construct}[def:TT:unique]
    Let $T$ be a tournament with a TT-unique partition  $(\dSet, \zSet)$ and let $\zCard \coloneqq |\zSet|$.
    Given a $\zCard$-colored graph $G$ of order $n$, we can construct an uncolored
    tournament $G^\ast$ of order $n + |\dSet|$ in time $O({(n + |\dSet|)}^2)$
    such that $T$ is isomorphic to a subtournament of $G^\ast$ if and only if $G$
    contains a colorful $\zCard$-clique.
\end{restatable*}

Assume that $T$ has a TT-unique partition  $(\dSet, \zSet)$ with $\zCard \coloneqq |\zSet|$. Given
a $\zCard$-colored graph $G$ of order $n$ with coloring $c$, let us briefly sketch how
to construct an uncolored tournament $G^\ast$ such that $T$ is isomorphic to a
subtournament of $G^\ast$ if and only if $G$ contains a colorful $\zCard$-clique.

To construct $G^\ast$, we start with the tournament $\tourG{G}{\indGraph{T}{\zSet}}$
(as defined previously, see also \cref{def:TourG}).
Note that $\tourG{G}{\indGraph{T}{\zSet}}$ simulate edges and non-edge in $G$
via the orientation of the edges in $\indGraph{T}{\zSet}$. Further, this tournament
is naturally $\zCard$-colored via the coloring $c$ of $G$. We also ensure that
vertices of the same color in $\tourG{G}{\indGraph{T}{\zSet}}$ always form
a transitive tournament.

To obtain $G^\ast$, we add the tournament $\indGraph{T}{\dSet}$ to $\tourG{G}{\indGraph{T}{\zSet}}$,
the orientation of the edges between $\indGraph{T}{\dSet}$ and $\tourG{G}{\indGraph{T}{\zSet}}$ being
carefully chosen. We give each vertex in $\dSet$ its own color, yielding a new coloring $c^\ast$ for $G^\ast$
(see \cref{fig:dec:reduction} for an example). Note that $c^\ast(x) = c(x)$ for all $x \in V(G)$.
Now, assume that $A \subseteq V(G^\ast)$ with $\indGraph{G^\ast}{A} \cong T$.
If we can ensure that the isomorphism from $\indGraph{G^\ast}{A}$ to $T$
is given via $c^\ast$ (i.e., $\indGraph{G^\ast}{A}$  is color prescribed
with respect to $c^\ast$), then by construction of $\tourG{G}{\indGraph{T}{\zSet}}$, we obtain that
$A \cap V(G)$ is a colorful $\zCard$-clique in $G$.
However, recall that we only have access to $\DecTourProb(\{T\})$
which completely ignores the coloring of $G^\ast$.
Therefore, we have to find a way to ensure that all subtournaments of $G^\ast$
that are isomorphic to $T$ are also always color prescribed.

This is the point where we use that $(\dSet, \zSet)$ is a TT-unique partition.
By the third property of TT-uniqueness, we obtain that the largest transitive subtournament
of $T$ (and therefore $\indGraph{G^\ast}{A}$)
is relatively small. Further, vertices of the same color
in $G^\ast$ always form a transitive tournament. Hence, in a first step,
we can ensure that not many vertices in $\indGraph{G^\ast}{A}$ have the same color, and
subsequently we can upper bound the size of $A \cap \zSet$.
This yields $|A \cap \dSet| \geq |\dSet| - \TT{T} \cdot |\zSet|$, which allows us to use
the third property of TT-uniqueness again. This time, we obtain
that the neighborhoods of two vertices in $V(T) \setminus (A \cap \dSet)$ differ by a lot.
These neighborhoods can then be used to simulate colors.
We combine this with the first two properties
of being TT-unique to directly control the location
of each vertex in $\indGraph{G^\ast}{A}$ to ensure that $\indGraph{G^\ast}{A}$
is color prescribed with respect to $c^\ast$.
This enables us to show that each isomorphic copy of $T$ in $G^\ast$
directly corresponds to a colorful $\zCard$-clique in $G$.

However, \cref{theo:TT:construct} works only if $T$ has a TT-unique
partition with a large set $\zSet$. By utilizing standard techniques from probability theory,
we show that this is the case for
random tournaments.

\begin{restatable*}[Random tournaments have TT-unique
    partition $(\dSet, \zSet)$ with large $|\zSet|$]{theorem}{theoprobTT}\dglabel{theo:prob:TT}[lem:prob:iso,lem:prob:vec]
    Let $T$ be a random tournament of order $k\geq 10^5$, then with probability
    at least $(1 - 3/k^3)$ it admits a TT-unique partition $(\dSet, \zSet)$
    with $|\zSet| \geq \lfloor k / (9 \log(k)) \rfloor$.
\end{restatable*}

Combining \cref{theo:prob:TT} with \cref{theo:TT:construct} yields
that $\DecTourProb(\{T\})$ is at least as hard as $\DecCliqueProb_{k/(9 \log(k))}$ for
almost all tournaments (see \cref{maintheorem:Dec:sub}).

Finally, we are ready to construct $\mathcal{T}_c$ from \cref{maintheorem:Dec:sub:para}.
To this end, we first use \cref{maintheorem:Dec:sub} to obtain
a tournament $T_0$ that is \emph{hard enough} (that is, a tournament that has a lower
bound under ETH: observe that for any specific tournament that admits a reduction,
\cref{maintheorem:Dec:sub} indeed yields a \emph{deterministic} reduction).
Next, we add transitive tournaments of varying sizes to $T_0$
(correspondingly to \cref{def:core}).
Since $T_k$ has a large spine by construction, \cref{theo:core:easy} yields
that $\DecSubProb(\mathcal{T}_c)$ is FPT.
In total, we obtain \cref{maintheorem:Dec:sub:para}.

\thmdecsubpara*

\clearpage
\section{Preliminaries}

We write $\log$ for the logarithm of base 2 and we write $a \equiv_p b$ for $a \equiv b
\bmod{p}$.

For a positive integer $k$, we write $\setn{k}$ for the set $\{1, \dots, k\}$.
Given a set $A$ and a nonnegative integer~$k$, we write
$\binom{A}{k} \coloneqq \{B \subseteq A : |B| = k\}$ for the set of all size-\(k\) subsets of \(A\).

For two sets $A$ and $B$, we write $\symdif{A}{B} \coloneqq (A \cup B) \setminus (A \cap B)$
for the symmetric difference of $A$ and $B$. Observe that an element is in $\symdif{A}{B}$
if it is either in $A$ or in $B$, but not in both sets.
Further, if \(A\) and \(B\) are disjoint, we also write $A \uplus B$ to denote their disjoint
union.

For two functions $f \colon A \to B$ and $g \colon B \to C$, we write $g \circ f \colon A
\to B$ for their concatenation, that is, for the function $(g \circ f)(x) \coloneqq
g(f(x))$.
We write $\sym{k}$ for the symmetric group on $\setn{k}$, that is, for the group of all
permutations of \(k\) elements with function composition.
Given two sets $A, B \subseteq \setn{k}$,
we say that a permutation $\sigma \in \sym{k}$ \emph{maps} $A$ to $B$
if $\sigma(x) \in B$ for all $x \in A$.

\paragraph*{Graphs}
In this paper, we use the term \emph{graph} for an undirected (labeled) graph without multi-edges and self-loops.
We write $\graphs{n}$ for the set of all graphs with vertex set $\setn{n}$.

Given a graph $G$, we write $V(G)$ for the vertex set of $G$ and $E(G)$ for the edge set of
$G$. The \emph{order} of a graph \(G\) is the number of vertices of \(G\).
Two vertices are \emph{adjacent} if they form an edge. An \emph{apex} is a vertex that is
adjacent to all other vertices. Observe that a graph may contain multiple apices.
A graph $G$ is a \emph{matching} if each vertex is adjacent to at most one other vertex.

We say that $H$ is a \emph{subgraph} of $G$ if $V(H) \subseteq V(G)$ and $E(H) \subseteq V(G)$.
Given a set of vertices $A \subseteq V(G)$, we write $\indGraph{G}{A}$ for the subgraph
of $G$ that is \emph{induced} by $A$, meaning that $V(\indGraph{G}{A}) = A$ and
$E(\indGraph{G}{A}) = E(G) \cap \binom{A}{2}$. We say that $H$ is an
\emph{edge-subgraph} of $G$ (denoted as $H \subseteq G$) if $V(H) = V(G)$ and $E(H) \subseteq E(G)$.
Given a set of edges $S \subseteq E(G)$, we write
$\edgesub{G}{S}$ for the edge-subgraph of $G$ that is induced by $S$. Formally,
$V(\edgesub{G}{S}) = V(G)$ and $E(\edgesub{G}{S}) = S$.
We say that $H$ is equal to $G$ (denoted as $H \equiv G$) if $V(H) = V(G)$ and $E(H) = E(G)$.

We write $K_n$ for the complete graph with vertex set $\setn{n}$.
A \emph{$k$-clique} of a graph is a set of $k$ vertices inducing a copy of $K_k$.
Further, we write $\IS_n$ for the graph without edges and vertex set $\setn{n}$.
For a graph $G$, we write $\overline{G}$ for the complement graph of $G$, that is, the
graph with vertex set \(V(G)\) and all edges that are missing from \(E(G)\) to turn \(G\)
into the complete graph.

A \emph{tree-decomposition} of a graph $G=(V,E)$ is a pair $(T,\mathcal{X})$ where $T=(I,F)$ is a tree,
and ${\mathcal{X}=(B_i)}_{i\in I}$ is a family of subsets of $V(G)$, called \emph{bags}
and indexed by the vertices of $T$, such that
\begin{enumerate}
    \item each vertex $v\in V$ appears in at least one bag, that is $\bigcup_{i\in I} B_i= V$,
    \item for each edge $e = \{x, y\} \in E$, there exists $i\in I$ such that $\{x,y\} \subseteq B_i$, and
    \item for each vertex $v\in V$, the set of nodes indexed by $\{ i : i\in I, v\in B_i\}$ forms a subtree of $T$.
\end{enumerate}
The \emph{width} of a tree decomposition is $\max_{i\in I} \{|B_i| -1\}$. The
\emph{treewidth} of $G$, denoted by $\tw (G)$, is the minimum width of a
tree-decomposition of $G$.

\paragraph*{Graph Homomorphisms, Isomorphisms, and Colorings}

A graph \emph{homomorphism} from $H$ to $G$ is a function $f \colon V(H) \to V(G)$ that
preserves edges (but not necessarily non-edges),
meaning that $\{f(u), f(v)\} \in E(G)$ for all $\{u, v\} \in E(H)$.
A graph \emph{isomorphism} from $H$ to $G$ is a function $f \colon V(H) \to V(G)$ that
preserves edges and non-edges, meaning that $\{f(u), f(v)\} \in E(G)$ if and only if $\{u, v\}
\in E(H)$---indeed, \(f\) is a bijection between \(H\) and \(G\).
If there is an isomorphism from
$H$ to $G$ then $H$ and $G$ are \emph{isomorphic}, which we denote by $H \cong G$.
An \emph{automorphism} of $H$ is an isomorphism from $H$ to itself. We write $\auts{H}$
for the set of automorphisms of $H$.

A graph $G$ is \emph{$k$-labeled} if $V(G) = \setn{k}$ and
\emph{$k$-colored} if \(G\) comes with a mapping $c \colon V(G) \to \setn{k}$.
Further, a subgraph $H$ of a $k$-colored graph $G$ is \emph{colorful with respect to $c$}
if the function $c$ the restricted to the vertices of $H$ is a bijection.
Observe that a $k$-labeled graph $G$ comes with a natural coloring $\id \colon \setn{k}
\to \setn{k}; x \mapsto x$.

Given a $k$-labeled graph $H$ and a $k$-colored graph $G$ with coloring $c$,
a subgraph $F$ of $G$ is \emph{color-prescribed} if
\(F\) is isomorphic to $H$ and respects the coloring $c$, that is,
$c$ restricted to $F$ defines an isomorphism to $H$ (or equivalently $\{u, v\} \in E(F)$ if and
only if $\{c(u), c(v)\} \in E(H)$).
Observe that a {color-prescribed} subgraph
is necessarily \emph{colorful}, but the opposite is not true in general.
Consult \cref{fig:cp:cf} for a visualization of an example.

\begin{figure}[tp]
    \renewcommand\tabularxcolumn[1]{m{#1}}
    \centering
    \begin{tabularx}{\linewidth}{*{2}{>{\centering\arraybackslash}X}}
    \includegraphics{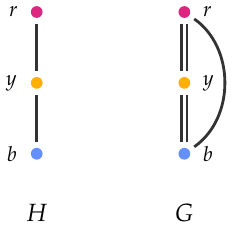}
    &
    \includegraphics{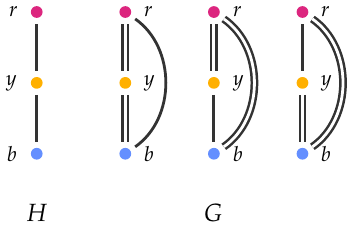}
    \\[-2ex]
    \begin{subfigure}[t]{\linewidth}
        \caption{There is a single color-prescribed copy of the 3-colored
            path \(H\) in the 3-colored triangle \(G\) (depicted by doubled lines).}\label{fig3a}
    \end{subfigure}
    &
    \begin{subfigure}[t]{\linewidth}
        \caption{There are three colorful copies of the 3-colored
            path \(H\) in the 3-colored triangle \(G\) (depicted by doubled lines in the
            three copies of \(G\)).}\label{fig3b}
    \end{subfigure}
    \end{tabularx}
    \caption{An illustration of the difference between color-prescribed subgraphs
        (\cref{fig3a}) and colorful subgraphs (\cref{fig3b}).}%
    \label{fig:cp:cf}
\end{figure}

\begin{definition}
    Let $H$ be a $k$-labeled graph and $G$ be a $k$-colored graph with coloring $c \colon
    V(G) \to \setn{k}$.
    \begin{itemize}
        \item We write $\GraphCFsubs{H}{G}$ for the number of
        subgraphs $F$ of $G$ that are isomorphic to $H$ and colorful with respect to $c$
        (that is, $c$ restricted to $V(F)$ is a bijection).

        \item We write $\CPsubs{H}{G}$ for the number of subgraphs $F$
        of $G$ that are isomorphic to $H$ and that respect the coloring $c$.

        \item We write $\CFsubs{H}{G}$ for the number of induced
        subgraphs $\indGraph{G}{A}$
        of $G$ that are isomorphic to $H$ and colorful with respect to $c$.

        \item We write $\CPindsubs{H}{G}$ for the number of induced
        subgraphs $\indGraph{G}{A}$ of $G$ that are isomorphic to $H$ and that respect the coloring $c$.

        \item For a positive integer $k$,
            we write $\clique{k}{G} \coloneqq\indsubs{K_k}{G}$ for the number of $k$-cliques in $G$
            and
            we write $\CFclique{k}{G} \coloneqq
            \CFsubs{K_k}{G} $ for the number of colorful $k$-cliques in $G$.
        \qedhere
    \end{itemize}
\end{definition}

\paragraph*{Tournaments}
A directed graph is a graph where each edge comes with a direction; we say a directed edge
(or \emph{arc}) \((u, v)\) goes from \(u\) to \(v\).
The \emph{out-neighborhood} $N^+_T(v)$ of a vertex \(v\) (in a directed graph \(T\)) is the set of all vertices \(w\) with a
directed edge from \(v\) to \(w\).
We write $\outDeg{T}{v} \coloneqq |N^+_T(v)|$ for the out-degree of $v$.
The \emph{in-neighborhood} $N^-_T(v)$ of a vertex \(v\) (in a directed graph \(T\))  is the set of all vertices \(u\) with a
directed edge from \(u\) to \(v\).
We write $\inDeg{T}{v} \coloneqq |N^-_T(v)|$ for the in-degree of $v$.

Given two directed graphs \(T\) and \(T'\), a function \(c \colon V(T) \to V(T')\)
and vertices \(u, v \in V\), we say
\begin{itemize}
    \item \(T\) and \(T'\) have the same orientation on \(\{u, v\}\) if
        both \((u, v) \in \DE{T}\) and \((c(u), c(v)) \in \DE{T'}\) or
        both \((v, u) \in \DE{T}\) and \((c(v), c(u)) \in \DE{T'}\).
        In this case, we also say that \(T\) and \(T'\) agree on \(\{u, v\}\).
    \item \(T\) and \(T'\) have the opposite orientation on \(\{u, v\}\) if
        both \((u, v) \in \DE{T}\) and \((c(v), c(u)) \in \DE{T'}\) or
        both \((v, u) \in \DE{T}\) and \((c(u), c(v)) \in \DE{T'}\).
        In this case, we also say that \(T\) and \(T'\) disagree on \(\{u, v\}\).
\end{itemize}
A noteworthy special case occurs for \(V(T) = V(T')\); in this case we typically assume
\(c\) to be the identity.
Finally, a \emph{flip} or \emph{change of orientation} of a directed edge \((u,v)\) is the operation that replaces
\((u,v)\) with \((v,u)\).

A \emph{tournament} $T$ is a directed graph with exactly one directed \emph{edge} between
any two vertices $u$ and $v$.
We write $V(T)$ for the vertex set of $T$ and $\DE{T}$ for the edge set of $T$.
The \emph{order} of a tournament \(T\) is the number of vertices of \(T\).

We say that $T'$ is a \emph{subtournament} of $T$ if $V(T') \subseteq V(T)$,
$\DE{T'} \subseteq \DE{T}$, and $T'$ is also a tournament.
Given a set of vertices $A \subseteq V(T)$, we write $\indGraph{T}{A}$ for the subtournament
of $T$ that is induced by $A$, meaning that $V(\indGraph{T}{A}) = A$ and
$\DE{\indGraph{T}{A}} \coloneqq \{(u, v) \in \DE{T} : u, v \in A\}$.
Observe that all subtournaments $T'$ of $T$ are induced (that is, $T' = \indGraph{T}{A}$ for some $A \subseteq V(T)$).
We say that $T'$ is equal to $T$ (denoted as $T' \equiv T$) if $V(T') = V(T)$ and $\DE{T'} = \DE{T}$.

A \emph{random tournament of order $k$} is a tournament drawn from the uniform
distribution over all $k$-labeled tournaments. Equivalently, a random tournament of order
$k$ can be obtained by orienting each edge of the complete graph $K_k$ independently, with
each direction chosen with probability $1/2$.

A tournament $T$ is {\it transitive} if it does not contain any directed cycle.
We write $\tranTour_n$ for the tournament with vertex set $\setn{n}$ and $(i, j) \in \DE{\tranTour_n}$
if and only if $i < j$. Observe that a tournament $T$ of order $n$ is transitive if and only
if it is isomorphic to $\tranTour_n$. We write $\TT{T}$ for the order of the
largest subtournament of $T$ that is transitive.
The {\it topological ordering} of a transitive tournament
$T$ is the unique ordering $u_1,\dots,u_n$ of $V(T)$ such that, for every $i<j$, $(u_i, u_j)\in \DE{T}$.

A tournament \emph{homomorphism} from $T$ to $T'$ is a function $f \colon V(T) \to V(T')$ that
preserves edges (but not necessarily non-edges), meaning that $(f(u), f(v)) \in \DE{T'}$
for all $(u, v) \in \DE{T}$.
A tournament \emph{isomorphism} from $T$ to $T'$ is a function $f \colon V(T) \to V(T')$ that
preserves edges and non-edges, meaning that $(f(u), f(v)) \in \DE{T'}$
if and only if $(u, v) \in \DE{T}$---indeed \(f\) is a bijection between \(T\) and \(T'\).
If there is an isomorphism between $T$ and $T'$ then \(T\) and \(T'\)
are \emph{isomorphic}, which we denote by $H \cong G$.
An \emph{automorphism} of $T$ is an isomorphism from $T$ to itself. We write $\auts{T}$
for the set of automorphisms of $T$.

A tournament $T$ is \emph{$k$-labeled} if $V(T) = \setn{k}$
and \emph{$k$-colored} if it comes with a mapping $c \colon V(T) \to \setn{k}$.
Further, a subtournament $T'$ of a
$k$-colored tournament $T$ is \emph{colorful with respect to $c$} if the function
$c$ restricted to the vertices of $T'$ is a bijection.
Observe that a $k$-labeled tournament $T$ comes with a natural coloring $\id \colon
\setn{k} \to \setn{k}; x \mapsto x$.

\begin{definition}
    \begin{itemize}
    \item Let $T$ be a directed graph and $G$ be tournament. We write $\tours{T}{G}$ for the number of
    induced subtournaments of $G$ that are isomorphic to $T$.
    \item Let $T$ be a $k$-labeled directed graph and $G$ be a $k$-colored
    tournament with coloring $c \colon V(G) \to \setn{k}$. We write $\CFsubs{T}{G}$
    for the number of induced subtournaments $F$
    of $G$ that are isomorphic to $T$ and colorful with respect to $c$
    (that is, $c$ restricted to $V(F)$ is a bijection).
    \qedhere
    \end{itemize}
\end{definition}

\paragraph*{Parameterized Complexity}

A \emph{parameterized (counting) problem} is a pair \((P, \kappa)\) of a function $P \colon \Sigma^\ast \to \nat$
and a computable parameterization $\kappa \colon \Sigma^\ast \to \nat$.
We say that \((P, \kappa)\) is a decision problem if \(P\) takes the values \(0\) or
\(1\).
A parameterized problem
$(P, \kappa)$ is \emph{fixed-parameter tractable} (FPT) if there is a computable function $f$
and a deterministic algorithm $\mathbb{A}$ that computes $P(x)$ in time $f(\kappa(x))\; |x|^{O(1)}$
for all $x \in \Sigma^\ast$.

A \emph{parameterized Turing reduction} from $(P, \kappa)$ to $(P', \kappa')$ is a deterministic
FPT algorithm that computes $P(x)$ using oracle access to $P'$ where each input \(y\) to
the oracle satisfies $\kappa'(y) \leq g(\kappa(y))$ for a computable function $g$.
We write $P \fpt P'$ whenever there is a parameterized
reduction from $P$ to $P'$. The relation \(\fpt\) is transitive~\cite[Theorem 13.3]{param_algo}.

Next, we introduce the main problems relevant for us.
To this end, we write
$\mathcal{T}$ for a recursively enumerable (r.e.) set of directed graphs and $\mathcal{H}$
for a r.e.\ set of undirected graphs.

\noindent
\begin{minipage}{1\linewidth}%
\begin{problem}-{\NUM{}\ensuremath{\textsc{Clique}}}%
    \label{prob:clique}
    \PInput{A pair of a graph \(G\) and a parameter \(k\).}
    \POutput{$\clique{k}{G}$; that is, the number of subsets $A \subseteq V(G)$ with $\indGraph{G}{A} \cong K_k$}
    \PParameter{\(k\)}
\end{problem}

\begin{problem}+-={\NUM{}\ensuremath{\textsc{cf}\text{-}\textsc{Clique}}}%
    \label{prob:cf:clique}
    \PInput{A pair of a $k$-colored graph \(G\) and a parameter \(k\).}
    \POutput{$\CFclique{k}{G}$; that is, the number of colorful subsets $A \subseteq V(G)$
    with $\indGraph{G}{A} \cong K_k$ }
    \PParameter{\(k\)}
\end{problem}

\begin{problem}+-={$\#\textsc{cf}\text{-}\textsc{IndSub}_{\TourSymb}$}%
    \label{prob:cf:IndSub}
    \PInput{A directed $k$-labeled graph $T \in \mathcal{T}$ and a $k$-colored tournament \(G\).}
    \POutput{$\CFsubs{T}{G}$; that is, the number of colorful subsets $A \subseteq V(G)$
    with $\indGraph{G}{A} \cong T$}
    \PParameter{\(k \coloneqq |V(T)|\)}
\end{problem}

\begin{problem}+={$\#\textsc{cp}\text{-}\textsc{Sub}$}%
    \label{prob:cp:sub}
    \PInput{A $k$-labeled graph $H \in \mathcal{H}$ and a $k$-colored graph \(G\) with
    coloring $c$.}
    \POutput{$\CPsubs{H}{G}$; that is, the number of subgraphs $F$ of $G$ that are
    isomorphic to $H$ via $c$}
    \PParameter{\(k \coloneqq |V(H)|\)}
\end{problem}
\end{minipage}

\begin{remark}
    We may safely assume that \(\mathcal{T}\) contains only tournaments, as
    for any tournament \(G\) and any non-tournament \(T\), we always have
    $\tours{T}{G} = \CFsubs{T}{G} = 0$.
\end{remark}

\begin{remark}
    We also use decision problem variants of the above problems, where we identify all
    nonzero output values.
    We write \emph{Dec} before a counting problem to denote the decision variant of the
    problem. For example, for a $k$-labeled graph $H$ the problem $\DecCPsubProb(\{H\})$ gets as
    input a $k$-colored graph $G$ and is asked to decide if
    $G$ has a subgraph that is isomorphic to $H$ and respects the coloring $c$ of $G$.
\end{remark}

A counting problem $(P,\kappa)$ is \w-hard if there is parameterized Turing reduction
from $\cliqueProb$ to $(P, \kappa)$.
A decision problem $(P,\kappa)$ is \Dw-hard if there is parameterized Turing reduction
from $\DecCliqueProb$ to $(P, \kappa)$.

\paragraph*{The Complexity Exponent and Fine-grained Complexity}

We are also interested in fine-grained lower and upper bounds, that is in the precise
exponent in the polynomial terms of running times.
We always specify the complexity of our graph problems with respect
to the number of vertices of the input graph and use the following shorthand notation.

\begin{definition}[Complexity exponent  $\Exp{P}$ ]\label{def:Exp}
    Given a counting/decision problem $P$ on graphs/tournaments, we write
    $\Exp{P}$ for the infimum over all $\beta$ such that there exists an
    algorithm that reads the whole input  and solves $P$ for all
    input graphs/tournaments $G$ in time $O(|V(G)|^\beta)$.%
    \footnote{
        Observe that $\Exp{P}$ is always at least $2$ since we consider only
        algorithms that read the whole input.
        Other authors (for instance~\cite{Yuster25}) use $c^\ast(k)$ for $\cExp{k}$,
        $c(T)$ for $\tExp{T}$, $d^\ast(k)$ for $\DeccExp{k}$, or $d(T)$ for $\DectExp{T}$.
    }
\end{definition}

\begin{remark}
    For technical reasons, we require that algorithms have to read the whole input.
    For the sake of readability, we henceforth omit this assumption from our statements.
\end{remark}

For fine-grained lower bounds, we rely
on the \emph{Exponential Time Hypothesis}
(ETH) due to~\cite{IMPAGLIAZZO2001512}, as formulated in~\cite{param_algo}.

\begin{conjecture}[ETH, {\cite[Conjecture 14.1]{param_algo}}]
    There is a real value \(\varepsilon > 0\) such that
    for all $c > 0$
    the problem $\sat{3}$ cannot be solved in time
    $O(2^{\varepsilon n} \cdot n^c)$,  where $n$ is the number of
    variables of the formula.
\end{conjecture}

The Exponential Time Hypothesis implies
that there are no FPT algorithms for \Dw-hard or \w-hard problems.

\begin{lemmaq}[{Clique lower bounds under ETH,~\cite[Modification of Lemma B.2]{DMW24_arxiv}}]\label{lemma:eth:clique}
    Assuming ETH, there exists a global constant $\alpha > 0$ such that for $k \geq 3$:
    \begin{itemize}
        \item No algorithm solves $\DecCliqueProb_{k}$
        for graphs $G$ of order $n$ in time $O(n^{\alpha k})$.
        Especially, we can choose $\alpha$ such that $\DeccExp{k} > \alpha k$.
        \item No algorithm computes $\cliqueProb_{k}$
        for graphs $G$ of order $n$ in time $O(n^{\alpha k})$.
        Especially, we can choose $\alpha$ such that $\cExp{k} > \alpha k$.
        \qedhere
    \end{itemize}
\end{lemmaq}

\section{Hardness of Counting Tournaments via Signatures}

In this section, we prove our results for the counting problem $\tourProb$; in
particular \cref{maintheorem:sub:para}.
As a first step of our reduction, we move to the colored version $\CFsubProb$.
To this end, we rely on a result by Yuster~\cite{Yuster25}; for completeness, we give the
proof in the appendix.

\begin{restatable*}[{$\tourProb(\{T\})$ is harder than $\CFsubProb(\{T\})$~\cite[Lemma 2.4]{Yuster25}}]{lemma}{lemsubscolor}\dglabel{lem:subs:to:color}
    For a $k$-labeled tournament $T$,
    assume that there is an algorithm that computes $\tourProb(\{T\})$
    for any tournament of order $n$ in time $O(f(n))$.
    Then there is an algorithm that computes $\CFsubProb(\{T\})$ for any $k$-colored tournament
    of order $n$ in time $O(2^{|V(T)|} \cdot f(n))$.
    In particular, $\Exp{\tourProb(\{T\})} \geq \Exp{\CFsubProb(\{T\})}$.

    Further, for an r.e.\ set of tournaments $\mathcal{T}$,
    we obtain $\CFsubProb(\mathcal{T}) \fpt \tourProb(\mathcal{T})$.
\end{restatable*}

First, we show how to obtain hardness results
via a structure of tournaments called \emph{signature} that originated in Yuster's work~\cite{Yuster25}.

\begin{definition}[{{Signature of a tournament, $\sig(T)$,~\cite[Definition 2.2]{Yuster25}}}]\dglabel{def:sig}
    Let \(T\) be a tournament.
    An \emph{edge flip of \(T\) touching \(R \subseteq V(T)\)} is a tournament that is
    obtained from \(T\) by flipping the orientation of at least one edge
    with both endpoints in $V(T)\setminus R$ but not changing the orientation of any edge with
    an endpoint in \(R\).

    A subset \(R \subseteq V(T)\) is a \emph{signature of \(T\)} if no edge flip of \(T\) touching
    \(R\) is isomorphic to \(T\).
    We write $\sig(T)$ for the smallest size of a signature of $T$
    (clearly, the entire vertex set of $T$ is a signature).
\end{definition}

Let $T$ denote a tournament of order $k$ and let $R$ be a signature of $T$ with $r \coloneqq |R|$.
Yuster uses signatures to obtain a reduction
from $\DecCFcliqueProb_{k-r}$ to $\DecCFsubProb(\{T\})$.
For completeness, we repeat the proof of Yuster with updated notation in the appendix.

\begin{restatable*}[{$\CFsubProb$ is harder than $\CFcliqueProb$ for tournaments with
        small signatures,~\cite[Modification of Lemma~2.5]{Yuster25}}]{lemma}{yusterred}\dglabel{lem:CFsubs:CFclique}[def:sig]
    Let $T$ be a tournament with $k$ vertices and $R$ be a signature of $T$ with $|R| = r$.
    Given a $(k-r)$-colored graph $G$ of order $n$,
    we can compute a tournament $G^\ast$ of order $(n + r)$ in time $O({(n + r)}^2)$ such that
    \[\CFsubs{T}{G^\ast} = \CFclique{k-r}{G}.\qedhere\]
\end{restatable*}

Now, to show hardness, we have to ensure that $|V(T)| - \sig(T)$ grows with the order of~$T$.
To this end, we identify a large transitive tournament inside of
$T$, which we are always able to do thanks to the result of Erd\H{o}s and Moser, that we
first recall here for convenience.

\emdirto

\begin{lemma}
    \dglabel*{lem:sig:lower:bound}[thm:directed_ramsey]($|V(T)| - \sig(T) \!\geq\! \log(|V(T)|) / 4$)
    Any tournament \(T\) of order $k$ satisfies $k - \sig(T) \!\geq\! \lceil \log(k)/4 \rceil$.
\end{lemma}
\begin{proof}
    For an integer $d$, we write $\CoutDeg{T}{d}$ for the number of vertices
    in $T$ whose out-degree is exactly $d$.

    Let $T$ be a tournament of order $k \geq p 2^{p-1}$.
    We show that $T$ contains a signature of size at most $k - p$.

    \begin{claim}
        There is a set $X \subseteq V(T)$ of $2^{p-1}$ vertices such that for every pair
        $u, v \in X$, we have either
        $\outDeg{T}{u} = \outDeg{T}{v}$ or
        $|\outDeg{T}{u} - \outDeg{T}{v}| \geq p$.
    \end{claim}
    \begin{claimproof}
        For all $i \in \{0, 1, \dots, p-1\}$,
        we define the set $X_i$ as the set of vertices $u$ such that $\outDeg{T}{u} \equiv_p i$.
        Observe that each set  $X_i$ has the desired property.
        Since ${(X_i)}_{i \in \{0, 1, \dots, p-1\}}$
        partitions $V(T)$, it follows that some $X_i$ has size at least $|V(T)|/p = k/p \geq 2^{p-1}$. This
        shows the existence of $X$.
    \end{claimproof}

    By \cref{thm:directed_ramsey}, since $|X|= 2^{p-1}$, there exists
    $Y \subseteq X$ of size $p$, such that $\indGraph{T}{Y}$ is isomorphic to $\tranTour_p$.
    We finish the proof with the following claim.

    \begin{claim}%
        \label{7-15-1}
        The set $R = V(T) \setminus Y$ is a signature of $T$ of size $k - p$.
    \end{claim}
    \begin{claimproof}
        Let us fix a tournament $T^\ast$ that is obtained from $T$ by changing the
        orientation of at least one edge in $\indGraph{T}{Y}$. We show that $T^\ast$ is not
        isomorphic to $T$, which implies the result.

        Let $B \coloneqq \DE{T}\setminus \DE{T^\ast}$ be the set of edges whose orientation changes
        and $y_1,\dots,y_p$ be the topological ordering of $\indGraph{T}{Y}$.
        Among all vertices of $Y$ incident to an edge in $B$, we let $v$ be the
        vertex with the smallest index, and set $d \coloneqq \outDeg{T}{v}$.
        Our goal is to show that
        $\CoutDeg{T}{d} \neq \CoutDeg{T^\ast}{d}$, which immediately yields that
        $T$ and $T^\ast$ are not isomorphic.

        By construction, in $T$ the vertex $v$ has no incoming edge in $B$, but
        at least one outgoing edge. Therefore, $\outDeg{T^\ast}{v} < \outDeg{T}{v}$.
        This means that $\CoutDeg{T}{d} = \CoutDeg{T^\ast}{d}$ could only be true
        if there exists a vertex $u \in Y$ with $\outDeg{T}{u} \neq \outDeg{T^\ast}{u} = d$.
        Let us show that such a vertex $u$ does not exist. Since $Y \subseteq X$,
        and because $\outDeg{T}{u} \neq d$, we obtain $|\outDeg{T}{u} - d| \geq p$.

        Now, since $|Y|=p$, the vertex $u$ is incident to at most $p-1$ edges in $B$.
        Thus, we have\[
            \outDeg{T}{u} - (p-1) \leq \outDeg{T^\ast}{u}  \leq \outDeg{T}{u} + (p-1),
        \]
        which implies that $\outDeg{T^\ast}{u} \neq d$. This in turn shows
        that $\CoutDeg{T}{d} \neq \CoutDeg{T^\ast}{d}$, and in particular
        $T$ is non-isomorphic to $T^\ast$.

        Thus, $R$ is a signature of size $k - p$.
    \end{claimproof}

    From \cref{7-15-1}, we obtain $k - \sig(T) \geq p$.
    Finally, we show how to pick \(
        p \geq \log(k)/4 + 1 \geq \left\lceil \log(k)/4 \right\rceil.
    \)
    Rearranging the terms, we obtain
    \[k \geq  \left(\frac{\log(k)}{4} +1\right)\cdot 2^{\log(k)/4}
    \; \Leftrightarrow \; k^{3/4} \geq \log(k^{1/4}) + 1,\]
    which is true for all $k \geq 1$.
    In total, we obtain that \(R\) is a signature of size at most \(k - \lceil \log(k)/4
    \rceil \).
\end{proof}

We are ready to prove \cref{maintheorem:sub:para}.
In contrast to~\cite{Yuster25}, we use $\CFcliqueProb$ as our source
of hardness. This allows us to strengthen and simplify the overall proof.

\begin{lemmaq}[{\cite[Lemma 1.11]{DBLP:phd/dnb/Curticapean15}}]
    \dglabel{lem:CFclique:w1}($\CFcliqueProb$ is \w-hard,~\cite[Lemma
    1.11]{DBLP:phd/dnb/Curticapean15})
    $\CFcliqueProb$ is \w-hard.%
    \footnote{The problem $\NUM{}\textsc{ColClique} / k$
    from~\cite{DBLP:phd/dnb/Curticapean15}  is equivalent to $\CFcliqueProb$.}
\end{lemmaq}

\thmsubpara
\begin{proof}
    Assume that $\mathcal{T}$ contains finitely many tournaments. Then, there is a $k$ such that
    $|V(T)| \leq k$ for all tournaments $T \in \mathcal{T}$. If $T \in \mathcal{T}$ is not a tournament
    then $\tours{T}{G} = 0$ for all input tournaments $G$. Otherwise,
    we compute $\tours{T}{G}$ in time $O(k^2 \cdot |V(G)|^k)$ by using a brute force algorithm.

    Next, assume that $\mathcal{T}$ contains infinitely many tournaments.
    From \cref{lem:subs:to:color}, we have $\CFsubProb(\mathcal{T}) \fpt \tourProb(\mathcal{T})$.
    Hence, it is enough to show that $\CFsubProb(\mathcal{T})$ is \w-hard.
    We present an FPT-reduction to $\CFcliqueProb$ which is \w-hard due to \cref{lem:CFclique:w1}.

    Let $G$ be a $p$-colored graph of order $n$.
    We start by searching for a tournament $T \in \mathcal{T}$ with $|V(T)| \geq 2^{4 p}$
    and a signature $R \subseteq V(T)$ (by \cref{lem:sig:lower:bound})
    with $|V(T)| - |R| = p$.
    Next, we use \cref{lem:CFsubs:CFclique}
    on \(T\), \(R\) and \(G\) to compute a tournament $G^\ast$.
    Finally we use our oracle for \(\CFsubProb\) to compute and return the number
    \(\CFsubs{T}{G^\ast}\).

    For the correctness, first observe that $\mathcal{T}$ contains tournaments
    of arbitrary large order.
    Now, due to \cref{lem:sig:lower:bound}
    every tournament $T$ with $|V(T)| \geq 2^{4 p}$ contains a signature $R$
    with $|V(T)| - |R| = p$.
    To this end, we observe $|V(T)| - \sig(T) \geq \log(2^{4p}) / 4
    = p$.
    Finally, \cref{lem:CFsubs:CFclique} returns a tournament \(G^\ast\)
    such that $\CFsubs{T}{G^\ast} = \CFclique{|V(T)| - |R|}{G} = \CFclique{p}{G}$;
    yielding correctness.

    For the running time, observe that there are computable function $g'$ and $g''$
    such that we can find $T$ and $R$ in time $O(g'(p))$
    and the parameter $|V(T)|$ of $\CFsubProb(\mathcal{T})$ is at most $O(g''(p))$.
    Further, \cref{lem:CFsubs:CFclique} runs in time $O(n^2)$. Hence,
    $\CFcliqueProb \fpt \CFsubProb(\mathcal{T})$ and therefore
    we obtain that $\CFsubProb(\mathcal{T})$ and $\tourProb(\mathcal{T})$ are both \w-hard.
\end{proof}

\section{Fine-grained Hardness of Counting Tournaments via Complexity Monotonicity}

In this section, we proceed to complement and supersede \cref{maintheorem:sub:para} with
stronger, fine-grained lower bounds.

\thmsub*

We proceed in two main steps. We first show that for understanding \(\tourProb(\{T\})\), it
suffices to understand the complexity of counting colored undirected \emph{anti-matchings}. In a
second step, we show that counting colored anti-matchings has tight lower bounds under ETH.

\begin{definition}[The anti-matching $\cmatch{k}$ of size \(k\)]\dglabel{def:match}
    For any $k$, we write $\match{k}$ for the canonical matching on \(k\) vertices, that
    is the graph with vertex set $\setn{k}$ and edge set
    \begin{itemize}
        \item $\{\{1, 2\}, \{3, 4\}, \dots, \{k-1, k\}\}$ if $k$ is even, or
        \item $\{\{1, 2\}, \{3, 4\}, \dots, \{k-2, k-1\}\}$ if $k$ is odd.
    \end{itemize}
    The \emph{anti-matching} $\cmatch{k}$ is the complement graph of $\match{k}$.
\end{definition}

\subsection{Counting Undirected, Colored Anti-Matchings via Directed Tournaments}
\label{sec:lin:combi}

As the first major step toward \cref{maintheorem:sub}, we show the following reduction.

\sevenfifteentwo*

Let $T$ be a $k$-labeled tournament.
We prove \cref{maintheorem:cfsub} by using the following chain of reductions:
\begin{align*}
    \tourProb(\{T\}) &\xrightarrow{\text{\ref{lem:subs:to:color}}}
    \CFsubProb(\{T\}) \xrightarrow{\text{\ref{theo:tour:ae}}}
    \sum_{H \in \graphs{k}} \ae{T}{H} \cdot \CPsubProb(\{H\})
    \xrightarrow{\text{\ref{cor:ae:reduction} and~\ref{theo:cmatch:not:zero}}}
    \CPsubProb(\{\cmatch{k}\}).
\end{align*}

\subsubsection*{Step 1a: Removing colors}

As mentioned in the last section, we use a result by Yuster~\cite{Yuster25}
to reduce from the uncolored version of counting tournaments to the colored version of the
problem.

\lemsubscolor*

\subsubsection*{Step 1b: Understanding $\CFsubProb(\{T\})$ via Linear Combinations}\label{sec:lin}

Since the input of $\CPsubProb(\{\cmatch{k}\})$
is an undirected $k$-colored graph $G$, we need to find
a way to transform $G$ into a tournament $\tourG{G}{T}$ for some $k$-labeled
tournament $T$.
To this end, we use the following construction that uses
the orientation of edges in $T$ to simulate edges and non-edges
of $G$.

\begin{definition}[The biased tournament $\tourG{G}{T}$ of a labeled tournament \(T\) and a
    colored graph \(G\)]\dglabel{def:TourG}
    Let $T$ be a $k$-labeled tournament and $G$ be a $k$-colored graph with coloring
    $c \colon V(G) \to \setn{k}$.

    The \emph{biased tournament} $\tourG{G}{T}$ of \(G\) and \(T\) is a $k$-colored
    tournament with vertex-set $V(G)$ and coloring $c$, such that for
    every $x,y\in V(G)$ with $c(x)\neq c(y)$, we have:
    \begin{itemize}
        \item if $\{x, y\} \in E(G)$ then $(x,y)\in \DE{\tourG{G}{T}}$ if and only if
            $(c(x),c(y))\in \DE{T}$ (that is, $\tourG{G}{T}$ and $T$ have the same orientation
            on $\{x, y\}$); and
        \item if $\{x, y\} \notin E(G)$ then $(x,y)\in \DE{\tourG{G}{T}}$ if and only if
            $(c(y),c(x))\in \DE{T}$ (that is, $\tourG{G}{T}$ and $T$ have a different
            orientation on $\{x, y\}$).
    \end{itemize}
    If $c(u) = c(v)$ the orientation of $\{u, v\}$ in $\tourG{G}{T}$ is arbitrary.\footnote{
        In this paper, we only use biased tournament when counting colorful tournaments.
        Hence, the orientation between vertices of the same color does not matter.
    }
\end{definition}

See \cref{fig:TourG} for an example of $\tourG{G}{T}$.
Our first hope is that $\CFsubs{T}{\tourG{G}{T}}$ is equal to the number of colorful $k$-cliques
in $G$. This would immediately
yield that $\CFsubProb(\{T\})$ is equal to the problem of counting colorful $k$-cliques which is a hard
problem in its own right. Indeed, $\CFsubs{T}{\tourG{G}{T}}$ counts colorful $k$-cliques
in $G$.
However, as the following example shows, $\CFsubs{T}{\tourG{G}{T}}$
also counts occurrences of other $k$-vertex graphs in $G$.

\begin{figure}[t]
    \centering
    \includegraphics{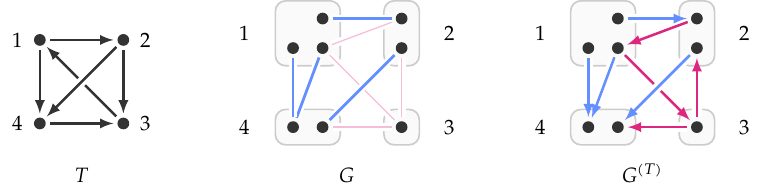}
    \vskip-1ex
    \caption{A tournament $T$, a graph $G$, and their corresponding
        biased tournament  $\tourG{G}{T}$.
        We depict a subset of the (non-)edges of \(G\) and \(\tourG{G}{T}\), where
        blue edges are edges in $G$ and red edges are
        non-edges in $G$.}%
    \label{fig:TourG}
\end{figure}

\begin{example}\label{ex:tri}
    The \emph{triangle tournament} $\tri$ is the tournament with vertex set $\setn{3}$ and edge set
    $\{(1, 2), (2, 3), (3, 1)\}$ (see \cref{fig:flip}). We claim that for any subset $A \subseteq V(G)$ the subtournament
    $\indGraph{\tourG{G}{\tri}}{A}$ is colorful
    and isomorphic to $\tri$ if and only if $\indGraph{G}{A}$ is either a colorful clique
    of a colorful independent set that respects $c$.

    To this end consider a pair of vertices $u, v \in A$. By construction of $\tourG{G}{\tri}$,
    we obtain that
    $\{u, v\}$ is an edge in $\indGraph{G}{A}$ if and only if $T$ and $\indGraph{\tourG{G}{\tri}}{A}$
    have the same orientation on $\{c(u), c(v)\}$ and $\{u, v\}$. Therefore, all
    non-edges of $\indGraph{G}{A}$ correspond to \emph{flipped}
    edges in $\indGraph{\tourG{G}{\tri}}{A}$. Thus, $\indGraph{\tourG{G}{\tri}}{A}$
    is isomorphic to $\tri$ if and only if $\indGraph{\tourG{G}{\tri}}{A}$ is equal to
    a tournament $\tri'$ that is isomorphic to $\tri$ and obtained from $\tri$ by flipping
    all non-edges of $\indGraph{G}{A}$.
    \Cref{fig:flip} shows that $\tri' \cong \tri$ can only be the case if we either
    flip no edges in $\tri$ (in this case $\indGraph{G}{A} \cong K_3$) or if we flip all edges in $\tri$
    (in this case $\indGraph{G}{A} \cong \IS_3$).
    Hence,
    \begin{align*}
        \CFsubs{\tri}{\tourG{G}{\tri}} = \CPindsubs{\IS_3}{G} + \CPindsubs{K_3}{G}.
        \tag*{\qedhere}
    \end{align*}
\end{example}

\begin{figure}[tp]
    \centering
    \includegraphics{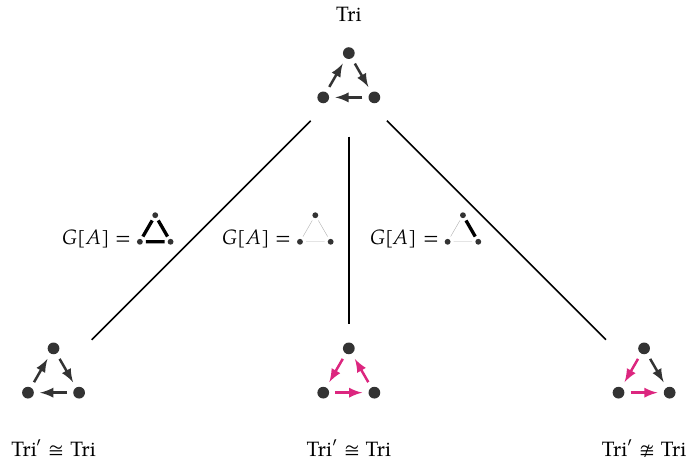}
    \caption{The tournament $\tri$ and the construction of $\tri'$
        depending on the choice of $\indGraph{G}{A}$. For readability, we omit
        the labels of the vertices.}%
    \label{fig:flip}
\end{figure}

We extend \cref{ex:tri} to general tournaments $T$. To this end, we introduce the
notion of \emph{flipping edges of $T$} with respect to some graph $H$ on the same vertex set as $T$.

\begin{definition}[$T_H$, the tournament obtained by flipping edges of a tournament \(T\) along a graph \(H\)]\dglabel{def:flip:graph}
    Let $T$ be a tournament and $H$ be a graph on the same vertex-set as $T$,
    we write $T_{H}$ for the tournament that we obtain from $T$ by flipping
    all edges $(u, v) \in \DE{T}$ with $\{u,v\}\notin E(H)$.
\end{definition}

We use  \cref{def:flip:graph} to write $\CFsubs{T}{\tourG{G}{T}}$
as a linear combination of $\CPindsubProb$-counts.

\begin{lemma}[{Expressing $\CFsubProb(\{T\})$ in the $\CPindsubProb$-basis}]\dglabel{lem:tour:colindsub}[def:flip:graph]
    Let $T$ be a $k$-labeled tournament $T$ and $G$ be a $k$-colored graph. Then
    \begin{align*}
        \CFsubs{T}{\tourG{G}{T}} = \sum_{H \in \graphs{k}} \left[T_{H} \cong T\right] \cdot \CPindsubs{H}{G},
    \end{align*}
    where $\left[T_{H} \cong T\right]$ is equal to 1 if $T_H \cong T$ and 0, otherwise.
\end{lemma}
\begin{proof}
    Let $c \colon V(G)  \to \setn{k}$ be the coloring of $G$ and $\tourG{G}{T}$.
    We show that $X \subseteq V(\tourG{G}{T}) = V(G)$ induces a colorful subtournament
    $\tourG{G}{T}[X]$ that is isomorphic to $T$ if and only if
    $\indGraph{G}{X}$ is isomorphic (with respect to $c$) to a graph $H_X$ with $T \cong T_{H_X}$.
    This proves the claim since $\CFsubs{T}{\tourG{G}{T}}$ counts the number
    of induced, colorful subtournaments of $\tourG{G}{T}$ that are isomorphic to $T$, while the sum
    counts the number of induced subgraphs of $G$ that are isomorphic (with respect to $c$)
    to some $H$ with $T \cong T_{H}$.

    To this end, let $X \subseteq V(\tourG{G}{T}) = V(G)$ be a colorful set of vertices.
    Recall that, since $X$ is colorful, $c$ is a bijection when restricted to $X$.
    For every such set $X$, we define $H_X$ as the graph with vertex-set $V(H_X)=\setn{k}$ and edge-set
    \[
        E(H_X)=\left\{\{u,v\} : (u,v)\in \DE{T} \mbox{~and~} (c^{-1}(u),c^{-1}(v))\in \DE{\tourG{G}{T}[X]}\right\}.
    \]

    \begin{claim}\label{claim:tour:H:iso}
        The function $x \mapsto c(x)$ defines an isomorphism from $\tourG{G}{T}[X]$ to $T_{H_X}$.
    \end{claim}
    \begin{claimproof}
        Let us fix two distinct vertices $x,y\in X$ and assume without loss of generality that $(x,y)$
        is an edge of $\tourG{G}{T}[X]$. We distinguish two cases.
        \begin{itemize}
            \item If $(c(x), c(y)) \in \DE{T}$, then by definition of $H_X$ we have
                $\{c(x), c(y)\} \in E(H_X)$, and by definition of $T_{H_X}$, we have
                $(c(x),c(y))\in \DE{T_{H_X}}$.
            \item Else, $(c(y), c(x)) \in \DE{T}$, then by definition of $H_X$ we have
                $\{c(x), c(y)\} \notin E(H_X)$ and by definition of $T_{H_X}$ we have
                $(c(x),c(y)) \in \DE{T_{H_X}}$.
        \end{itemize}
        In both cases, $(c(x),c(y))$ is an edge of $T_{H_X}$.
        The claim follows.
    \end{claimproof}
    Next, we show that $H_X$ and $\indGraph{G}{X}$ are isomorphic with respect to $c$.
    \begin{claim}\label{claim:H:G:iso}
        The function $f \colon \setn{k} \to X, u \mapsto c^{-1}(u)$ defines an isomorphism from $H_X$ to $\indGraph{G}{X}$.
    \end{claim}
    \begin{claimproof}
        Let us fix two distinct vertices $u,v\in \setn{k}$, and let $x_u = c^{-1}(u)$, $x_v=c^{-1}(v)$. We distinguish two cases.
        \begin{itemize}
            \item  Assume first that $\{x_u, x_v\}$ is an edge of $\indGraph{G}{X}$. Then,
                since $c(x_u)\neq c(x_v)$ and by definition of $\tourG{G}{T}$, $(u,v)\in \DE{T}$
                if and only if $(x_u,x_v) \in \DE{\tourG{G}{T}}$.
                Therefore, by definition of $H_X$, $\{u, v\}$ is an edge of $H_X$.

            \item Conversely, assume now that $\{x_u, x_v\}$ is not an edge of $\indGraph{G}{X}$. Then,
                since $c(x_u)\neq c(x_v)$ and by definition of $\tourG{G}{T}$, $(u,v)\in \DE{T}$
                if and only if $(x_v,x_u) \in \DE{\tourG{G}{T}}$.
                Therefore, by definition of $H_X$, $\{u, v\}$ is not an edge of $H_X$.
        \end{itemize}
        The claim follows.
    \end{claimproof}

    On the one hand, if $\tourG{G}{T}[X]$ is isomorphic to $T$, then by \cref{claim:tour:H:iso} we obtain that $T$ is
    isomorphic to $T_{H_X}$. Further, by \cref{claim:H:G:iso}, we obtain that $H_X$ is isomorphic (with
    respect to $c$) to $\indGraph{G}{X}$.
    Thus, $\indGraph{G}{X}$ is isomorphic to a subgraph $H_X$ with $T \cong T_{H_X}$.

    On the other hand, let $\indGraph{G}{X}$ be isomorphic (with respect to $c$) to some $H$ with $T \cong T_{H}$.
    We show that $H \equiv H_X$, where $H_X$ is the graph that we obtain from $\tourG{G}{T}[X]$.
    By \cref{claim:H:G:iso} the function $c^{-1}$ restricted to $X$ defines
    an isomorphism from $H_X$ to $\indGraph{G}{X}$. Further, by assumption, $c$ restricted to $X$ defines
    an isomorphism from $\indGraph{G}{X}$ to $H$. Thus, $c \circ c^{-1} = \id_{\setn{k}}$ is an isomorphism from $H_X$ to $H$
    and hence $H \equiv H_X$.
    By \cref{claim:tour:H:iso}, we obtain that $\tourG{G}{T}[X]$ is isomorphic to $T_{H_X} = T_{H}$ with respect to $c$.
    Thus, $\tourG{G}{T}[X]$ is also isomorphic to $T \cong T_{H}$.  The lemma follows.
\end{proof}

\begin{remark}\dglabel^{rem:indsub:basis}[lem:tour:colindsub]
    Observe that $\CPindsubs{K_k}{\star}$ is always part of the linear combination
    in \cref{lem:tour:colindsub}. In an ideal world, we would be able to extract the
    term $\CPindsubs{K_k}{\star}$ directly from the linear combination, which would then allow us
    to count $k$-cliques. Unfortunately, this is not possible for linear combinations
    of $\CPindsubProb$-counts. For example, consider the sum over all graphs with $k$-vertices
    \[f(\star) = \sum_{H \in \graphs{k}} \CPindsubs{H}{\star},\]
    then $f(G)$ counts the number of colorful induced subgraph of $G$ of size $k$. This can
    be computed in linear time by simply taking the product over the size of all $k$ color classes
    in $G$.
    However, $f(\star)$ also contains the clique counting term $\CPindsubs{K_k}{\star}$ which cannot
    be computed in linear time, unless ETH fails.
\end{remark}

Due to \cref{rem:indsub:basis}, we need a way to rewrite the linear combination of $\CPindsubProb$-counts
into a new basis that allows us to extract single terms. In \cref{sec:complex:monotone},
we show that the $\CPsubProb$-basis has this property which then yields our
reduction from $\CPsubProb(\{\cmatch{k}\})$ to $\CFsubProb(\{T\})$.
Already in~\cite{DBLP:conf/stoc/CurticapeanDM17}, the authors use similar ideas
of changing bases. In particular, consult~\cite[Section~3.1]{DBLP:conf/stoc/CurticapeanDM17}
for an implicit proof of the following \cref{lem:colindsub:colsub}.
For completeness, we defer the proof of \cref{lem:colindsub:colsub} to the appendix.
\begin{restatable*}[{Basis transformation $\CPindsubProb$-basis to $\CPsubProb$-basis}]{lemma}{lemcolindsub}\dglabel{lem:colindsub:colsub}
    Let $H$ be a $k$-labeled graph and $G$ be a $k$-colored graph, then
    \begin{align*}
        \CPindsubs{H}{G} = \sum_{H \subseteq H'} {(-1)}^{|E(H')| - |E(H)|} \cdot \CPsubs{H'}{G},
        &\qedhere
    \end{align*}
    where the sum ranges over all edge-supergraphs $H'$ of $H$.
\end{restatable*}

\Cref{lem:colindsub:colsub} allows us to rewrite the linear combination of \cref{lem:tour:colindsub}.
We therefore obtain a new linear combination that uses $\CPsubProb$-counts.
Thus, we also obtain new coefficients that are given by the \emph{alternating enumerator}
which we define next.

\begin{definition}[The alternating enumerator $\ae{T}{H}$ of a tournament \(T\) and a
    graph \(H\)]\dglabel{def:alternating:enumerator}
    Let $T$ be a tournament and $H$ a graph on the same vertex-set.
    The \emph{alternating enumerator} of $H$ with respect to $T$ is defined as
    \begin{align*}
        \ae{T}{H} \coloneqq {(-1)}^{|E(H)|} \sum_{H' \subseteq H} {(-1)}^{|E(H')|} \left[T_{H'} \cong T\right],
    \end{align*}
    where the sum ranges over all edge-subgraphs $H'$ of $H$, and $\left[T_{H'} \cong T\right]$
    is equal to $1$ if $T_{H'} \cong T$ and $0$ otherwise.
\end{definition}

\begin{remark}
    Observe that this alternating enumerator is very similar to the alternating enumerator
    of counting graph properties that was studied in~\cite{topo,alge,DMW24,DBLP:conf/soda/DoringMW25,CN24}.
\end{remark}

By combining the results of this section, we obtain
our main technical result.

\sevenfifteenfour
\begin{proof}
    By  \cref{lem:tour:colindsub} and \cref{lem:colindsub:colsub}, we obtain
    \begin{align*}
        \CFsubs{T}{\tourG{G}{T}} &= \sum_{H' \in \graphs{k}} \left[T_{H'} \cong T\right] \cdot \CPindsubs{H'}{G} \\
        &= \sum_{H' \in \graphs{k}} \left[T_{H'} \cong T\right] \sum_{H' \subseteq H} {(-1)}^{|E(H)| - |E(H')|} \cdot \CPsubs{H}{G}
    \end{align*}
    where the last sum ranges over all edge-supergraphs $H$ of $H'$.
    For a fixed $H \in \graphs{k}$ observe that each edge-subgraph $H' \subseteq H$
    contributes to $\CPsubs{H}{G}$ with a factor of $\left[T_{H'} \cong T\right] \cdot {(-1)}^{|E(H)| - |E(H')|}$. By \cref{def:alternating:enumerator}, we obtain
    \[\CFsubs{T}{\tourG{G}{T}} = \sum_{H \in \graphs{k}} \ae{T}{H} \cdot \CPsubs{H}{G}.
        \qedhere
    \]
\end{proof}

\subsubsection*{Step 1c: Understanding the Complexity of Linear Combinations of $\CPsubProb$-counts}
\label{sec:complex:monotone}

In this section, we show that the $\CPsubProb$-basis has a very useful
property that is known as
\emph{complexity monotonicity}.
This allows us to extract single terms from a linear combinations
of $\CPsubProb$-counts.

\sevenfifteensix
\begin{proof}
    Let $G$ be a $k$-colored input graph with $n$ vertices with coloring $c \colon V(G) \to \setn{k}$.
    For any $k$-labeled graph $F$, let $G_F$ be the $k$-colored edge-subgraph of
    $G$ that we obtain by deleting all edges $\{u, v\} \in E(G)$ with
    $\{c(u), c(v)\} \notin E(F)$. The coloring of $G_F$ is given by $c$.

    \begin{claim}\label{claim:colsub:monotone}
        Let $F$ be any $k$-labeled graph, then
        \[f(G_F) = \sum_{\substack{i = 1 \\ H_i \subseteq F}}^m \alpha_i \cdot \CPsubs{H_i}{G},\]
        where the sum ranges over graphs $H_i$ that are edge-subgraphs of $F$.
    \end{claim}
    \begin{claimproof}
        Fix a graph  $F'$.
        The claim directly follows from
        \[\CPsubs{F'}{G_F} =
        \begin{cases}
            \CPsubs{F'}{G} & \text{if $ F' \subseteq F$},\\
            0            & \text{otherwise}.
          \end{cases}
        \]
        First, if $F'$ is not an edge-subgraph of $F$ then $F'$ contains an edge $\{u', v'\} \notin E(F)$.
        Note that all edges $\{u, v\}$ in $G$ with $\{c(u) = u', c(v) = v'\}$ were deleted in $G_F$,
        meaning that there are no subgraphs in $G_F$ that are isomorphic to $F'$ and color prescribed.

        Second, let $F'$ be an edge-subgraph of $F$. Observe that
        $\CPsubs{F'}{G_F} \leq \CPsubs{F'}{G}$ holds since $G_F$ is a subgraph of $G$.
        Let $\edgesub{\indGraph{G}{A}}{S} \cong F'$ be a color respecting subgraph of $G$ with
        $A \subseteq V(G)$, $S \subseteq E(G) \cap A^2$.
        Then $S$ is also a subset of $E(G_F) \cap A^2$, as $\{u, v\} \in S$ implies
        $\{c(u), c(v)\} \in E(F)$ since $c$ defines an isomorphism
        from $\edgesub{\indGraph{G}{A}}{S}$ to $F$. Thus, $\edgesub{\indGraph{G}{A}}{S}$
        is also a color respecting subgraph of $G_F$
        which shows $\CPsubs{F'}{G_F} \geq \CPsubs{F'}{G}$. The claim follows.
    \end{claimproof}
    For every $i \in \setn{m}$, we write $\alpha_{H_i}$ instead of $\alpha_{i}$. We
    further write $H \coloneqq H_j$ and show that
    \begin{align}\label{eq:lin:comb:cpsub}
        \alpha_{H} \cdot \CPsubs{H}{G} = \sum_{F \subseteq H} {(-1)}^{|E(F)| - |E(H)|} f(G_F).
    \end{align}
    By \cref{claim:colsub:monotone}, we obtain
    \[\sum_{F \subseteq H} {(-1)}^{|E(F)| - |E(H)|} f(G_F) = \sum_{F \subseteq H} {(-1)}^{|E(F)| - |E(H)|}
    \sum_{\substack{i = 1 \\ H_i \subseteq F}}^m \alpha_{H_i} \cdot \CPsubs{H_i}{G} =
    \sum_{H_i \subseteq H} \beta_{H_i} \cdot \CPsubs{H_i}{G}, \]
    where we collected the coefficients for each $\CPsubs{H_i}{G}$-term
    to obtain the values $\beta_{H_i}$. In the following, we prove
    \cref{eq:lin:comb:cpsub} by showing that $\beta_H = \alpha_H$ and $\beta_{H_i} = 0$
    for all other graphs. Observe that
    \[\beta_{H_i} = \sum_{H_i \subseteq F \subseteq H} {(-1)}^{|E(F)| - |E(H)|}
        \cdot \alpha_{H_i} = \alpha_{H_i} \cdot {(-1)}^{|E(H)|} \sum_{H_i \subseteq F
    \subseteq H} {(-1)}^{|E(F)|},\]
    where the sum ranges over all graphs $F$ such that $H_i$ is an edge-subgraph of $F$
    and $F$ is an edge-subgraph of $H$. We can now continue with
    \[\sum_{H_i \subseteq F \subseteq H} {(-1)}^{|E(F)|} = \sum_{S \subseteq E(H) \setminus E(H_i)}
        {(-1)}^{|S \uplus E(H_i)|} = {(-1)}^{|E(H_i)|} \cdot
    \sum_{S \subseteq E(H) \setminus E(H_i)} {(-1)}^{|S|},\]
    which is an alternating sum over all subsets of $E(H) \setminus E(H_i)$.
    This sum is zero unless $E(H) \setminus E(H_i) = \emptyset$
    which is equivalent to $H_i \equiv H$.
    Thus, $\beta_H = \alpha_H$ and $\beta_{H_i} = 0$
    for all other graphs.

    By~\cref{eq:lin:comb:cpsub}, we compute
    $\alpha_{H} \cdot \CPsubs{H}{G}$ by calling $f(\star)$ at most
    $2^{|E(H)|} \leq 2^{\binom{k}{2}}$ times on graphs of order $n$ that are all computed
    in time $O(n^2)$.
\end{proof}

\begin{remark}
    Observe that \cref{lem:colsub:monotone} assumes that all graphs $H_i$ have exactly $k$ vertices.
    This is necessary since the input consists of a $k$-colored
    graph $G$ and $\CPsubs{H_i}{G}$ is only well-defined if $H_i$ is $k$-labeled.
    In~\cite[Lemma A.3]{CN24_arxiv}, the authors showed a similar statement for
    the $\subProb$-basis. Again, they assumed that
    all graphs of the linear combination have the same number of vertices.
    This is necessary since otherwise we can obtain linear combination of
    $\subProb$-counts that contain \emph{hard}-terms but whose linear combination is
    easy to compute (see~\cite[Example 1.12]{DBLP:conf/stoc/CurticapeanDM17}).
\end{remark}

Using \cref{lem:colsub:monotone}, we directly obtain a reduction from
$\CPsubProb(\{H\})$ to $\CFsubProb(\{T\})$ whenever $\ae{T}{H}$ is
non-vanishing.

\begin{theorem}[{$\CFsubProb(\{T\})$ is harder than $\CPsubProb(\{H\})$
    for non-vanishing alternating enumerator $\ae{T}{H} \neq 0$}]\dglabel{cor:ae:reduction}[theo:tour:ae,lem:colsub:monotone]
    Let $T$ be a fixed $k$-labeled tournament and $H$ be a $k$-labeled graph with $\ae{T}{H} \neq 0$.
    Further, let $\mathcal{T} = \{T_1, T_2, \dots \}$ be a r.e.\ set of
    tournaments and $\mathcal{H} = \{H_1, H_2, \dots\}$ be a r.e.\ set of
    graphs with $\ae{T_i}{H_i} \neq 0$ for every $i\in \mathbb{N}$.

    Assume that there is an algorithm that
    computes $\CFsubProb(\{T\})$ for any $k$-colored tournament of order $n$
    in time $O(f(n))$. Then
    there exists an algorithm that computes
    $\CPsubProb(\{H\})$ for any $k$-colored graph of order $n$
    in time $O(g(k) \cdot f(n))$ for some computable function $g(k)$.
    In particular, $\Exp{\CFsubProb(\{T\})} \geq \Exp{\CPsubProb(\{H\})}$
    and $\CPsubProb(\mathcal{H}) \fpt \CFsubProb(\mathcal{T})$.
\end{theorem}
\begin{proof}
    Let $G$ be a $k$-colored graph of order $n$.
    By~\cref{theo:tour:ae}, we rewrite
    $\CFsubProb(\{T\})$ as a linear combination of colored subgraph counts
    that has at most $m = 2^{\binom{k}{2}}$ many terms. Note that we can compute all coefficients
    of this linear combination in time $O(g'(k))$ for some computable function $g'$.
    Since the coefficient of $H$ is non-vanishing, we use \cref{lem:colsub:monotone}
    to extract $\CPsubs{H}{G}$ by calling $\CFsubProb(\{T\})$ at most
    $h(k)$ many times on graphs of order at most $n$ that are all computable in time $O(n^2)$.
    Thus, $\CPsubProb(\{H\})$ can be computed in time $O(g(k) \cdot f(n))$
    for some computable function $g$.

    For a recursively enumerable set of tournaments $\mathcal{T}$,
    we use the above construction to obtain a parameterized Turing reduction
    from $\CPsubProb(\mathcal{H})$ to $\CFsubProb(\mathcal{T})$.
    Observe that, given an input $(H, G)$, we first find
    a graph $T \in \mathcal{T}$ with $\ae{T}{H} \neq 0$ in time
    $O(g''(|V(H)|))$ for some computable function $g''$. Note that
    $|V(T)| = |V(H)|$. Thus, the size of the parameter does not change
    in the following. Next,
    we use the above construction to compute $\CPsubs{H}{G}$
    by querying $\CFsubProb(\mathcal{T})$ at most $h(|V(T)|)$ times
    on inputs $(T, G')$ with $|V(G')| \leq |V(G)|$.
    All these computations take time $O(h(k) \cdot |V(G)|^2)$
    for some computable function $h$. This shows that there is a parameterized
    Turing reduction from $\CPsubProb(\mathcal{H})$ to $\CFsubProb(\mathcal{T})$.
\end{proof}

\subsubsection*{Step 1d: Analyzing the Alternating Enumerator of Anti-matchings}

By \cref{cor:ae:reduction}, for any tournament $T$ of order $k$, we obtain a reduction from $\CPsubProb(\{\cmatch{k}\})$
to $\CFsubProb(\{T\})$ if we show that the alternating
enumerator of $\cmatch{k}$ is always non-vanishing.
To this end, we first rewrite $\ae{T}{\cmatch{k}}$ using permutations:
given a $k$-labeled tournament $T$
and a permutation $\sigma \in \sym{k}$, we write
$\Tperm{T}{\sigma}$ for the tournament that we obtain by applying $\sigma$ to $T$.
Formally, $V(\Tperm{T}{\sigma}) \coloneqq V(T)$ and
$\DE{\Tperm{T}{\sigma}} \coloneqq \{(\sigma(u), \sigma(v)) : (u, v) \in \DE{T}\}$.
Now, we use $\Tperm{T}{\sigma}$ to replace the isomorphism test inside
the alternating enumerator with an equality test.

\begin{lemma}[Alternating enumerator via permutations]\dglabel{lem:ae_reformulation}[def:alternating:enumerator]
    Let $H$ be a $k$-labeled, then for any $k$-labeled tournament $T$, we have
    \[ |\auts{T}| \cdot \ae{T}{H} = {(-1)}^{|E(H)|} \sum_{H' \subseteq H} \sum_{\sigma \in \sym{k}}
    {(-1)}^{|E(H')|} \left[T_{H'} \equiv\Tperm{T}{\sigma}\right], \]
    where the sums ranges over all edge-subgraphs $H'$ of $H$ and permutations in $\sym{k}$.
    We define $\left[T_{H'} \equiv T\right]$ to be $1$ if $T_{H'} \equiv T$ and $0$ otherwise.
\end{lemma}
\begin{proof}
    By \cref{def:alternating:enumerator} of the alternating enumerator, we have
    \[ \ae{T}{H} \coloneqq {(-1)}^{|E(H)|} \sum_{H' \subseteq H} {(-1)}^{|E(H')|}
    \left[T_{H'} \cong T\right]. \]
    Fix an $H' \subseteq H$. Observe that the tournaments $T_{H'}$ and $T$ are isomorphic if and only if
    there exists a permutation $\sigma \in \sym{k}$---this accounts for the isomorphism between $T_{H'}$
    and $T$. Therefore, $T_{H'}$ and $T$ are isomorphic if and only if there exists a permutation
    $\sigma \in \sym{k}$ such that $\left[ T_{H'} \equiv T_\sigma \right]$. Moreover, there are
    exactly $|\auts{T}|$ such permutations, as two distinct isomorphisms from $T_{H'}$ to $T$ differs
    by an automorphism of $T$.
    We thus obtain
    \[|\auts{T}| \cdot \left[ T_{H'} \cong T \right] =  \sum_{\sigma \in \sym{k}} \left[
    T_{H'} \equiv \Tperm{T}{\sigma} \right]. \]

    The result then follows by replacing the term $\left[ T_{H'} \cong T \right]$ in the definition
    of the alternating enumerator and rearranging the sums.
\end{proof}

We continue with simplifying the sum of \cref{lem:ae_reformulation} by reducing the number of terms.
Currently, for each $H' \subseteq H$, we compute the sum
$\sum_{\sigma } \left[T_{H'} \equiv\Tperm{T}{\sigma}\right]$. Note however
that there is at most one permutation $\sigma \in \sym{k}$ with $T_{H'} \equiv\Tperm{T}{\sigma}$,
meaning that we can replace this sum by checking if there exists a permutation $\sigma$
with  $T_{H'} \equiv\Tperm{T}{\sigma}$. To this end, we introduce the
\emph{symmetric difference} of tournaments.
Given two tournaments $T$ and $T'$
on the same vertex set, we write
$\flip{T}{T'}$ for the set of undirected edges $\{u, v\}$ on which $T$ and $T'$
disagree (that is, $(u, v) \in \DE{T}$ and $(u, v) \notin \DE{T'}$ (or vice versa)).

\begin{lemma}[{Alternating enumerator via symmetric difference}]\dglabel{lem:ae:complement}[lem:ae_reformulation]
    For all tournaments $T$ and graphs $H$ of order $k$, we have
    \[|\auts{T}| \cdot \ae{T}{\overline{H}} = {(-1)}^{|E(H)|} \cdot
        \sum_{\substack{\sigma \in \sym{k} \\ E(H) \subseteq \flip{T}{\Tperm{T}{\sigma}}}}
    {(-1)}^{| \flip{T}{\Tperm{T}{\sigma}} |}. \]
\end{lemma}
\begin{proof}
    \Cref{lem:ae_reformulation} yields
    \begin{align}\label{eq:x:sum1}
        |\auts{T}| \cdot \ae{T}{\overline{H}}  =
        {(-1)}^{|E(\overline{H})|}
        \sum_{H' \subseteq \overline{H}} \sum_{\sigma \in \sym{k}}
        {(-1)}^{|E(H')|} \left[T_{H'} \equiv T_\sigma\right].
    \end{align}
    Next, we define for all permutations $\sigma \in \sym{k}$
    the value $x_\sigma \coloneqq \sum_{H' \subseteq \overline{H}}
    {(-1)}^{|E(H')|} \left[T_{H'} \equiv T_\sigma\right] $.
    By swapping  the order of summation in \cref{eq:x:sum1}, we obtain
    \begin{align}\label{eq:x:sum}
        |\auts{T}| \cdot \ae{T}{\overline{H}}  =
        {(-1)}^{|E(\overline{H})|}
        \sum_{\sigma \in \sym{k}}\sum_{H' \subseteq \overline{H}}
        {(-1)}^{|E(H')|} \left[T_{H'} \equiv T_\sigma\right]
        = {(-1)}^{|E(\overline{H})|} \cdot \sum_{\sigma \in \sym{k}} x_\sigma.
    \end{align}
    Next, we show
    \begin{align}\label{eq:x:sigma}
        x_\sigma = \begin{cases}
            {(-1)}^{\binom{k}{2}} \cdot {(-1)}^{|\flip{T}{T^\sigma}|} &\text{ if $E(H) \subseteq \flip{T}{T^\sigma}$} \\
            0 &\text{otherwise}.
        \end{cases}
    \end{align}
    To this end, observe that $[T_{H'} \equiv T^\sigma] = 1$ if and only if
    the edges on which $T$ and $T^\sigma$ disagree are exactly the missing edges of $H'$.
    In other words, $E(K_k) \setminus E(H') = \flip{T}{T^\sigma}$, meaning
    that there is a unique $H'$ with $[T_{H'} \equiv T^\sigma] = 1$. However,
    this specific $H'$ only appears in the sum of $x_\sigma$ if $H'$
    is an edge-subgraph of $\overline{H}$. This condition is equivalent to
    \[E(H') = E(K_k) \setminus \flip{T}{T^\sigma} \subseteq E(\overline{H}) = E(K_k) \setminus E(H),\]
    which is equivalent to $E(H) \subseteq \flip{T}{T^\sigma}$. So, if this condition is violated
    then $x_\sigma = 0$. Otherwise, we obtain
    \[x_\sigma = {(-1)}^{|E(H')|} = {(-1)}^{|E(K_k) \setminus \flip{T}{T^\sigma}|} = {(-1)}^{\binom{k}{2}} \cdot {(-1)}^{|\flip{T}{T^\sigma}|}, \]
    which proves \cref{eq:x:sigma}. By plugging this into \cref{eq:x:sum}, we obtain
    \[ |\auts{T}| \cdot \ae{T}{\overline{H}}  =  {(-1)}^{\binom{k}{2}} \cdot {(-1)}^{|E(H)|}
        \sum_{\substack{\sigma \in \sym{k} \\ E(H) \subseteq \flip{T}{\Tperm{T}{\sigma}}}}
        {(-1)}^{\binom{k}{2}} \cdot  {(-1)}^{|\flip{T}{\Tperm{T}{\sigma}}|} = {(-1)}^{|E(H)|} \cdot
        \sum_{\substack{\sigma \in \sym{k} \\ E(H) \subseteq \flip{T}{\Tperm{T}{\sigma}}}}
    {(-1)}^{|\flip{T}{\Tperm{T}{\sigma}}|}.
    \qedhere
\]
\end{proof}

\begin{remark}
    Observe that our results so far did not need that \(T\) is a tournament---indeed,
    everything up to this point may be used to analyze $\ae{T}{H}$ for a general graph \(T\).
\end{remark}

Next, we compute $\ae{T}{\cmatch{k}}$.
By \cref{lem:ae:complement}, we have to find the permutations $\sigma$
with $E(\match{k}) \subseteq \flip{T}{\Tperm{T}{\sigma}}$.

\begin{definition}[{Ordered maximal matchings $\matchOrder{k}$, unordered maximal matchings $\matchSet{k}$}]\dglabel{def:ordered:matching}
    For every integer $k$, we define the following.
    \begin{enumerate}
        \item We write $\matchOrder{k}$ for the set of all ordered tuples of $\lfloor k/2 \rfloor$
            edges in $E(K_k)$ that together form a matching.

        \item For $M \in \matchOrder{k}$ and $\sigma \in \sym{\lfloor k/2 \rfloor}$,
            we write $\Mperm{M}{\sigma}$ for the ordered tuples that is obtained by permuting
            the $\lfloor k/2 \rfloor$ edges of $M$ according to $\sigma$.

        \item For two $M, M' \in \matchOrder{k}$, we write $M \sim M'$ if their underlying edges
            are the same. This is equivalent to the existence of a	$\sigma \in \sym{\lfloor k/2 \rfloor}$
            with $\Mperm{M}{\sigma} = M'$. Note that $\sim$ is an equivalence relation on $\matchOrder{k}$.

        \item We write $\matchSet{k}$ for a set of representatives under $\sim$. Note that
            $\matchOrder{k} = \{\Mperm{M}{\sigma} : M \in \matchSet{k}, \sigma \in \sym{\lfloor k/2 \rfloor}  \}$.
            \qedhere
    \end{enumerate}
\end{definition}

Observe that elements of $\matchOrder{k}$ are all maximal matchings on $K_k$
whose edges are ordered. The elements of $\matchSet{k}$ are
\emph{unordered maximal matchings on $K_k$}.

\begin{lemma}[Cardinality of matching set $\matchSet{k}$]\dglabel{lem:card:matching}[def:ordered:matching]
    For all $k \geq 2$, the cardinality of $\matchSet{k}$ is odd.
\end{lemma}
\begin{proof}
    If $k$ is even then each element of $\matchSet{k}$ is a perfect matching
    on $K_k$. The number of perfect matchings on $K_k$ is equal to $(k-1)!!$
    (see~\cite{callan2009combinatorialsurveyidentitiesdouble})
    and therefore odd.\footnote{We write $k!!$ for the double factorial of $k$, that is the product of all the positive integers up to $k$ that have the same parity as $k$.}

    If $k$ is odd then each $\matchSet{k}$ is a matching on $k-1$ vertices, where one
    vertex remains unmatched. There are $k$ possible choices for the unmatched vertex and
    $(k-2)!!$ possible perfect matchings for the remaining $k-1$ vertices. Thus, $|\matchSet{k}|$
    is equal to $k \cdot (k-2)!! = k!!$, which is odd.
\end{proof}

We now show how to transform an ordered maximal matching into a permutation.

\begin{definition}[Maximal matchings and permutations $\MTour{M}{T}$]\dglabel{def:Miso}[def:match,def:ordered:matching]
    For each $k$-labeled tournament $T$ and
    $M = (\{u_1, v_1\}, \dots, \{u_{\lfloor k/2 \rfloor}, v_{\lfloor k/2 \rfloor}\}) \in \matchOrder{k}$,
    we define $\MTour{M}{T} \coloneqq \Tperm{T}{\Miso{M}{T}}$ as the tournament
    that is obtained by applying the following permutation $\Miso{M}{T} \in \sym{k}$ to $T$.

    For all $i \in \setn{\lfloor k/2 \rfloor}$, define $\Miso{M}{T}$ such that $\Miso{M}{T}$ maps
    $\{u_i, v_i\}$ to $\{2i - 1, 2i\}$ and the tournaments $\Tperm{T}{\Miso{M}{T}}$ and $T$
    have the opposite orientation on
    $\{\Miso{M}{T}(u_i), \Miso{M}{T}(v_i)\} = \{2i- 1, 2i\}$. Further,
    if $k$ is odd, then $\Miso{M}{T}$ maps the unique unmatched vertex in $M$ to $k$.
\end{definition}

\begin{figure}[t]
    \centering
    \includegraphics{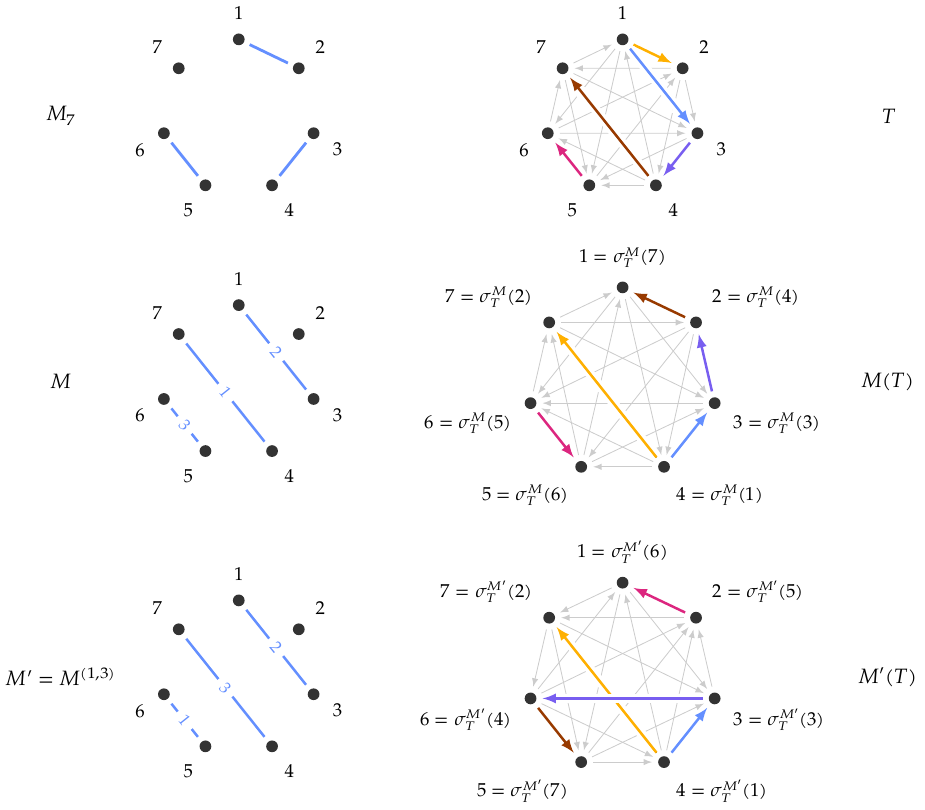}
    \caption{The 7-matching $M_7$, a tournament $T$,
        the ordered maximal matching $M\in \matchOrder{7}$ consisting
        of $(\{4,7\}, \{1,3\}, \{5,6\})$, the tournament $\MTour{M}{T}$,
        the ordered maximal matching $M'\in \matchOrder{7}$ obtained from $M$
        via the transposition $(1,3)$, and the tournament $\MTour{M'}{T}$.
    We highlight select edges of $T$, $\MTour{M}{T}$, and $\MTour{M'}{T}$.}%
    \label{fig:Mperm}
\end{figure}

See \cref{fig:Mperm} for an example.
Observe that two different
ordered maximal matchings $M, M' \in \matchOrder{k}$ define
two different permutations $\Miso{M}{T}$ and $\Miso{M'}{T}$.
We show that these permutations are precisely those
permutations with the property that $E(\match{k}) \subseteq \flip{T}{\Tperm{T}{\sigma}}$
meaning that we can use them for \cref{lem:ae:complement}.

\begin{lemma}[Symmetric difference and permutations]\dglabel{lem:bij:matching}[def:Miso]
    For each $k$-labeled tournament $T$, we have
    $\{\sigma \in \sym{k}: E(\match{k}) \subseteq \flip{T}{\Tperm{T}{\sigma}}\} =
    \{\Miso{M}{T} : M \in \matchOrder{k}\}$.
\end{lemma}
\begin{proof}
    Let $M \in \matchOrder{k}$. By definition of $\Miso{M}{T}$ (see \cref{def:Miso}), tournaments
    $\Tperm{T}{\Miso{M}{T}}$ and $T$ disagree on all edges
    $\{2i - 1, 2i\}$. Thus, $E(\match{k})  \subseteq \flip{T}{T^{\Miso{M}{T}}}$ which proves
    $\{\Miso{M}{T} : M \in \matchOrder{k}\} \subseteq \{\sigma \in \sym{k}:E(\match{k}) \subseteq \flip{T}{\Tperm{T}{\sigma}}\}$.

    For the other direction, let
    $\sigma \in  \sym{k}$ be such that $E(\match{k}) \subseteq \flip{T}{\Tperm{T}{\sigma}}$.
    For every $i \in \setn{\lfloor k/2 \rfloor}$, there are exactly two vertices $u_i$ and
    $v_i$ with $\sigma(u_i) = 2i-1$ and $\sigma(v_i) = 2i$. We define
    $M = (\{u_1, v_1\}, \dots, \{u_{\lfloor k/2 \rfloor}, v_{\lfloor k/2 \rfloor}\})$.
    To see that $M \in \matchOrder{k}$, note that $M$ is a
    matching since $\sigma^{-1}$ is an isomorphism that maps the matching $\match{k}$ to the
    edge set $M$.

    We show that $\sigma = \Miso{M}{T}$. Note that
    $\Miso{M}{T}$ is defined in a way such that  $\Tperm{T}{\Miso{M}{T}}$ and $T$
    disagree on
    $\{\Miso{M}{T}(u_i), \Miso{M}{T}(v_i)\} = \{2i- 1, 2i\}$.
    Moreover, $\Tperm{T}{\sigma}$ and $T$ disagree on
    $\{\sigma(u_i), \sigma(v_i)\} = \{2i- 1, 2i\}$ since $\{2i - 1, 2i\} \in  \flip{T}{\Tperm{T}{\sigma}}$.
    Since  $\{\sigma(u_i), \sigma(v_i)\} = \{\Miso{M}{T}(u_i), \Miso{M}{T}(v_i)\}  = \{2i- 1, 2i\}$,
    we obtain that $\Tperm{T}{\Miso{M}{T}}$ and $\Tperm{T}{\sigma}$ agree on $\{2i - 1, 2i\}$.
    This implies $\sigma(u_i) = \Miso{M}{T}(u_i)$ and $\sigma(v_i) = \Miso{M}{T}(v_i)$. If
    $k$ is even, then this shows that $\sigma(u_i)$ and $\Miso{M}{T}$ coincide on all elements.
    If $k$ is odd, then this shows that this shows that $\sigma(u_i)$ and $\Miso{M}{T}$ coincide
    on all but in a single element. Since two distinct permutations differ on at least
    two elements, we obtain $\sigma = \Miso{M}{T}$ in this case, too. This shows
    $\{\sigma \in \sym{k}: E(\match{k}) \subseteq \flip{T}{\Tperm{T}{\sigma}}\} \subseteq \{\Miso{M}{T} : M \in \matchOrder{k}\}$.
\end{proof}

By combining  \cref{lem:ae:complement,lem:bij:matching}, we obtain
\begin{align}\label{eq:t:ae:mk}
    |\auts{T}| \cdot \ae{T}{\overline{\match{k}}} = {(-1)}^{|E(\match{k})|} \cdot
    \sum_{M \in \matchSet{k}}  \sum_{\sigma \in \sym{\lfloor k/2 \rfloor}}
    {(-1)}^{| \flip{T}{\MTour{\Mperm{M}{\sigma}}{T}} | }.
\end{align}
Here, we use that $\matchOrder{k} = \{\Mperm{M}{\sigma} : M \in \matchSet{k}, \sigma \in \sym{\lfloor k/2 \rfloor}  \}$.
Thus, it is instructive to prove that $| \flip{T}{\MTour{\Mperm{M}{\sigma}}{T}} |$ has the same parity for all
permutations $\sigma \in \sym{\lfloor k/2 \rfloor}$.

\begin{lemma}
    \dglabel{lemma:parity:matching:permutation}($|\flip{\MTour{M}{T}}{\MTour{\Mperm{M}{\varphi}}{T}}|$ is even)
    Let $\varphi = (i, j)$ be a transposition and $M \in \matchOrder{k}$. Then
    the cardinality of $\flip{\MTour{M}{T}}{\MTour{\Mperm{M}{\varphi}}{T}}$ is even.
\end{lemma}
\begin{proof}
    In the following, we write $\sigma$ instead of $\Miso{M}{T}$.
    Thus, $\MTour{M}{T} \equiv \Tperm{T}{\sigma}$.
    Further, without loss of generality we assume that
    $M \coloneqq (\{u_1, v_1\}, \dots, \{u_{\lfloor k/2 \rfloor}, v_{\lfloor k/2 \rfloor}\})$,
    where we use the convention that $u_t$ and $v_t$ are named in such a way
    that $(\sigma(u_t), \sigma(v_t)) \in \DE{T^\sigma}$.  We start by proving the following claim.

    \begin{claim}\label{claim:M:transposition}
        Let $\psi \coloneqq (\sigma(u_i), \sigma(u_j)) \circ (\sigma(v_i), \sigma(v_j))$
        be the permutation that swaps
        $\sigma(u_i)$ with $\sigma(u_j)$ and $\sigma(v_i)$ with $\sigma(v_j)$, then
        $\Mperm{(\MTour{M}{T})}{\psi}
        \equiv \MTour{\Mperm{M}{\varphi}}{T}$.
    \end{claim}
    \begin{claimproof}
        Observe that we obtain $\Mperm{M}{\varphi}$
        by swapping $\{u_i, v_i\}$ with $\{u_j, v_j\}$ in $M$. This means that the only
        difference between $\MTour{M}{T}$ and $\MTour{\Mperm{M}{\varphi}}{T}$
        is that in $\MTour{M}{T}$ the vertices $\{u_i, v_i\}$ maps to
        $\{2i- 1, 2i\} = \{\sigma(u_i), \sigma(v_i)\}$, and the vertices $\{u_j, v_j\}$ maps to
        $\{2j- 1, 2j\} = \{\sigma(u_j), \sigma(v_j)\}$. While in $\MTour{\Mperm{M}{\varphi}}{T}$
        the vertices $\{u_i, v_i\}$ maps to
        $\{2j- 1, 2j\} = \{\sigma(u_j), \sigma(v_j)\}$, and the vertices $\{u_j, v_j\}$ maps to
        $\{2i- 1, 2i\} = \{\sigma(u_i), \sigma(v_i)\}$. Thus, we obtain $\MTour{\Mperm{M}{\varphi}}{T}$
        by applying a permutation $\psi$ to $\MTour{M}{T}$ (i.e.,
        $\Mperm{(\MTour{M}{T})}{\psi} = \MTour{\Mperm{M}{\varphi}}{T}$)
        that swaps the vertices $\{\sigma(u_i), \sigma(v_i)\}$ with
        $\{\sigma(u_j), \sigma(v_j)\}$. There are in principle two possible
        permutations that swaps the elements of $\{\sigma(u_i), \sigma(v_i)\}$ with $\{\sigma(u_j), \sigma(v_j)\}$. However,
        we show that only the permutation $\psi$ with $\psi(\sigma(u_i)) = \sigma(u_j)$
        and $\psi(\sigma(v_i)) = \sigma(v_j)$ is possible. For this
        observe that, by definition (see \cref{def:Miso}), $\MTour{M}{T}$ and $\MTour{\Mperm{M}{\varphi}}{T}$ both disagree with $T$ on $\{\sigma(u_i), \sigma(v_i)\}$
        and $\{\sigma(u_j), \sigma(v_j)\}$, meaning that $\MTour{M}{T}$ and $\MTour{\Mperm{M}{\varphi}}{T}$ have the same orientation on $\{\sigma(u_i), \sigma(v_i)\}$
        and $\{\sigma(u_j), \sigma(v_j)\}$. Since $(\sigma(u_i), \sigma(v_i)) \in \DE{\MTour{M}{T}}$
        and $(\sigma(u_j), \sigma(v_j)) \in \DE{\MTour{M}{T}}$, this implies $\psi(\sigma(u_i)) = \sigma(u_j)$
        and $\psi(\sigma(v_i)) = \sigma(v_j)$, proving the claim.
    \end{claimproof}

    We define the sets $A = \{\sigma(u_i), \sigma(v_i), \sigma(u_j), \sigma(v_j)\}$
    and $B = V(T) \setminus A$. This allows us to partition the set
    $\flip{\MTour{M}{T}}{\MTour{\Mperm{M}{\varphi}}{T}}$ into three sets
    $\outE$, $\inE$, and $\intE$, where $\outE$ contains all edges included
    in $B$,  $\inE$ contains all edges included in $A$, and $\intE$ contains
    all edges between in $A$ and $B$.
    We show
    \[|\outE| \equiv_2 0, \quad |\inE| \equiv_2 0, \quad |\intE| \equiv_2 0.\]
    Observe that $|\outE| \equiv_2 0$ due to \cref{claim:M:transposition}. To this end, observe that
    $\MTour{M}{T}$ and $\MTour{\Mperm{M}{\varphi}}{T}$ are identical on $B$
    which implies $\outE = \emptyset$.

    Next, we show $|\intE| \equiv_2 0$.
    For all $z \in B$,
    write $B^u_z = \{\{\sigma(u_i), z\}, \{\sigma(u_j), z\}\}$.
    We show that $|B^u_z \cap \intE|$ is even.
    By \cref{claim:M:transposition}, when going from
    $\MTour{M}{T}$ to $\MTour{\Mperm{M}{\varphi}}{T}$,
    we swap the edge $\{\sigma(u_i), z\}$ in $\MTour{M}{T}$ with the edge
    $\{\sigma(u_j), z\}$ in $\MTour{M}{T}$.%
    \footnote{For instance
        if $(\sigma(u_i), z) \in \DE{\MTour{M}{T}}$ and $(\sigma(u_j), z) \in \DE{\MTour{M}{T}}$
        then $(\sigma(u_j), z) \in \DE{\MTour{\Mperm{M}{\varphi}}{T}}$ and
        $(\sigma(u_j), z) \in \DE{\MTour{\Mperm{M}{\varphi}}{T}}$.
        Next, if $(\sigma(u_i), z) \in \DE{\MTour{M}{T}}$ and $(z, \sigma(u_j)) \in \DE{\MTour{M}{T}}$
        then $(\sigma(u_j), z) \in \DE{\MTour{\Mperm{M}{\varphi}}{T}}$ and
        $(z, \sigma(u_j)) \in \DE{\MTour{\Mperm{M}{\varphi}}{T}}$.
    }

    If both edges have the same
    orientation\footnote{i.e., if either ($(\sigma(u_i), z) \in \DE{\MTour{M}{T}}$ and
    $(\sigma(u_j), z) \in \DE{\MTour{M}{T}}$) or ($(z, \sigma(u_i)) \in \DE{\MTour{M}{T}}$ and
    $(z, \sigma(u_j)) \in \DE{\MTour{M}{T}}$). }
    in $\MTour{M}{T}$ then $\{\sigma(u_i), z\}$  and  $\{\sigma(u_j), z\}$
    also have the same orientation in $\MTour{\Mperm{M}{\varphi}}{T}$.
    Hence, $B^u_z \cap \intE = \emptyset$.
    If both edges have the opposite orientation\footnote{i.e., if either
    ($(\sigma(u_i), z) \in \DE{\MTour{M}{T}}$ and
    $(z, \sigma(u_j)) \in \DE{\MTour{M}{T}}$) or ($(z, \sigma(u_i)) \in \DE{\MTour{M}{T}}$ and
    $(\sigma(u_j), z) \in \DE{\MTour{M}{T}}$).}
    then $\{\sigma(u_i), z\}$  and  $\{\sigma(u_j), z\}$ both
    have the opposite orientation in $\MTour{M}{T}$ and $\MTour{\Mperm{M}{\varphi}}{T}$.
    Hence, $B^u_z \subseteq \intE$. This proofs that $|B^u_z \cap \intE|$
    is either 0 or 2 and therefore even.

    We define $B^v_z = \{\{\sigma(v_i), z\}, \{\sigma(v_j), z\}\}$. By
    swapping the role of $\sigma(u_i)$ with $\sigma(v_i)$ and
    $\sigma(u_j)$ with $\sigma(v_j)$, we obtain that $|B^v_z \cap \intE|$
    is also even. We use $B^u_z$ and $B^v_z$ to partition $\intE$ which yields
    \begin{align*}
        |\intE| &= \left|\biguplus_{z \in B} \left(B^u_z \cap \intE\right) \uplus \left(B^v_z \cap \intE\right)\right|  = \sum_{z \in B} |B^u_z \cap \intE| + |B^v_z \cap \intE| \equiv_2 \sum_{z \in B} 0 + 0  \equiv_2 0 .
    \end{align*}

    Finally, we show $|\inE| \equiv_2 0$ by proving $|\inE| \in \{2, 4\}$. There
    are 6 edges between the vertices of $A = \{\sigma(u_i), \sigma(v_i), \sigma(u_j), \sigma(v_j)\}$ that
    we analyze in the following. First observe that $\{\sigma(u_i), \sigma(v_i)\}$ and
    $\{\sigma(u_j), \sigma(v_j)\}$ are both not in $\inE$. The reason is that by
    \cref{claim:M:transposition} both $\MTour{M}{T}$ and
    $\MTour{\Mperm{M}{\varphi}}{T}$ agree on $\{\sigma(u_i), \sigma(v_i)\}$ and
    $\{\sigma(u_j), \sigma(v_j)\}$. Next, by \cref{claim:M:transposition},
    observe that the edges
    $\{\sigma(u_i), \sigma(u_j)\}$ and $\{\sigma(v_i), \sigma(v_j)\}$ both are flipped
    when applying $\psi$. Hence, $\{\sigma(u_i), \sigma(u_j)\}, \{\sigma(v_i), \sigma(v_j)\} \in \inE$.

    Lastly, consider the edges $\{\sigma(u_i), \sigma(v_j)\}$ and $\{\sigma(u_j), \sigma(v_i)\}$.
    We consider all four cases:
    \begin{itemize}
        \item $(\sigma(u_i), \sigma(v_j)) \in \DE{\MTour{M}{T}}$
        and $(\sigma(u_j), \sigma(v_i)) \in \DE{\MTour{M}{T}}$:
        Then $(\sigma(u_j), \sigma(v_i)) \in \DE{\MTour{\Mperm{M}{\varphi}}{T}}$
        and $(\sigma(u_i), \sigma(v_j)) \in \DE{\MTour{\Mperm{M}{\varphi}}{T}}$,
        thus $\{\sigma(u_i), \sigma(v_j)\}, \{\sigma(u_j), \sigma(v_i)\} \notin \inE$.
        \item $(\sigma(v_j), \sigma(u_i)) \in \DE{\MTour{M}{T}}$
        and $(\sigma(u_j), \sigma(v_i)) \in \DE{\MTour{M}{T}}$:
        Then $(\sigma(v_i), \sigma(u_j)) \in \DE{\MTour{\Mperm{M}{\varphi}}{T}}$
        and $(\sigma(u_i), \sigma(v_j)) \in \DE{\MTour{\Mperm{M}{\varphi}}{T}}$,
        thus $\{\sigma(u_i), \sigma(v_j)\}, \{\sigma(u_j), \sigma(v_i)\} \in \inE$.
        \item  $(\sigma(u_i), \sigma(v_j)) \in \DE{\MTour{M}{T}}$
        and $(\sigma(v_i), \sigma(u_j)) \in \DE{\MTour{M}{T}}$:
        Then $(\sigma(u_j), \sigma(v_i)) \in \DE{\MTour{\Mperm{M}{\varphi}}{T}}$
        and $(\sigma(v_j), \sigma(u_i)) \in \DE{\MTour{\Mperm{M}{\varphi}}{T}}$,
        thus $\{\sigma(u_i), \sigma(v_j)\}, \{\sigma(u_j), \sigma(v_i)\} \in \inE$.
        \item  $(\sigma(v_j), \sigma(u_i)) \in \DE{\MTour{M}{T}}$
        and $(\sigma(v_i), \sigma(u_j)) \in \DE{\MTour{M}{T}}$:
        Then $(\sigma(v_i), \sigma(u_j)) \in \DE{\MTour{\Mperm{M}{\varphi}}{T}}$
        and $(\sigma(v_j), \sigma(u_i)) \in \DE{\MTour{\Mperm{M}{\varphi}}{T}}$,
        thus $\{\sigma(u_i), \sigma(v_j)\}, \{\sigma(u_j), \sigma(v_i)\} \notin \inE$.
    \end{itemize}
    In each case, either $\{\sigma(u_i), \sigma(v_j)\}$ and $\{\sigma(u_j), \sigma(v_i)\}$ are both in
    $\inE$, or neither are. This shows together with the other cases that
    $|\inE| \in \{2, 4\}$. The result now follows from
    \[|\flip{\MTour{M}{T}}{\MTour{\Mperm{M}{\varphi}}{T}}| = |\outE| + |\intE| + |\inE| \equiv_2 0 + 0 + 0.
    \qedhere\]
\end{proof}

\begin{lemma}
    \dglabel{lem:parity}[lemma:parity:matching:permutation]($| \flip{T}{\MTour{M}{T}} | \equiv_2 |\flip{T}{\MTour{\Mperm{M}{\varphi}}{T}}|$)
    Let $T$ be a $k$-labeled tournament, $M \in \matchSet{k}$ and
    $\varphi \in \sym{\lfloor k/2 \rfloor}$. Then
    \[| \flip{T}{\MTour{M}{T}} | \equiv_2 |\flip{T}{\MTour{\Mperm{M}{\varphi}}{T}}|. \]
\end{lemma}
\begin{proof}
    It is enough to show the statement for $\varphi = (i, j)$ since
    each permutation can be written as a composition of transpositions.
    We first show that
    \begin{align}\label{eq:sym:dif:flip}
        \flip{T}{\MTour{\Mperm{M}{\varphi}}{T}} =
    \symdif{(\flip{T}{\MTour{M}{T}})}{(\flip{\MTour{M}{T}}{\MTour{\Mperm{M}{\varphi}}{T}})}.
    \end{align}
    To see this, recall that $\{u, v\} \in  \flip{T}{\MTour{\Mperm{M}{\varphi}}{T}}$
    if and only if $T$ and $\MTour{\Mperm{M}{\varphi}}{T}$ disagree  on $\{u, v\}$.
    This is logical equivalent to either
    \begin{itemize}
        \item $T$ disagrees with $\MTour{M}{T}$ on $\{u, v\}$ and
         $\MTour{M}{T}$ agrees with $\MTour{\Mperm{M}{\varphi}}{T}$ on $\{u, v\}$, or
         \item $T$ agrees with $\MTour{M}{T}$ on $\{u, v\}$ and
         $\MTour{M}{T}$ disagrees with $\MTour{\Mperm{M}{\varphi}}{T}$ on $\{u, v\}$.
    \end{itemize}
    This statement is logical equivalent to
    $\{u, v\} \in \symdif{(\flip{T}{\MTour{M}{T}})}{(\flip{\MTour{M}{T}}{\MTour{\Mperm{M}{\varphi}}{T}})}$,
    thus implying \cref{eq:sym:dif:flip}.
    Lastly, observe that
    \begin{align}\label{eq:sym:dif:size}
        |\symdif{A}{B}| = |A| + |B| - 2|A \cap B|
    \end{align}
    To see this, note that $|A| + |B|$ counts each element that either appears in $A$ or
    $B$ once, and each element that appears in both $A$ and $B$ twice. Thus,
    $|A| + |B| - 2|A \cap B|$ only counts elements that either appear in $A$ or
    $B$. With this, we have everything to prove the theorem. We obtain
    \begin{align*}
        |\flip{T}{\MTour{\Mperm{M}{\varphi}}{T}}| &\stackrel{(\ref{eq:sym:dif:flip})}{\equiv_2}
        |\symdif{(\flip{T}{\MTour{M}{T}})}{(\flip{\MTour{M}{T}}{\MTour{\Mperm{M}{\varphi}}{T}})}| \\
        &\stackrel{(\ref{eq:sym:dif:size})}{\equiv_2} |\flip{T}{\MTour{M}{T}}| +
        |\flip{\MTour{M}{T}}{\MTour{\Mperm{M}{\varphi}}{T}}| + 0 \\
        &\equiv_2 |\flip{T}{\MTour{M}{T}}| + 0 + 0,
    \end{align*}
    where we use \cref{lemma:parity:matching:permutation} for the last step.
\end{proof}

We now combine \cref{lem:parity,eq:t:ae:mk} to show that
$\ae{T}{\cmatch{k}}$ is always non-vanishing.

\begin{theorem}
    \dglabel{theo:cmatch:not:zero}[lem:ae:complement,lem:card:matching,lem:bij:matching,lem:parity](The
    alternating enumerator of the anti-matching $\ae{T}{\cmatch{k}}$ is nonzero)
    Every $k$-labeled tournament \(T\) satisfies $\ae{T}{\cmatch{k}} \neq 0$.
\end{theorem}
\begin{proof}
    By \cref{lem:ae:complement}, we first obtain
    \[X \coloneqq \frac{|\auts{T}| \cdot \ae{T}{\cmatch{k}}}{(\lfloor k/2 \rfloor)!} =
    \frac{{(-1)}^{|E(\match{k})|}}{(\lfloor k/2 \rfloor)!}  \cdot
    \sum_{\substack{\sigma \in \sym{k} \\ E(\match{k}) \subseteq \flip{T}{\Tperm{T}{\sigma}}}}
    {(-1)}^{\left| \flip{T}{\Tperm{T}{\sigma}} \right|}. \]
    Next, by \cref{lem:bij:matching} we can take the sum over the set
    $\{\Miso{M}{T} : M \in \matchOrder{k}\}$ instead. Thus, we obtain
    \[X = \frac{{(-1)}^{|E(\match{k})|}}{(\lfloor k/2 \rfloor)!}
    \cdot \sum_{M \in \matchOrder{k}} {(-1)}^{| \flip{T}{T^{\Miso{M}{T}}} |} =
    \frac{{(-1)}^{|E(\match{k})|}}{(\lfloor k/2 \rfloor)!}
    \cdot \sum_{M \in \matchSet{k}}  \sum_{\sigma \in \sym{\lfloor k/2 \rfloor}}
    {(-1)}^{| \flip{T}{\MTour{\Mperm{M}{\sigma}}{T}} | }, \]
    where we use $\matchOrder{k} = \{\Mperm{M}{\sigma} : M \in \matchSet{k}, \sigma \in \sym{\lfloor k/2 \rfloor}  \}$.
    By \cref{lem:parity}, we have ${(-1)}^{|\flip{T}{\MTour{M}{T}}|} =
    {(-1)}^{|\flip{T}{\MTour{\Mperm{M}{\sigma}}{T}}|}$ for all $\sigma \in \sym{\lfloor k/2 \rfloor}$.
    Thus, we collect all these terms to obtain
    \[X = \frac{{(-1)}^{|E(\match{k})|}}{(\lfloor k/2 \rfloor)!}
    \cdot \sum_{M \in \matchSet{k}} (\lfloor k/2 \rfloor)!  \cdot {(-1)}^{| \flip{T}{\MTour{M}{T}}
    |} =
    {(-1)}^{|E(\match{k})|}
    \cdot \sum_{M \in \matchSet{k}} {(-1)}^{| \flip{T}{\MTour{M}{T}} |}.\]
    Lastly, by \cref{lem:card:matching}, the cardinality of $\matchSet{k}$ is odd, which
    means that $\sum_{M \in \matchSet{k}} {(-1)}^{| \flip{T}{\MTour{M}{T}} |}$ is an alternating
    sum with an odd number of terms and therefore odd, too. This proves that $X$ is odd, and in particular
    $X \neq 0$, which implies $\ae{T}{\cmatch{k}} \neq 0$.
\end{proof}

We now have everything to prove our reduction from $\CPsubProb(\{\cmatch{k}\})$ to $\tourProb(\{T\})$.

\sevenfifteentwo
\begin{proof}
    Assume that there is
    an algorithm that reads the whole input and computes $\tourProb(\{T\})$
    for any tournament of order $n$ in time $O(n^\gamma)$. Now, \cref{lem:subs:to:color}
    shows that there is an algorithm that reads the whole input\footnote{This is implicitly given
    since by assumption $\gamma \geq 2$.} and computes $\CFsubProb(\{T\})$ in time $O(n^\gamma)$.
    Since \cref{theo:cmatch:not:zero} implies that $\ae{T}{\cmatch{k}} \neq 0$,
    we continue with \cref{cor:ae:reduction} which yields an algorithm that
    computes $\CPsubProb(\{\cmatch{k}\})$ for $k$-colored graphs of order $n$ in time $O(n^{\gamma})$.
\end{proof}

\subsection{Showing that
\texorpdfstring{$\CPsubProb(\{\cmatch{k}\})$}{\#cp-Sub(Mk)} is hard}

To finish our hardness results, we show that
$\CPsubProb(\{\cmatch{k}\})$ is hard to solve. To this end, we show that
$\cmatch{k}$ has large treewidth and large clique minors.

\begin{lemma}
    \dglabel{lem:matching:tw}(Anti-matchings have unbounded treewidth)
    For every integer $k\geq 2$, the graph $\cmatch{k}$ has treewidth $k-2$.
\end{lemma}
\begin{proof}
    First note that $K_k$ is the unique graph of order $k$ and treewidth $k-1$
    (if $\{u,v\}$ is a non-edge of $G$, then there is a trivial tree-decomposition
    of $G$ consisting of two bags $V(G)\setminus \{u\}$ and $V(G)\setminus \{v\}$
    respectively, the width of which is $|V(G)|-2$),
    hence $\tw(\cmatch{k}) \leq k-2$. On the other side, the minimum degree of
    $\cmatch{k}$ is equal to $k-2$ which implies $\tw(\cmatch{k}) \geq k-2$
    (see~\cite[Lemma 4]{BODLAENDER20111103}).
\end{proof}

\begin{lemma}
    \dglabel{lem:cmatch:clique:minor}(The anti-matchings \(\cmatch{k}\) has clique-minor
    $K_{\lfloor 3k/4 \rfloor}$)
    For every integer $k \geq 2$, the graph $K_{\lfloor 3k/4 \rfloor}$ is a minor of $\cmatch{k}$.
    Further, if $k$ is odd then $K_{1 + \lfloor 3(k-1)/4 \rfloor}$ is a minor of $\cmatch{k}$.
\end{lemma}
\begin{proof}
    In the following, we define $\eta(H)$ to be the
    size of the largest clique-minor of $H$. We start by considering the case that $k$ is even.
    Let $P_k$ be a path with $k$ vertices (i.e., $E(P_k) = \{\{1, 2\}, \{2, 3\}, \dots, \{k-1, k\} \}$).
    By~\cite[Lemma 4.3]{DW22}, $\eta(\overline{P_k}) = \lfloor(k + \omega(\overline{P_k}))/2\rfloor$,
    where $\omega(\overline{P_k})$ is the size of the largest clique in $\overline{P_k}$.
    Observe that the vertices $\{1, 3, 5, \dots\}$ form a clique in $\overline{P_k}$ of size
    $\lceil k/2 \rceil$. Hence,
    \[\eta(\overline{P_k}) \geq \left\lfloor \frac{k + \lceil \frac{k}{2} \rceil}{2} \right\rfloor \geq \left\lfloor \frac{3k}{4} \right\rfloor.\]
    Now, the result follows from $\eta(\cmatch{k}) \geq \eta(\overline{P_k})$ since $\overline{P_k}$
    is an edge-subgraph of $\cmatch{k}$ and adding more edges only increase the size of
    the largest clique-minor. Thus, $\eta(\cmatch{k}) \geq \lfloor 3k/4 \rfloor$.

    If $k$ is odd then $\cmatch{k}$ is composed of an apex $x$ that is attached to a graph
    that is isomorphic to $\cmatch{k-1}$.
    By using the above construction on the $\cmatch{k-1}$-part, we obtain
    $\eta(\cmatch{k}) \geq 1 + \lfloor 3(k-1)/4 \rfloor \geq \lfloor 3k/4 \rfloor$.
\end{proof}

Having large clique-minors is important since there
is a reduction from $\CPsubProb(\{H'\})$ to $\CPsubProb(\{H\})$
whenever $H'$ is a minor of $H$ (see \cref{lem:cmatch:clique:minor}).
With this, we obtain our hardness result for $\CPsubProb(\{\cmatch{k}\})$.

\sevenfifteenseven
\begin{proof}
    By \cref{lem:cmatch:clique:minor}, $\cmatch{k}$ contains $K_{\lfloor 3k/4 \rfloor}$
    as a minor. Thus, \cref{lem:cpsub:minor} shows that there is an algorithm
    that computes $\CPsubProb(\{K_{\lfloor 3k/4 \rfloor}\})$ in time $O(n^\gamma)$.
    Next, \cref{lem:cp:cf} shows the existence of an algorithm that computes
    $\GraphCFsubProb(\{K_{\lfloor 3k/4 \rfloor}\}) = \CFcliqueProb_{\lfloor 3k/4 \rfloor}$
    in time $O(n^\gamma)$.
    Lastly, \cref{lem:CFclique:clique} implies an algorithm that
    solves $\cliqueProb_{\lfloor 3k/4 \rfloor}$ for any graph of order $n$
    in time $O(n^{\gamma})$.
\end{proof}

\subsection{Main Hardness Results for Counting Tournaments}\label{sec:counting:main:results}

We now prove our hardness results for $\tourProb(\{T\})$.

\thmsub
\begin{proof}
    Let $T$ be a tournament of order $k$ such that there is
    an algorithm that reads the whole input and computes $\tourProb(\{T\})$
    for any tournament of order $n$ in time $O(n^\gamma)$.
    By combining \cref{maintheorem:cfsub,lem:cmatch:to:clique}, we
    obtain an algorithm that computes $\cliqueProb_{\lfloor 3k/4 \rfloor}$ for graphs of order $n$
    in time $O(n^{\gamma})$.

    Assuming ETH, there is a constant $\alpha > 0$
    such that no algorithm computes $\cliqueProb_{\lfloor 3k/4 \rfloor}$
    in time $O(n^{\alpha \cdot 3k/4})$ (see \cref{lemma:eth:clique}).
    If we set $\beta = \alpha \cdot 3/4$ then this implies that no algorithm computes $\tourProb(\{T\})$
    for graphs of order $n$ in time $O(n^{\beta k})$.
\end{proof}

\subsection{Further Implications of Our Approach}

In this section we fill in some of the details left open in \cref{sec:techov}.

\begin{definition}[The pied graph $\GraphT{G}{T}$ of a labeled tournament \(T\) and a
    colored tournament \(G\)]\dglabel{def:GraphT}
    Let $T$ be a $k$-labeled tournament and $G$ be a $k$-colored tournament with coloring
    $c \colon V(G) \to \setn{k}$.

    The \emph{pied graph}%
    \footnote{``pied'' as in ``thrown into disorder'' \emph{and} ``pied'' as in ``of two
    or more colors''.}
    $\GraphT{G}{T}$ of \(T\) and \(G\) is the $k$-colored
    graph with vertex-set $V(G)$, coloring $c$, and edges $\{x, y\} \in E(\GraphT{G}{T})$
    if and only if
        $c(x) \neq c(y)$ and
        $G$ and $T$ have the same orientation on $\{x, y\}$ and $\{c(x),c(y)\}$.
\end{definition}

Next, we use pied graphs to compute $\CFsubs{T}{G}$ by using a linear combination of $\CPsubProb$-counts.

\sevenfifteeneight
\begin{proof}
    By \cref{theo:tour:ae}, we obtain
    \[\CFsubs{T}{\tourG{(\GraphT{G}{T})}{T}} = \sum_{H \in \graphs{k}} \ae{T}{H} \cdot \CPsubs{H}{\GraphT{G}{T}}.\]
    We show $\CFsubs{T}{G} = \CFsubs{T}{\tourG{(\GraphT{G}{T})}{T}}$. Observe that the tournaments $G$ and $\tourG{(\GraphT{G}{T})}{T}$ have the same
    vertex set and same coloring $c$. It is therefore enough to show that they
    have the same orientation on all edges $\{x, y\}$ with $c(x) \neq c(y)$.
    By \cref{def:GraphT}, $\{x, y\} \in E(\GraphT{G}{T})$ if
    and only if $G$ and $T$ have the same orientation on $\{x, y\}$ and $\{c(x),c(y)\}$.
    Further, by \cref{def:TourG}, $\tourG{(\GraphT{G}{T})}{T}$  and $T$
    have the same orientation on $\{x, y\}$ and $\{c(x),c(y)\}$ if and only if $\{x, y\} \in E(\GraphT{G}{T})$.

    Combining these two statements yields that $\tourG{(\GraphT{G}{T})}{T}$ and $G$ have the
    same orientation on $\{x, y\}$. Hence, $\CFsubs{T}{G} = \CFsubs{T}{\tourG{(\GraphT{G}{T})}{T}} $,
    proving the first part of the theorem.

    For the second part, first without loss of generality, we assume $\gamma \geq 2$.
    Let $G$ be a $k$-colored tournament of order $n$.
    Observe that we can compute the coefficients $\ae{T}{H}$ for a graph $H$ in time $O(g(k))$,
    where $g$ is a computable function. Further, the graph $\GraphT{G}{T}$, can be computed in
    time $O(n^2)$. Thus, we can compute $\CFsubs{T}{G}$ in time $O(2^{\binom{k}{2}} \cdot g(k) \cdot n^\gamma)$
    by evaluating $\sum_{H \in \graphs{k}} \ae{T}{H} \cdot \CPsubs{H}{\GraphT{G}{T}}$.
    This yields an algorithm that computes $\CFsubProb(\{T\})$ in time $O(n^\gamma)$ since $k$ is fixed.
\end{proof}

\theocfsub
\begin{proof}
    The statement immediately follows from \cref{theo:GraphT,cor:ae:reduction}.
\end{proof}

\begin{remark}
    The following is a reformulation of \cref{theo:cfsub:hardest:term}.
    For all $k$-labeled tournaments $T$, we have
    \[\Exp{\CFsubProb(\{T\})} = \max_{\substack{H \in \graphs{k}, \ae{T}{H} \neq 0}} \Exp{\CPsubProb(\{H\})}.\]
\end{remark}

\begin{remark}\label{remark:yuster:limits}
    By \cref{theo:cfsub:hardest:term}, the problem $\CFsubProb(\{T\})$ is exactly as hard
    as the linear combination of $\CPsubProb$-counts from \cref{theo:tour:ae}.
    Thus, our approach only transforms $\CFsubProb(\{T\})$ into a different
    problem with the same complexity. This way, we obtain
    very precise complexity results.

    In contrast, the
    previous approach by Yuster~\cite{Yuster25} transforms $\CFsubProb(\{T\})$ into an easier
    problem. For example, our approach shows
        $\Exp{\CFsubProb(\{\tranTour_k\})} \geq \Exp{\cliqueProb_{\lfloor 3k/4\rfloor}}$.
        However, Yuster's approach only shows
        $\Exp{\CFsubProb(\{\tranTour_k\})} \geq \Exp{\cliqueProb_{\lceil k/2\rceil}}$,
    since $\sig(\tranTour_{k}) = \lfloor k/2\rfloor$ (see~\cite[Lemma 2.5]{Yuster25}).\footnote{
        To see $\sig(\tranTour_{k}) = \lfloor k/2\rfloor$ observe that $\{2, 4, 6, \dots\}$
        is a signature of $\tranTour_{k}$ of size $\lfloor k/2\rfloor$.  Hence,
        $\sig(\tranTour_{k}) \leq \lfloor k/2\rfloor$.
        In contrast, if $R$ is a set of vertices such that for some $v \in \setn{k-1}$
        we have $\{v, v+1\} \cap R = \emptyset$ then we immediately obtain that
        $R$ is not a signature of $\tranTour_k$ since we can flip the edge $\{v, v+1\}$ to obtain
        an isomorphic tournament. Thus, each signature of $\tranTour_{k}$
        contains at least $\lfloor k/2\rfloor$ many vertices and therefore
        $\sig(\tranTour_{k}) \geq \lfloor k/2\rfloor$.
    }
\end{remark}

By \cref{theo:cfsub:hardest:term}, the problem $\tourProb(\{T\})$
is exactly as hard as the hardest $\CPsubProb(\{H\})$ term with $\ae{T}{H} \neq 0$.
Next, we argue why the anti-matching $\cmatch{k}$ is a
good candidate for the hardest $\CPsubProb(\{H\})$ term with $\ae{T}{H} \neq 0$.
(Recall that $\ae{T}{\cmatch{k}} \neq 0$ due to \cref{theo:cmatch:not:zero}).
We start by proving that $\ae{T}{H}$ is vanishing if $H$ has two apices
(an apex is vertex $v \in H$ that is adjacent to all other vertices in $H$).

\begin{lemma}
    \dglabel{lem:ae:apex}(Graphs with two apices have a vanishing alternating enumerator
    $\ae{T}{H} = 0$)
    Let $T$ be a $k$-labeled tournament and $H$ be a $k$-labeled graph with at least two apices, then
    $\ae{T}{H} = 0$.
\end{lemma}
\begin{proof}
    For a fixed tournament $T$ and a graph $H$ with two apices, we define a function
    $f \colon \graphs{k} \to \graphs{k}$ with the following three properties for all $F \subseteq H$:
    \begin{enumerate}
        \item $f(F) \subseteq H$ and $f\circ f(F) \equiv F$,
        \item $T_{F} \cong T_{f(F)}$, and
        \item $|E(F)| \equiv_2 |E(f(F))| + 1 $.
    \end{enumerate}
    If such a function exists, then we can partition
    the set $\mathcal{H} = \{H' \subseteq H \colon T_{H'} \cong T\}$ into pairs $\{F, f(F)\}$.
    For a system of representatives  $\tilde{\mathcal{H}}$
    (that is, $\mathcal{H} = \tilde{\mathcal{H}} \uplus \{f(F) : F \in \tilde{\mathcal{H}}  \}$), we obtain
    \[
    \ae{T}{H} \coloneqq {(-1)}^{|E(H)|} \sum_{H' \subseteq H} {(-1)}^{|E(H')|} \left[T_{H'} \cong T\right]
     = {(-1)}^{|E(H)|}\sum_{F \in \tilde{\mathcal{H}}} \Big({(-1)}^{|E(F)|} + {(-1)}^{|E(f(F))|} \Big)
     =  0,\]
     where we use that ${(-1)}^{|E(F)|} + {(-1)}^{|E(f(F))|} = 0$ due to property 3.

    It is therefore enough to show that such a function $f$ exists.
    To this end, let $u$ and $v$ be two apices in $H$.
    Further, let $\psi = (u, v)$ be the permutation that permutes $u$ with $v$.
    We define $f$ in the following way: given a $k$-labeled graph $F$, define $f(F)$ as the $k$-labeled graph with edge set
    $E(K_k) \setminus (\flip{T}{(T_F)^\psi})$.

    First, note that $f(F)$ is the unique graph with
    \begin{align}\label{eq:T:fF}
        T_{f(F)} \equiv (T_F)^\psi,
    \end{align}
    since the non-edges of $f(F)$ are exactly the edges on which $T$ and $(T_F)^\psi$ disagree.
    This shows that $T_{F} \cong T_{f(F)}$.
    Further, \cref{eq:T:fF} shows that $f\circ f(F) = F$
    since $\psi \circ \psi = \id$. To see that $f(F) \subseteq H$, note that \cref{eq:T:fF}
    also implies that $F$ and $f(F)$ are equal on all edges $\{x, y\}$ with $x, y \notin \{u, v\}$
    since $\psi$ only changes edges adjacent to $u$ or $v$. Since $F \subseteq H$, this
    immediately yields that all edges of $f(F)$ that are non-adjacent to $u$ or $v$ are also
    in $H$. Lastly, since $u$ and $v$ are both apices, we also obtain that all edges
    adjacent to $u$ or $v$ are in $H$ which yields $f(F) \subseteq H$.

    Next, we show that $|E(F)| \equiv_2 |E(f(F))| + 1$. We show that
    $|\symdif{E(F)}{E(f(F))}|$ is odd, which yields that we have to change an odd number of
    edges to transform $F$ into $f(F)$. To this end, we define the sets $A = \{u,v\}= \{\psi(u), \psi(v)\}$
    and $B = V(T) \setminus A$. This allows us to partition the set
    $\symdif{E(F)}{E(F(f))}$ into three sets $\outE$, $\inE$, and $\intE$, where
    $\outE$ contains all edges that start and end in $B$,  $\inE$ contains all edges that start and
    end in $A$, and $\intE$ contains all edges that between $A$ and $B$.
    To prove the statement, we show
    $|\outE| \equiv_2 0$, $|\inE| \equiv_2 1$, and $|\intE| \equiv_2 0$.

    To this end, we start with $|\outE| \equiv_2 0$. Due to \cref{eq:T:fF},
    $F$ and $f(F)$ are identical on $B$, thus $\outE = \emptyset$.
    Next, note that $\inE = \{\{u, v\}\}$ since $T_{F}$ and $T_{f(F)}$ have a different
    orientation on $\{u, v\}$ due to \cref{eq:T:fF}.

    Lastly, we us show $|\intE| \equiv_2 0$. For all $z \in B$,
    write $B_z = \{\{\psi(u), z\}, \{\psi(v), z\}\}$. We show that $|B_z \cap \intE|$ is even.
    By \cref{eq:T:fF}, when going from $T_F$ to $T_{f(F)}$,
    we swap the edge $\{\psi(u), z\}$ in $T_F$ with the edge
    $\{\psi(v), z\}$ in $T_F$. Assume that $\{\psi(u), z\}$ and $\{\psi(v), z\}$ have the same
    orientation in $T_F$,\footnote{i.e., if either ($(\psi(u), z) \in \DE{T_F}$ and
    $(\psi(v), z) \in \DE{T_F}$) or ($(z, \psi(u)) \in \DE{T_F}$ and
    $(z, \psi(v)) \in \DE{T_F}$). }
    then they also have the same orientation in
    $T_{f(F)}$. Now, since $T_F$ and $T_{f(F)}$ agree on these edges, we obtain that
    $B_z \cap E(F) = B_z \cap E(f(F))$. Hence, $B_z \cap \intE = \emptyset$.

    If both edges have the opposite orientation on $T_F$,\footnote{i.e.
    if either ($(\psi(u), z) \in \DE{T_F}$ and
    $(z, \psi(v)) \in \DE{T_F}$) or ($(z, \psi(u)) \in \DE{T_F}$ and
    $(\psi(v), z) \in \DE{T_F}$). } then $T_F$ and $T_{f(F)}$
    disagree on $\{\psi(u), z\}$, and $T_F$ and $T_{f(F)}$
    disagree on $\{\psi(v), z\}$. This implies for $\{a, b\} \in B_z$
    that $\{a, b\} \in E(F)$ if and only if $\{a, b\} \notin E(f(F))$. Hence,
    $B_z \subseteq \intE$. In both cases we obtain that $|B_z \cap \intE|$ is even. This yields
    \begin{align*}
        |\intE| &= \left|\biguplus_{z \in B} \left(B_z \cap \intE\right) \right|
        \equiv_2 \sum_{z \in B} 0 \equiv_2 0
    \end{align*}
     The result now follows from
    $|\symdif{E(F)}{E(f(F))}| = |\outE| + |\intE| + |\inE| \equiv_2 0 + 0 + 1$.
\end{proof}

\begin{corollary}
    \dglabel{lem:ae:clqiue}[lem:ae:apex](The alternating enumerator vanishes for cliques)
    Let $T$ be a $k$-labeled tournament with $k \geq 2$ then $\ae{T}{K_k} = 0$.
\end{corollary}
\begin{proof}
    Since $K_k$ has two apices,
    the claim directly follows from \cref{lem:ae:apex}.
\end{proof}

Note that $\cmatch{k}$ is \emph{just below} of having two apices,
meaning that $\cmatch{k}$ would obtain two apices if we add a single edge to it.
We use this to show that $\cmatch{k}$ is the densest graph with $\ae{T}{H} \neq 0$.

\sevenfifteennine
\begin{proof}
    Suppose that $|E(H)| > |E(\cmatch{k})|$. This is equivalent to $|E(\overline{H})| < |E(\match{k})| = \lfloor k/2 \rfloor$.
    Hence there are at least two isolated vertices in $\overline{H}$ implying that $H$
    has at least two apices, and therefore $\ae{T}{H} = 0$ by \cref{lem:ae:apex}.
\end{proof}

\begin{remark}
    \dglabel^{7-15-10}[lem:matching:tw]
    By~\cite{DBLP:phd/dnb/Curticapean15,DBLP:journals/toc/Marx10,CDNW24}, the problem $\CPsubProb(\{H\})$
    is harder to solve for graphs $H$ with high treewidth. Further,
    \cref{lem:cpsub:minor} shows that $\CPsubProb(\{H\})$ is at least as hard
    as $\CPsubProb(\{H'\})$ for all edge-subgraphs $H'$ of $H$. Hence,
    the problem $\CPsubProb(\{H\})$ also becomes harder to solve for denser graphs.\footnote{Also note
    that graphs with more edges tend to have a higher treewidth.} Now,
    \cref{cor:cmatch:densest} shows that $\cmatch{k}$ is the densest
    graph with $\ae{T}{H} \neq 0$, which makes $\CPsubProb(\{\cmatch{k}\})$ a good candidate for
    the complexity of $\CFsubProb(\{T\})$. Lastly, by \cref{lem:matching:tw}
    we obtain $\tw(\cmatch{k}) = k - 2$ which is also the highest possible
    treewidth for a graph with non-vanishing alternating enumerator.
    To see this, observe that $\tw(H) = k - 1$ is only possible if $H \equiv K_k$
    and $\ae{T}{K_k}$ is zero due to \cref{lem:ae:clqiue}.
\end{remark}

\clearpage
\section{The Complexity of Finding Tournaments}

In this section, we study the complexity of finding a fixed tournament $T$
inside an input tournament $T'$ (that is, deciding if $T$ is isomorphic to a subtournament
of $T'$).

\subsection{Easy Cases for Finding Tournaments}

By \cref{thm:directed_ramsey}, every tournament $T'$ of order at least $2^{k-1}$
contains a subtournament that is isomorphic to $\tranTour_k$. This immediately
yields that $\DecTourProb(\{\tranTour_k\})$ is easy to compute since we may always return
true for large enough input tournaments. In the following,
we use this observation to find other tournaments $T$ for which $\DecTourProb(\{T\})$
is also easy to solve. To this end, we split $T$ into two parts. A small
part that we may find via brute force (that is, iterate over all possibilities inside $T'$)
and a large remaining part that is isomorphic to a transitive
tournament and is therefore easy to find.

\begin{figure}[t]
    \centering
    \includegraphics[width=.45\textwidth]{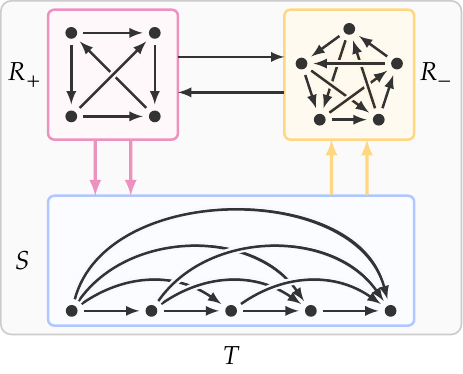}
    \caption{A spine decomposition of a tournament $T$.
        The spine $S$ forms a transitive tournament. All vertices of the $R_+$-part have
        outgoing edges toward the $S$-part and all vertices of the $R_-$-part have ingoing
        edges from the $S$-part. Edges inside $R_+ \uplus R_-$ may be oriented arbitrarily.}%
    \label{fig:neighborhood:decomposition2}
\end{figure}

\defcore

Consult \cref{fig:neighborhood:decomposition2} for a visualization of a spine decomposition.

We now show that $\DecTourProb(\{T\})$ is easy whenever $\core{T}$ is large.

\thmdecsubeasy
\begin{proof}
    Let $T$ be a tournament of order $k$. We start by computing a spine decomposition
    $(R_+, R_-, S)$ of $T$ with $c \coloneqq |R_+| + |R_c| = |V(T)| - \core{T}$.
    Observe that this can be done by iterating over all partitions $(R_+, R_-, S)$ and
    checking which of them form a spine decomposition.
    Thus, we can find $(R_+, R_-, S)$ in time $O(g(k))$ where $g$ is some compute function.
    Further, let $R_+ = \{u_1, \dots, u_a\}$ and
    $R_- = \{w_1, \dots, w_b\}$, where $a \coloneqq |R_+|$ and $b \coloneqq |R_-|$.
    We show in the following
    that $\DecTourProb(\{T\})$ can be solved in time $O(n^{c+2})$ for input tournaments of order $n$.

    For a tournament $T'$, we start by iterating through all tuples
    $(\hat{u}_1, \dots, \hat{u}_a) \in {V(T')}^a$ and
    $(\hat{w}_1, \dots, \hat{w}_b) \in {V(T')}^b$. We write
    $R'_+ \coloneqq \{\hat{u}_1, \dots, \hat{u}_a\}$,
    $R'_- \coloneqq \{\hat{w}_1, \dots, \hat{w}_b\}$ and
    check if $\varphi(u_i) = \hat{u}_i$, $\varphi(w_i) = \hat{w}_i$  defines
    an isomorphism from $\indGraph{T}{R_+ \cup R_-}$ to $\indGraph{T'}{R'_+ \cup R'_-}$. If this is the
    case, we compute
    \[N' \coloneqq \left(\bigcap_{v \in R'_+} N^{+}_{T'}(v) \right) \; \cap \;
    \left(\bigcap_{v \in R'_-} N^{-}_{T'}(v) \right).\]
    If $|N'| \geq 2^{k - c - 1}$, return true. Otherwise, check if $\indGraph{T'}{N'}$ contains
    $\tranTour_{\core{T}}$ as a subtournament via a brute-force algorithm. If this is the case return true,
    otherwise continue with the pair of tuples. Finally, return
    false after checking all pair of tuples.

    Observe that checking if $\varphi$ defines an isomorphism from
    $\indGraph{T}{R_+ \cup R_-}$ to $\indGraph{T'}{R'_+ \cup R'_-}$ can be done time $O(c^2)$. Further, the set
    $N'$ can be computed in time $O(n \cdot c)$. Finally, if $|N'| <  2^{k - c - 1}$,
    then checking if $\indGraph{T'}{N'}$ contains
    $\tranTour_{k - c}$ is in time $O(g'(k))$ for some computable function $g'$.
    Since there are $n^c$ many pair of tuples, the above algorithm runs in time
    $O(g(k) + n^c \cdot (c^2 + g'(k) \cdot n \cdot c))$ which is in $O(f(k) \cdot n^{c+2})$ for
    some computable function $f$.~\footnote{It is actually in time $O(f(k) \cdot n^{c+1})$. However,
    we still want that our algorithm reads the whole input if $c = 0$.}
    If $k$ is fixed then this is
    in $O(n^{c+2})$.

    To prove the correctness, we start by assuming that our algorithm returns
    true on input $T'$. Note that this can only happen if there are vertices
    $(\hat{u}_1, \dots, \hat{u}_a) \in {V(T')}^a$
    and $(\hat{w}_1, \dots, \hat{w}_b) \in {V(T')}^b$
    such that $\varphi(u_i) = \hat{u}_i$, $\varphi(w_i) = \hat{w}_i$ defines
    an isomorphism from $\indGraph{T}{R_+ \cup R_-}$ to $\indGraph{T'}{R'_+ \cup R'_-}$ and
    $N'$ either contains at least $2^{k - c - 1}$ many vertices
    or $\indGraph{T'}{N'}$ contains  $\tranTour_{\core{T}}$ as a subtournament. If
    $|N'| \geq 2^{k - c - 1}$ then $\indGraph{T'}{N'}$ contains $\tranTour_{\core{T}}$
    due to \cref{thm:directed_ramsey}. Let $M' \subseteq N'$ such that $\indGraph{T'}{M'}$
    is isomorphic to  $\tranTour_{\core{T}}$. Since $(R_+, R_-, S)$ is a spine decomposition
     of $T$, we obtain that $\indGraph{T'}{R'_+ \uplus R'_- \uplus M'}$
    is isomorphic to $T$.

    In contrast, let $A \subseteq V(T')$ be set of vertices such that
    there is an isomorphism $\varphi$
    from $\indGraph{T'}{A}$ to $T$. Set $(\hat{u}_1 = \varphi^{-1}(u_1), \dots,
    \hat{u}_a = \varphi^{-1}(u_a)) \in {V(T')}^a$, $(\hat{w}_1 = \varphi^{-1}(w_1), \dots,
    \hat{w}_b = \varphi^{-1}(w_b)) \in {V(T')}^b$, $R'_+ \coloneqq \{\hat{u}_1, \dots, \hat{u}_a\}$,
    and $R'_- \coloneqq \{\hat{w}_1, \dots, \hat{w}_b\}$
    then  $\varphi(u_i) = \hat{u}_i$, $\varphi(w_i) = \hat{w}_i$  defines
    an isomorphism from $\indGraph{T}{R_+ \cup R_-}$ to $\indGraph{T'}{R'_+ \cup R'_-}$. Since $(R_+, R_-, S)$ is a
    spine decomposition of $T$,
    we obtain that $\indGraph{T'}{N'}$
    contains $\tranTour_{\core{T}}$ as a subtournament. Observe that our algorithm
    successively detects if this is the case since it iterates through all possible tuples.

    Lastly, let $\mathcal{T}$ be a set of tournaments such that there
    is a constant $c$ with $|V(T)| - \core{T} \leq c$ for all $T \in \mathcal{T}$, then
    we can use the algorithm from above to compute $\DecTourProb(\mathcal{T})$
    in time $O(f(k) \cdot n^{c+2})$, proving that $\DecTourProb(\mathcal{T})$ is FPT.
\end{proof}

\begin{remark}
    \Cref{theo:core:easy} shows that $\DecTourProb(\{T\})$ is easy for some tournaments $T$
    that are close to being transitive. On the other side, there are tournaments that
    are close to being transitive for which \cref{theo:core:easy} fails. For example,
    let $F_k$ the tournament obtained from $\tranTour_k$ by flipping the edge
    $\{\lfloor k/2\rfloor - 1, \lfloor k/2\rfloor + 1 \}$.
    Now, $F_k$ is very close to being transitive but $|V(F_k)| - \core{F_k} \geq \lfloor k/2\rfloor - 1$.
\end{remark}

\subsection{Analyzing Tournaments that Have a Large TT-unique Partition}

From the last section, we know that $\DecTourProb(\{T\})$ is easy for particular
tournaments that are close to being transitive. However, almost all tournaments are far away
from being transitive. Hence, in this section we show that $\DecTourProb(\{T\})$
is almost surely hard for a random tournament $T$ of order $k$. Here,
hard means that we can use $\DecTourProb(\{T\})$ to solve
$\DecCliqueProb_{\lfloor k/(9\log(k)) \rfloor}$.

\TTunique

Given a TT-unique partition of $T$ and an input tournament $G$,
we now show how to construct a tournament $G^\ast$ such that $G^\ast$ has a colorful clique
if and only if $T$ is isomorphic to subtournament of $G$.

\begin{figure}[tp]
    \centering
    \includegraphics{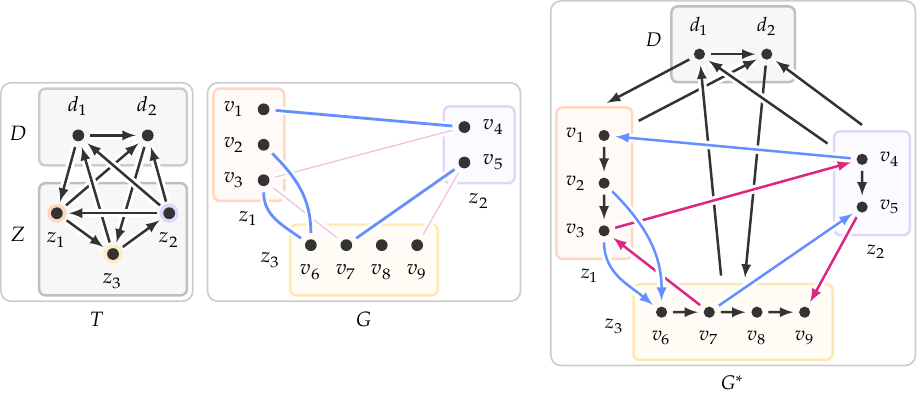}
    \caption{A tournament $T$ with a TT-unique partition $(\dSet, \zSet)$, a
        graph $G$, and the corresponding tournament \(G^\ast\).
        Edges of \(G\) are depicted in blue (and thick); non-edges of \(G\) are depicted
        in red (and thin); we depict only a subset of the (non-)edges.
        Further, in \(G^\ast\), each of the components (\(\zCard_1, \zCard_2, \zCard_3,\) and \(\dSet\))
        induce a transitive tournament.
        In \(G^\ast\), blue arcs between vertices of components correspond
        to edges in \(G\) and thus have the same orientation as the corresponding edge of
        the tournament \(T\).
        In \(G^\ast\), red arcs between vertices of components correspond
        to non-edges in \(G\) and thus have the opposite orientation as the corresponding edge of
        the tournament \(T\).}%
    \label{fig:dec:reduction}
\end{figure}

\ttReduction
\begin{proof}
	Let $k$ be the order of $T$.
    Without loss of generality, we assume
    $\zSet = [\zCard]$  and $\dSet = \{\zCard+1, \dots, \zCard + |\dSet|\}$.
    We write $\dCard \coloneqq |\dSet|$, $V(G) = \{v_{1}, \dots, v_{n}\}$ for the vertices of $G$,
    and $c \colon V(G) \to [\zSet]$ for the coloring of $G$.

    We construct $G^\ast$ in the following way. The vertex set of $G^\ast$ is
    $V(G^\ast) \coloneqq V(G) \uplus \dSet^\ast$, where $\dSet^\ast \coloneqq \{v_{n+1}, \dots, v_{n+\dCard}\}$
    is a set of $\dCard$ new vertices. We define a coloring $c^\ast \colon V(G^\ast) \to [\dCard+\zCard]$ via
    \[c^\ast(v_i) = \begin{cases}
        c(v_i), \quad &\text{if $i \leq n$} \\
        (i - n) + \zCard, \quad &\text{if $i > n$}
    \end{cases}.\]
    Observe that $c^\ast$ is equal to $c$ on all vertices in $V(G)$ and
    that it maps $v_{n+i} \in \dSet^\ast$ to $\zCard+i \in \dSet$. Even though $c^\ast$ defines a coloring on $G^\ast$,
    we still consider $G^\ast$ to be an uncolored tournament.
    The orientation of $\{v_i, v_j\}$ in $G^\ast$ is defined in the following way:
    \begin{itemize}
        \item If $i, j \in V(G)$ and $c(v_i) = c(v_j)$ then $(v_i, v_j) \in \DE{G^\ast}$
        if and only if $i < j$. Note that $\indGraph{G^\ast}{c^{-1}(i)}$ is a transitive tournament.
        \item If $i, j \in V(G)$, $c(v_i) \neq c(v_j)$, $\{v_i, v_j\} \in E(G)$ then
        $(v_i, v_j)\in \DE{G^\ast}$ if and only if $(c^\ast(v_i),c^\ast(v_j))\in \DE{T}$.
        \item If $i, j \in V(G)$, $c(v_i) \neq c(v_j)$, $\{v_i, v_j\} \notin E(G)$
        then $(v_i, v_j)\in \DE{G^\ast}$ if and only if $(c^\ast(v_j),c^\ast(v_i))\in \DE{T}$.
        \item If $v_i \in \dSet^\ast$ or $v_j \in \dSet^\ast$ then
        $(v_i, v_j)\in \DE{G^\ast}$ if and only if $(c^\ast(v_i), c^\ast(v_j))\in \DE{T}$.
        Note that $\indGraph{G^\ast}{\dSet^\ast} \cong \indGraph{T}{\dSet}$.
    \end{itemize}
    See \cref{fig:dec:reduction} for an example of $G^\ast$. Note that $G^\ast$
    can be constructed in time $O({(n + \dCard)}^2)$. It remains to show that
    $T$ is isomorphic to a subtournament of $G^\ast$ if and only if $G$ contains a colorful $\zCard$-clique.

    First, assume that there is a set of vertices $A \subseteq V(G)$ such that
    $\indGraph{G}{A}$ is a colorful $\zCard$-clique. We show that $\indGraph{G^\ast}{A \uplus \dSet^\ast}$
    is isomorphic to $T$ via $c^\ast$. Let $v_i, v_j \in A \uplus \dSet^\ast$. If $v_i \in \dSet^\ast$
    or $v_j \in \dSet^\ast$ then by construction $(v_i, v_j)\in \DE{G^\ast}$ if and only if
    $(c^\ast(v_i), c^\ast(v_j))\in \DE{T}$. In contrast, if $v_i, v_j \in V(G)$
    then $c(v_i) \neq c(v_j)$ since $A$ is colorful. Furthermore, since $\{v_i, v_j\} \in E(G)$
    we obtain that $(v_i, v_j)\in \DE{G^\ast}$ if and only if $(c^\ast(v_i),c^\ast(v_j))\in \DE{T}$.
    Thus, $c^\ast$ is an isomorphism from  $\indGraph{G^\ast}{A \uplus \dSet^\ast}$
    is isomorphic to $T$. This proves that if $V(G)$ contains a colorful $\zCard$-clique then
    there is a subtournament of $G^\ast$ that is isomorphic to $T$.

    Next, let $A^\ast \subseteq V(G^\ast)$ be a set of vertices such that $\indGraph{G}{A^\ast}$
    is isomorphic to $T$ via an isomorphism $\varphi \colon A^\ast \to V(T)$. Observe that $|A^\ast| = k$.
    To show that $G$ contains a
    colorful $\zCard$-clique, we first justify that $\varphi$ is equal to $c^\ast$, which requires multiple steps.
    We start by proving that there are many vertices in $\dSet^\ast \cap A^\ast$
    that get mapped to $\dSet$ via $\varphi$.
    \begin{claim}\label{claim:TT:card}
        Let $\delta \coloneqq \dCard - \TT{T} \cdot \zCard = k-\TT{T}\cdot \zCard - \zCard$, then
        $|\varphi^{-1}(\dSet) \cap \dSet^\ast|\geq \delta$.
    \end{claim}
    \begin{claimproof}
        We first show that
        $|\dSet^\ast \cap A^\ast| \geq k - \TT{T} \cdot \zCard$. Assume for a contradiction
        that there are more than $\TT{T} \cdot \zCard$ vertices $x$ in $A^\ast$
        with $c^\ast(x) \leq \zCard$ (i.e., $x \in V(G)$). By the pigeonhole-principle,
        we obtain that there is an $i \in [\zCard]$ such that  $S_i \coloneqq \{x \in A^\ast : c^\ast(x) = i\}$
        has strictly more than $\TT{T}$ vertices.
        By construction of $G^\ast$, $\indGraph{G^\ast}{S_i}$ is a
        transitive subtournament of $\indGraph{G^\ast}{A^\ast}$ of order at least $\TT{T} + 1$ which
        is a contradiction since  $\indGraph{G^\ast}{A^\ast}$ is isomorphic to $T$ and therefore
        $\TT{\indGraph{G^\ast}{A^\ast}} = \TT{T}$.

        This shows that $ |\dSet^\ast \cap A^\ast | \geq k - \TT{T} \cdot \zCard$. Observe that
        there are at most $\zCard$ elements in $ \dSet^\ast \cap A^\ast$ that get mapped to $\zSet$ via
        $\varphi$ since $\varphi$ is injective. Thus,
        $|\varphi^{-1}(\dSet) \cap \dSet^\ast \cap A^\ast| = |\varphi^{-1}(\dSet) \cap \dSet^\ast| \geq k - \TT{T} \cdot \zCard - \zCard$.
    \end{claimproof}
    In the following, we define $X \coloneqq \varphi^{-1}(\dSet) \subseteq A^\ast$
    and $Y \coloneqq \varphi^{-1}(\zSet) \subseteq A^\ast$.
    Observe that $(X, Y)$ is TT-unique with respect to $\indGraph{G^\ast}{A^\ast}$
    since $\varphi$ is an isomorphism from $\indGraph{G^\ast}{A^\ast}$ to $T$.
    Now, \cref{claim:TT:card} ensures that $X \cap \dSet^\ast$ is a large
    subset of $X$. Using \cref{def:TT:unique}, along the next three claims we
    show that $\varphi(v_i) = c^\ast(v_i)$ for all $v_i \in A^\ast$.
    \begin{claim}\label{claim:TT:dif}
        For all vertices $u, v \in A^\ast$ with $u, v \in V(G)$, we have $c^\ast(v) \neq c^\ast(u)$.
    \end{claim}
    \begin{claimproof}
        Assume otherwise, then there are two vertices $u, v \in V(G)$
        such that $c^\ast(u) = c^\ast(v)$.
        By \cref{claim:TT:card}, the set
        $X' = X \cap \dSet^\ast$ has at least $\delta$ elements.
        By construction of $G^\ast$, we obtain that $u, v \notin X'$
        since $u, v \notin \dSet^\ast$. Further,
        $\Tvec{u}{X'} = \Tvec{v}{X'}$ since all vertices with the same
        color have the same orientation towards vertices of $\dSet^\ast$.
        Hence, the existence of $u$ and $v$ contradictions to the TT-uniqueness
        of $(X, Y)$.
    \end{claimproof}
    Using \cref{claim:TT:dif}, we now show that $\varphi$ and $c^\ast$ coincide on
    all vertices that live in $\dSet^\ast$.
    \begin{claim}\label{claim:TT:R}
        $\dSet^\ast$ is a subset of $A^\ast$ and $\varphi(v_i) = c^\ast(v_i)$ for all $v_i \in \dSet^\ast$.
    \end{claim}
    \begin{claimproof}
        We first show $\dSet^\ast \subseteq A^\ast$. Assume otherwise, then $|\dSet^\ast \cap A^\ast| < \dCard$.
        Since $k = \dCard + \zCard$, we obtain that $|V(G) \cap A^\ast| \geq \zCard+1$. By the pigeonhole-principle, we
        now obtain two vertices
        $u, v \in A^\ast$ with $u, v \in V(G)$ and $c^\ast(v) = c^\ast(u)$,
        a contradiction to  \cref{claim:TT:dif}. Hence, $\dSet^\ast \subseteq A^\ast$.
        
        To show $\varphi(v_i) = c^\ast(v_i)$ for all $v_i \in \dSet^\ast$, we first use that
        ${(c^\ast)}^{-1}$ restricted to $\dSet$ is an isomorphism from
        $\indGraph{T}{\dSet}$ to $\indGraph{G^\ast}{\dSet^\ast}$.
        Further, since $\dSet^\ast \subseteq A^\ast$, we obtain that $\varphi$ restricted to $\dSet^\ast$
        is an isomorphism from $\indGraph{G^\ast}{\dSet^\ast}$ to a subtournament $T'$ of $T$. 
        Thus, $\varphi \circ {(c^\ast)}^{-1}$ defines an isomorphism
        from $\indGraph{T}{\dSet}$ to $T'$. Now, the TT-uniqueness of $(\dSet, \zSet)$
        yields that $T' = \indGraph{T}{\dSet}$ since otherwise $T$ would contain 
        two different isomorphic copies of $\indGraph{T}{\dSet}$. Hence,
        $\varphi \circ {(c^\ast)}^{-1}$ is an automorphism which further implies 
        that $\varphi \circ {(c^\ast)}^{-1} = \id_R$ since $\indGraph{T}{\dSet}$
        has only trivial automorphisms. Because $\varphi$ restricted to $\dSet^\ast$
        and ${(c^\ast)}^{-1}$ restricted to $\dSet$ are both bijections, we obtain
        $\varphi(v_i) = c^\ast(v_i)$ for all $v_i \in \dSet^\ast$.
    \end{claimproof}

    \Cref{claim:TT:R} allows us to define $\zSet^\ast \coloneqq A^\ast \setminus \dSet^\ast$.
    Observe that $\zSet^\ast \subseteq V(G)$ and $|\zSet^\ast| = k - \dCard = \zCard$. In remains
    to show that $\varphi$ and $c^\ast$ coincide on $\zSet^\ast$.
    \begin{claim}\label{claim:TT:phi}
        The set $A^\ast$ is colorful with respect to $c^\ast$
        and $\varphi(v_i) = c^\ast(v_i)$ for all $v_i \in A^\ast$.
    \end{claim}
    \begin{claimproof}
        We start by showing that $A^\ast$ is colorful with respect to $c^\ast$.
        By construction $c^\ast$ is injective on $\dSet^\ast \subseteq A^\ast$. Further, by \cref{claim:TT:dif},
        there are no two vertices $u, v \in A^\ast \cap V(G)$ with $c^\ast(u) = c^\ast(v) \in \zSet$.
        Thus, $c^\ast$ restricted to $A^\ast$ is injective and therefore bijective, showing that
        $A^\ast$ is colorful.

        For the second part of the statement observe that we already
        know that $\varphi(v_i) = c^\ast(v_i)$ for all $v_i \in \dSet^\ast$ due to \cref{claim:TT:R}.
        What remains is to prove that $\varphi(v_i) = c^\ast(v_i)$ for all $v_i \in \zSet^\ast$.
        By construction of $G^\ast$, we obtain for all $v \in \zSet^\ast$, $x \coloneqq c^\ast(v)$,
        and $\dCard \in \dSet^\ast$, that
        \begin{align}\label{eq:vec:R}
            \dCard \in \Tvec{v}{\dSet^\ast} \;\text{ if and only if }\;  c^\ast(\dCard) \in \Tvec{x}{\dSet}.
        \end{align}
        Set $y \coloneqq \varphi(v)$. Since $\varphi$ is an isomorphism and $\varphi(\dCard) = c^\ast(\dCard)$
        for all $\dCard \in \dSet^\ast$, we obtain
        \begin{align}\label{eq:vec:R2}
            \dCard \in \Tvec{v}{\dSet^\ast}  \;\text{ if and only if }\; \varphi(\dCard) \in \Tvec{y}{\dSet}.
        \end{align}
        If $\varphi(v) \neq c^\ast(v)$ for some $v \in \zSet^\ast$ then
        \cref{eq:vec:R} and \cref{eq:vec:R2} imply that $\Tvec{y}{\dSet} = \Tvec{x}{\dSet}$,
        where $x \neq y \in \zSet$. However, observe that this is not possible
        since $(\dSet, \zSet)$ is TT-unique. Hence,
        $\varphi(v) = c^\ast(v)$.
    \end{claimproof}

    Finally, we show that $G$ contains a colorful $\zCard$-clique. To this end,
    we show that $\indGraph{G}{\zSet^\ast}$ is a colorful $\zCard$-clique.
    \Cref{claim:TT:phi} and the fact that $c(v) = c^{\ast}(v)$
    for all $v \in V(G)$  immediately imply
    that $\indGraph{G}{\zSet^\ast}$ is colorful with respect to $c$.
    Further, due to \cref{claim:TT:phi}, we obtain
    that $c^\ast$ is an isomorphism from $\indGraph{G}{A^\ast}$ to $T$.
    This implies that for all $v_i, v_j \in \zSet^\ast$ (i.e., $v_i, v_j \in V(G)$),
    $(v_i, v_j)\in \DE{G^\ast}$ if and only if $(c^\ast(v_i),c^\ast(v_j))\in \DE{T}$.
    By construction of $G^\ast$, this is equivalent to $\{v_i, v_j\} \in E(G)$, proving that
    $\indGraph{G}{\zSet^\ast}$ is a clique.
\end{proof}

\Cref{theo:TT:construct} provides us with a reduction from $\DecCFcliqueProb_{|\zSet|}$
to $\DecTourProb(\{T\})$, which we extend to a reduction that starts from
$\DecCliqueProb_{|\zSet|}$.

\begin{theorem}[Reduction from $\DecCliqueProb_{|\zSet|}$ to $\DecTourProb(\{T\})$
    via TT-unique partition $(\dSet, \zSet)$]\dglabel{theo:TT:reduction}[theo:TT:construct,lem:DecCFclique:DecClique]
    Let $T$ be a tournament and $(\dSet, \zSet)$ be a TT-unique partition of $T$.
    Assume that there is an algorithm that
    solves $\DecTourProb(\{T\})$ for any tournaments of order $n$ in time $O(n^\gamma)$.
    Then there is an algorithm that solves $\DecCliqueProb_{|\zSet|}$
    for any graphs of order $n$ in time $O(n^\gamma)$.
\end{theorem}
\begin{proof}
    Assume that there is an algorithm that reads the whole
    input and solves $\DecTourProb(\{T\})$ for any tournaments of order $n$ in time $O(n^\gamma)$.
    Let $G$ be a $|\zSet|$-colored input graph of order $n$. Due to
    \cref{theo:TT:construct}, we can compute an uncolored tournament $G^\ast$
    of order $n + |\dSet|$ in time $O({(n + |\dSet|)}^2)$ such that $T$ is isomorphic to a subtournament of
    $G^\ast$ if and only if $G$ contains a colorful $|\zSet|$-clique. Thus, we can solve $\DecCFcliqueProb_{|\zSet|}$
    on input $G$
    in time $O(n^\gamma)$ by computing $\DecTours{T}{G^\ast}$ since $|\dSet|$ is a constant.
    Lastly, \cref{lem:DecCFclique:DecClique}
    implies an algorithm that solves $\DecCliqueProb_{|\zSet|}$ in time $O(n^\gamma)$.
\end{proof}

\subsection{\texorpdfstring{$\DecTourProb(\{T\})$}{Dec-IndSub({T})} is Hard for Random Tournaments}

In order to use \cref{theo:TT:reduction}, we have to find graphs
that have a TT-unique partition $(\dSet, \zSet)$ where $\zSet$ is large.
The goal of this section is to show
that random tournaments admit a TT-unique partition $(\dSet, \zSet)$ where $\zSet$ is
large with high probability.

\theoprobTT*

First, we show that the $\indGraph{T}{\dSet}$ has no automorphism 
and that $\indGraph{T}{\dSet}$ appears exactly once in $T$
with high probability. Our proof mostly follows the proof of Lemma 2.3 in~\cite{Yuster25}.

\begin{lemma}[Random tournaments satisfy the first two properties of TT-uniqueness:  
    Random tournaments have a trivial automorphism group and $\indGraph{T}{\dSet}$ appears exactly once]\dglabel{lem:prob:iso}
    Let $T$ be random tournament of order $k \geq 10^5$ with vertex set $\{v_1, \dots, v_k\}$.
    Further, set $\zCard \coloneqq \lfloor k / (9 \log(k)) \rfloor$ and $\dSet \coloneqq \{v_{\zCard+1}, \dots, v_k\}$.
    Then the following event occurs with probability at least  $1 - 1/k^3$:
    The subtournament $\indGraph{T}{\dSet}$ has a trivial automorphism group
    and $\indGraph{T}{\dSet}$ appears exactly once in $T$.
\end{lemma}
\begin{proof}
    Set $\dCard \coloneqq |\dSet|$. Further,
    let $E$ be the event that
    there exist an isomorphism $f$ from $\indGraph{T}{\dSet}$ to a subtournament $T'$ of $T$
    such that $f \neq \id_\dSet$. Note that the event $E$ is equivalent to the event that
    the automorphism group of $\indGraph{T}{\dSet}$ is nontrivial 
    or $\indGraph{T}{\dSet}$ appears more than once in $T$. Therefore,
    it is enough to show that the event $E$ only occurs with probability at most $1/k^3$.

    Since $f$ is not the identity function, it has at least one non-stationary point
    (that is a value $x$ with $f(x) \neq x$). Thus,
    \begin{align*}
        \mathbb{P}[E] &= \mathbb{P}[\exists p \geq 1; f \colon \dSet \to \dSet' : \text{$f$
    has $p$ non-stationary points and is isomorphism from $\indGraph{T}{\dSet}$ to $\indGraph{T}{\dSet'}$ } ] \\
    &\leq \sum_{p = 1}^{\dCard} \mathbb{P}[\exists f \colon \dSet \to \dSet' : \text{$f$
    has $p$ non-stationary points and is isomorphism from $\indGraph{T}{\dSet}$ to $\indGraph{T}{\dSet'}$ } ],
    \end{align*}
    where we use union bound for the last step. Further,
    $\dSet'$ stands for an arbitrary non-fixed subset of $V(T)$ of size $\dCard$. This means that the
    existential quantifier ranges over all possible functions $f$ that map $\dSet$ into some subset $D'$ of size $\dCard$.
    We write $\mathcal{F}_p$ for the set of all functions $f \colon \dSet \to \dSet'$
    with $p$ non-stationary points.
    A union bound yields
    \begin{align}\label{eq:E:isom}
        \mathbb{P}[E] \leq \sum_{p = 1}^{\dCard} \sum_{f \in \mathcal{F}_p }
        \mathbb{P}[\text{$f$ is isomorphism from $\indGraph{T}{\dSet}$ to
        $\indGraph{T}{\dSet'}$ } ].
    \end{align}
    In the following, we show that for a fixed function $f\colon \dSet \to \dSet'$ with
    $p$ non-stationary points the probability that $f$ is an isomorphism
    is sufficiently low. We start by consider the case that $f$ has at most 11
    non-stationary points.
    \begin{claim}\label{claim:less:13}
        Let $f \colon \dSet \to \dSet'$ be a function with $1 \leq p \leq 11$ non-stationary points,
        then \[\mathbb{P}[\text{$f$ is isomorphism from $\indGraph{T}{\dSet}$ to
        $\indGraph{T}{\dSet'}$ } ] \leq \frac{1}{2^{p \dCard/12}} \]
    \end{claim}
    \begin{claimproof}
        Let $v$ be some non-stationary point of $f$ and let $\dSet^\ast = \dSet \setminus \{v, f(v)\}$. For
        $f$ being an isomorphism, we need that, for all $u \in \dSet^\ast$,  $(u, v) \in \DE{T}$ if and only
        if $(f(u), f(v)) \in \DE{T}$. Observe that this requires that $|\dSet^\ast|$ different pairs of edges have the
        same orientation, meaning that the probability of this happening is at most $2^{-|\dSet^\ast|}$.
        Lastly, note that
        \[2^{-|\dSet^\ast|} \leq \frac{1}{2^{\dCard - 2}} \leq \frac{1}{2^{p \dCard/12}},\]
        where we use that for $\dCard \geq k/2 \geq 10^5/2$ we have $11 \dCard/12 \leq \dCard - 2$ and $p/12 \leq 11/12$.
    \end{claimproof}
    Next, we consider the case that $f$ has more than 11 non-stationary points.
    \begin{claim}\label{claim:above:13}
        Let $f \colon \dSet \to \dSet'$ be a function with $p \geq 12$ non-stationary points,
        then
        \[
        \mathbb{P}[\text{$f$ is an isomorphism from $\indGraph{T}{\dSet}$ to $\indGraph{T}{\dSet'}$ } ] \leq \frac{1}{2^{p \dCard/12}}.
         \]
    \end{claim}
    \begin{claimproof}
        Without loss of generality we can assume that $f$ is a bijection.
        Set $q \coloneqq \lfloor p/4 \rfloor$. Since $2 q < p$, observe that we can choose $q$ non-stationary points
        $u_1, \dots, u_q \in \dSet$ such that $\{u_1, \dots, u_q\} \cap \{f(u_1), \dots, f(u_q)\} = \emptyset$.
        We define $\dSet^\ast = \dSet \setminus \{u_1, \dots, u_q, f(u_1), \dots, f(u_q)\}$. Note that
        $|\dSet^\ast| \geq \dCard - 2q$.
        If $f$ is an isomorphism then for all $i \in [q]$ and $u \in \dSet^\ast$
        we have $(u_i, u) \in \DE{T}$ if and only if $(f(u_i), f(u)) \in \DE{T}$. This involves
        that the orientation of $q |\dSet^\ast|$ distinct pairs of edges have to coincide. Thus,
        $f$ is an isomorphism with probability at most
        \[\frac{1}{2^{q|\dSet^\ast|}} \leq \frac{1}{2^{q(\dCard - 2q)}}  \leq \frac{1}{2^{q(\dCard - \dCard/2)}}
        \leq \frac{1}{2^{q \dCard/2}} \leq \frac{1}{2^{p \dCard/12}}, \]
        where we use that $2q \leq \dCard/2$ and $q\geq p/6$.
    \end{claimproof}
    We now use
    \cref{claim:less:13} and \cref{claim:above:13} to upper bound \cref{eq:E:isom}:
    \[\mathbb{P}[E] \leq \sum_{p = 1}^{\dCard} \sum_{f \in \mathcal{F}_p }
        \mathbb{P}[\text{$f$ is isomorphism from $\indGraph{T}{\dSet}$ to $\indGraph{T}{\dSet'}$ } ]
        \leq \sum_{p = 1}^\dCard \frac{|\mathcal{F}_p|}{2^{p \dCard/12}}. \]
    Further by using $|\mathcal{F}_p| \leq \binom{\dCard}{p}k^p \leq \dCard^p k^p$ we obtain
    \[\mathbb{P}[E] \leq \sum_{p = 1}^\dCard \frac{\dCard^p k^p}{2^{p \dCard/12}}  \leq \sum_{p = 1}^\dCard \frac{k^{2p}}{2^{p \dCard/12}} =
    \sum_{p = 1}^\dCard 2^{2p\log(k) - p \dCard/12} = \sum_{p = 1}^\dCard 2^{p(2\log(k) - \dCard/12)},  \]
    where we use that $\dCard \leq k$.
    Note that for $k \geq 10^5$, we have $\zCard \leq k/2$ and hence
    $\dCard/12 \geq k/24 \geq 6\log(k)$.
    Thus, we may continue our computation with
    \[\mathbb{P}[E] \leq \sum_{p = 1}^\dCard 2^{-4p \log(k)} \leq \sum_{p = 1}^\dCard 2^{-4 \log(k) }
    = \sum_{p = 1}^\dCard \frac{1}{k^4} \leq \frac{1}{k^3}.  \]
    Finally, we obtain
    $\mathbb{P}[\overline{E}] = 1 - \mathbb{P}[E] \geq 1 - 1/k^3$.
\end{proof}

Next, we show \cref{theo:prob:TT} by proving that there is a
set $\dSet$ such that all the neighborhoods $\Tvec{v}{\dSet}$ with respect
to $\dSet$ are different.
The intuition is that two different neighborhoods
$\Tvec{v}{\dSet}$ and $\Tvec{u}{\dSet}$ behave like a binomial distributed
random variables with success probability $1/2$ and $|\dSet|$ many
repetitions. Hence, with high probability, they are different.

\begin{lemma}[Random tournaments satisfy the third property of TT-uniqueness: Random tournaments contain a large number of vertices with a unique
    neighborhood]\dglabel{lem:prob:vec}
    Consider a random tournament \(T\) with \(k \ge 10^5\) vertices $\{v_1, \dots, v_k\}$.
    Write \(\zCard \coloneqq \lfloor k / (9 \log(k)) \rfloor\) and set $\dSet \coloneqq \{v_{\zCard+1}, \dots, v_k\}$
    and $\delta \coloneqq |\dSet| - \TT{T} \cdot |\zSet|$.
    Then, with probability at least $1 - 2/k^3$,
    all $\dSet' \subseteq \dSet$ with $|\dSet'| \geq \delta$ and all
    $v \neq u \in V(T) \setminus \dSet'$ satisfy $\Tvec{v}{\dSet'} \neq \Tvec{u}{\dSet'}$.
\end{lemma}
\begin{proof}
    We start with the following claim.
    \begin{claim}\label{claim:tt:prob}
        The probability that $\TT{T} \leq 3 \log(k)$ is at least $1 - 1 / k^{3}$.
    \end{claim}
    \begin{claimproof}
        Let $c \coloneqq \lceil 3 \log(k) \rceil$, we show that the probability $P$ of the event that
        $T$ contains a subtournament $\indGraph{T}{A}$ of order $c$ with
        $\indGraph{T}{A} \cong \tranTour_c$
        is at most $k^{-3}$. Let $A \subseteq V(T)$ be a set of vertices with
        $|A| = c$, we show
        \[\mathbb{P}[\indGraph{T}{A} \cong \tranTour_c] = \frac{c!}{2^{\binom{c}{2}}},\]
        where the probability is taken with respect to the randomness of $T$.
        To this end, observe that there are $2^{\binom{c}{2}}$ possible tournaments with
        vertex set $A$ that are all equally likely and $c!$ many of them are transitive
        since a transitive tournament is uniquely described by its topological ordering.
        Union bound yields
        \[P \leq \sum_{\substack{A \subseteq V(T) \\|A| = c}} \mathbb{P}[\indGraph{T}{A} \cong \tranTour_c]
        = \binom{k}{c} \cdot \frac{c!}{2^{\binom{c}{2}}} \leq \frac{k^c}{2^{\binom{c}{2}}}
        = 2^{c \log(k) - (c^2 - \frac{c}{2})} =
        2^{c(\log(k) - c + \frac{1}{2})}.\]
        Observe that this expression becomes smaller for larger values for $c$. Since $c \geq 3 \log(k)$,
        we can continue with
        \[P \leq 2^{-3\log(k)(2 \log(k)  - 1/2)} \leq k^{-3}.\]
        Thus, the probability that $T$ contains a subtournament isomorphic to $\tranTour_c$ is
        at most $k^{-3}$. Note that this implies that the event $\TT{T} \leq 3 \log(k)$ has a
        probability of at least $1 - k^{-3}$.
    \end{claimproof}

    In the following, we assume that $\TT{T} \leq 3 \log(k)$ which yields
    \[\delta \coloneqq |\dSet| - \TT{T} \cdot |\zSet| \geq
    |\dSet| - k/3 \geq 2k/3 - \lfloor k / (9 \log(k)) \rfloor  \geq 5k/9.\]
    We show that it is enough to prove that the sets $\Tout{v}$
    are all very different. Set $V \coloneqq V(T)$.
    \begin{claim}\label{claim:dis:vec}
        If, for all $v \neq u \in V$, we have $|\,\symdif{\Tout{v}}{\Tout{u}}\,| \geq k - \delta + 1$,
        then for all $\dSet' \subseteq V$ with $|\dSet'| \geq \delta$ we have $\Tvec{v}{\dSet'} \neq \Tvec{u}{\dSet'}$.
    \end{claim}
    \begin{proof}
        Let $\dSet'$ be a set such that they are $v \neq u \in V$ with
        $\Tvec{v}{\dSet'} = \Tvec{u}{\dSet'}$. Then $\dSet' \cap (\symdif{\Tout{v}}{\Tout{u}}) = \emptyset$
        since otherwise $\Tout{v}$ and $\Tout{u}$ would disagree on a vertex in $\dSet'$.
        Thus, $\dSet' \subseteq V \setminus (\symdif{\Tout{v}}{\Tout{u}})$
        which implies $|\dSet'| \leq k - ( k- \delta + 1) < \delta$.
    \end{proof}
    Let $P'$ be the probability that for all $\dSet' \subseteq \dSet$ with $|\dSet'| \geq \delta$ and all
    $v \neq u \in V(T) \setminus \dSet'$ we have $\Tvec{v}{\dSet'} \neq \Tvec{u}{\dSet'}$. By \cref{claim:dis:vec}
    we obtain
    \begin{align*}
        P' &\geq \mathbb{P}\big[\forall v \neq u \in V: |\symdif{\Tout{v}}{\Tout{u}}| \geq k - \delta + 1\big]\\
           &=1 - \mathbb{P}[\exists v \neq u \in V: |\symdif{\Tout{v}}{\Tout{u}}| \leq k - \delta],
    \end{align*}
    where the probability is taken over the randomness of $T$. Union bound yields
    \begin{align}\label{eq:prob:symdif}
        P' \geq 1 - \sum_{v \neq u \in V} \mathbb{P}\big[ |\symdif{\Tout{v}}{\Tout{u}}| \leq k - \delta\big].
    \end{align}
    In the following, we estimate the probability of the event
    $|\symdif{\Tout{v}}{\Tout{u}}| \leq k - \delta$ for fixed vertices $v \neq u$.
    For any vertex $x_1 \in V \setminus \{u, v\}$, observe that $x_1 \in \Tout{v}$
    with probability $\frac{1}{2}$ since the edges $(x_1, v)$ and
    $(v, x_1)$ are equally likely in $T$. Further, this event is independent of
    the event $x_1 \in \Tout{u}$ that also has a probability of $\frac{1}{2}$ of occurring. Thus,
    the probability of $x_1 \in \symdif{\Tout{v}}{\Tout{u}}$ is equal to $\frac{1}{2}$.
    Further, let $x_2, x_3, \dotsc \in V \setminus \{u, v\}$ be other vertices, then the events
    $x_i \in \symdif{\Tout{v}}{\Tout{u}}$
    are all independent of each other and all have a success probability of $\frac{1}{2}$.
    Additionally, $v$ and $u$ are always in $\symdif{\Tout{v}}{\Tout{u}}$. Hence,
    \[\mathbb{P}[ |\symdif{\Tout{v}}{\Tout{u}}| \leq k - \delta] =
        \mathbb{P}[X + 2 \leq  k - \delta],
    \] where $X$ is a binomial distributed random variable
    on $k - 2$ events with success probability $\frac{1}{2}$. Let $\mu = k/2 - 1$
    be the expected value of $X$, then  $\mathbb{P}[X + 2 \leq  k - \delta] =
    \mathbb{P}[X \leq k - \delta - 2] \leq
    \mathbb{P}[X \leq (1 - 1/9) \cdot \mu]$,
    where the last step follows from $k-\delta-2\leq 4k/9 - 2 \leq 4k/9 - 8/9 = 8/9 \cdot \mu$.
    By~\cite[Theorem 4.5]{probAndComp} we upper bound the previous expression with
    \[\mathbb{P}[X \leq (1 - 1/9) \cdot \mu] \leq \mathrm{e}^{-\mu/162} \leq \mathrm{e}^{-k/500}
    \leq k^{-5}, \]
    where the last step follows for $k \geq 10^5$. By using \cref{eq:prob:symdif}, we obtain
    $P' \geq 1 -  \sum_{v \neq u \in V} k^{-5} \geq 1 - 1/k^3$. Lastly, by combining
    \cref{claim:tt:prob,claim:dis:vec}, we obtain, with probability at
    least $1 - 2/k^3$, that
    for all $\dSet' \subseteq \dSet$ with $|\dSet'| \geq \delta$ and all
    $v \neq u \in V(T) \setminus \dSet'$ we have $\Tvec{v}{\dSet'} \neq \Tvec{u}{\dSet'}$.
\end{proof}

\theoprobTT
\begin{proof}
    Set $V(T) = \{v_1, \dots, v_k\}$, $\zCard = \lfloor k / (9 \log(k)) \rfloor$,
    $\zSet \coloneqq \{v_1, \dots, v_\zCard\}$ and $\dSet \coloneqq \{v_{\zCard+1}, \dots, v_k\}$.
    By combining \cref{lem:prob:vec,lem:prob:iso} we obtain that $(\dSet, \zSet)$ is TT-unique
    with probability at least $1 - 3/k^3$.
\end{proof}

Now, combining \cref{theo:TT:reduction,theo:prob:TT} yields \cref{maintheorem:Dec:sub}.

\thmdecsub
\begin{proof}
    By \cref{theo:prob:TT}, with probability at
    least $1 - 3/k^3$ there is a partition
    $(\dSet, \zSet)$ of $V(T)$ such that $\zCard \coloneqq |\zSet| \geq \lfloor k / (9 \log(k))\rfloor$
    and $(\dSet, \zSet)$ is TT-unique.
    If there is an algorithm that reads the whole
        input and solves $\DecTourProb(\{T\})$ for any tournament of order $n$ in time $O(n^\gamma)$,
        then \cref{theo:TT:reduction} yields an algorithm that solves
        $\DecCliqueProb_{\lfloor k/(9\log(k)) \rfloor}$ for any graphs of order $n$
        in time $O(n^\gamma)$.

    Further, assuming ETH, there exists a constant
    $\alpha$ such that no algorithm, that reads the whole input, solves
    $\DecCliqueProb_{k}$ in time $O(n^{\alpha k})$ (see \cref{lemma:eth:clique}).
    By using $\beta = \alpha/10$, we obtain that no algorithm, that reads the whole input, solves
    $\DecTourProb(\{T\})$ in time $O(n^{\beta k/\log(k)})$.
\end{proof}

\subsection{The Complexity of \texorpdfstring{$\DecTourProb$}{Dec-IndSub}}

\thmdecsubpara
\begin{proof}
    \Cref{theo:prob:TT} ensures the existence of a graph
    $T$ that has a partition $(\dSet, \zSet)$ that is TT-unique with $|\zSet| \geq c$. We use this to construct
    set of tournaments $\mathcal{T}_c = \{T_0, T_1, T_2, \dots\}$, where $T_k$
    is obtained by adding $\tranTour_k$ to $T$ and adding the edges
    $(u, v)$ for $u \in V(T)$ and $v \in V(T_k) \setminus V(T)$ to $\DE{T_k}$. Note that
    $(V(T), \emptyset, V(\tranTour_k))$
    is a spine decomposition of $T_k$ implying that $|V(T_k)| - \core{T_k} \leq |V(T)|$
    for all $T_k \in \mathcal{T}_c$. Thus, \cref{theo:core:easy} yields that $\DecTourProb(\mathcal{T}_c)$
    is FPT.

    For the second part, assume ETH. By \cref{lemma:eth:clique}
    there is a global constant $\alpha > 0$ such that no algorithm, that reads the whole input, solves
    $\DecCliqueProb_k$ on graphs of order $n$ in time $O(n^{\alpha k})$.
    If there is an algorithm that reads the whole input and solves $\DecTourProb(\mathcal{T}_c)$
    in time $O(f(k) \cdot n^{\alpha c})$ for some computable function $f$,
    then $\DecTourProb(\{T_0\})$ can be solved in time $O(n^{\alpha c})$.
    However, since $|\zSet| \geq c$ and due to \cref{maintheorem:Dec:sub} this would imply that
    $\DecCliqueProb_{c}$ can be solved in time $O(n^{\alpha c})$, which is not possible unless ETH fails.
\end{proof}

\bibliographystyle{alphaurl}
\bibliography{main}

\appendix
\clearpage
\section{On the Complexity of Colored Subgraph Counting}

\begin{lemma}[{$\CPsubProb(\{H\})$ is harder than $\GraphCFsubProb(\{H\})$}]\dglabel{lem:cp:cf}
    For a $k$-labeled graph $H$, assume that there
    is an algorithm that
    computes $\CPsubProb(\{H\})$
    for any graphs of order $n$ in time $O(f(n))$.
    Then there is an algorithm that computes $\GraphCFsubProb(\{H\})$ for any $k$-colored graph
    $G$ of order $n$ in time $O(g(k) \cdot f(n))$ for some computable function $g$.
    In particular, $\Exp{\CPsubProb(\{H\})} \geq \Exp{\GraphCFsubProb(\{H\})}$.
\end{lemma}
\begin{proof}
    Let $G$ be $k$-colored graph of order $n$ with coloring $c \colon V(G) \to [k]$.
    For a permutation $\sigma \in \sym{k}$, we write $(G, \sigma \circ c)$ for the
    $k$-colored graph $G$ with coloring $\sigma \circ c$. The statement immediately follows from
    \begin{align}\label{eq:cf:cp}
        \sum_{\sigma \in \sym{k}} \CPsubs{H}{(G, \sigma \circ c)} = |\auts{H}| \cdot \GraphCFsubs{H}{G}.
    \end{align}
    To show \cref{eq:cf:cp}, we first prove the following claim.
    \begin{claim}\label{claim:auts}
        Let $A \subseteq V(G)$ and $S \subseteq E(G) \cap \binom{A}{2}$ with
        $\edgesub{\indGraph{G}{A}}{S}$ being colorful with respect to $c$ and isomorphic to $H$ then
        \[|\{\sigma \in \sym{k} : \text{$\sigma \circ c$ is an isomorphism from
        $\edgesub{\indGraph{G}{A}}{S}$ to $H$}\}| = |\auts{H}|.\]
    \end{claim}
    \begin{proof}
        Let $H'$ be the image of $\edgesub{\indGraph{G}{A}}{S}$ with respect to $c$. Observe
        that $H'$ is isomorphic to $H$. According to~\cite[Theorem 4]{Hoffmann82},
        the set of isomorphisms from $H'$ to $H$
        is equal to $\{\varphi \circ \psi : \varphi \in \auts{H}\}$ where $\psi \in \sym{k}$.
        Thus,
        $\{\sigma \in \sym{k} : \text{$\sigma \circ c$ is an isomorphism from
        $\edgesub{\indGraph{G}{A}}{S}$ to $H$}\} = \{ \varphi \circ \psi : \varphi \in \auts{H} \}$
        proving the claim.
    \end{proof}

    Note that $\GraphCFsubs{H}{G} = |\{A, S : \text{$\edgesub{\indGraph{G}{A}}{S}$ colorful with respect to $c$, isomorphic to $H$}\}|$.
    Hence, \Cref{claim:auts} yields
    $|\auts{H}| \cdot \GraphCFsubs{H}{G} = |\{A, S, \sigma : \text{$\sigma \circ c$
    is isomorphism from $\edgesub{\indGraph{G}{A}}{S}$ to $H$}\}|$.
    However, note that $\CPsubs{H}{(G, \sigma \circ c)} = |\{A, S : \text{$\sigma \circ c$
    is isomorphism from $\edgesub{\indGraph{G}{A}}{S}$ to $H$}\}|$. By summing
    over all possible permutations on the left side, we obtain
    \cref{eq:cf:cp}.

    Assume that we can compute  $\CPsubProb(\{H\})$
    for graphs of order $n$ in time $O(f(n))$. Now \cref{eq:cf:cp} allows us to
    compute $\GraphCFsubProb(\{H\})$ by computing $|\auts{H}|$ and calling
    $\CPsubProb(\{H\})$ a total of $k!$ times. This takes time
    $O(g(k) \cdot f(n))$ for some computable function $g$, proving the lemma.
\end{proof}

\lemcolindsub
\begin{proof}
    Let $c$ be the coloring of $G$.
    We prove the statement by using the inclusion-exclusion principle. Let
    $H$ be a $k$-labeled graph and $S \subseteq \binom{[k]}{2}$ be a set of edges, we write
    $H \cup S$ for the graph that we obtain by adding $S$ to the edge set of $H$.
    For all $S \subseteq \overline{E(H)} \coloneq E(K_k) \setminus E(H)$, we define
    \begin{align*}
        B &\coloneqq \{A \subseteq V(G) : \text{$\indGraph{G}{A}$ is colorful, and it contains a subgraph isomorphic to $H$ under $c$}\} \\
        A_S &\coloneqq \{A \subseteq V(G) : \text{$\indGraph{G}{A}$ is colorful, and it contains a subgraph isomorphic to $H \cup S$ under $c$}\}
   \end{align*}
    Observe that $A_S \subseteq B$ because, for $A \subseteq V(G)$, if $\indGraph{G}{A}$
    contains a subgraph isomorphic to $H \cup S$ then it also contains a subgraph isomorphic to $H$.
    Further, we consider $B$ as our base set, that is to say we
    define $\overline{A_S} \coloneqq B \setminus A_S$. Also,
    if $\indGraph{G}{A}$ contains a subgraph that is isomorphic to $H \cup S$ under $c$,
    then this subgraph is unique in $\indGraph{G}{A}$ due to $c$. Hence $|A_S| = \CPsubs{H \cup S}{G}$.
     In the following, we write $A_i$ for $A_{\{i\}}$. We show
    \begin{align}\label{eq:inc:exc:cpindsub}
    \CPindsubs{H}{G} = \left| \bigcap_{i \in \overline{E(H)} } \overline{A_{i}}\right|.
    \end{align}
    To see this, note that $\CPindsubs{H}{G}$ counts exactly those induced subgraphs $\indGraph{G}{A}$
    that are isomorphic to $H$ under $c$. This is equivalent to, for all $i \in \overline{E(H)}$, the  induced subgraphs $\indGraph{G}{A}$ does not contain any
    subgraph that is isomorphic to
    $H \cup \{i\}$. Further, $\overline{A_i}$ consists of
    those graphs $G[A]$ that contain a subgraph isomorphic to $H$ under $c$
    but do not contain a subgraph isomorphic to $H \cup \{i\}$ under $c$.
    Hence, the intersection on the right hand side consists of all induced subgraphs $\indGraph{G}{A}$ of $G$
    that isomorphic to $H$ under $c$. This shows \cref{eq:inc:exc:cpindsub}.
    Next, observe that for all $S \subseteq \overline{E(H)}$
    \[
    \left|\bigcap_{i \in S} A_{i}\right| = |A_S| = \CPsubs{H \cup S}{G}.
    \]

    By applying the inclusion-exclusion principle, we obtain
    \begin{align*}
        \CPindsubs{H}{G} &= \sum_{S \subseteq \overline{E(H)}} {(-1)}^{|S|} \left|\bigcap_{i \in S} A_{i}\right|  \\
        &= \sum_{S \subseteq \overline{E(H)}} {(-1)}^{|S|} \CPsubs{H \cup S}{G} \\
        &= \sum_{H \subseteq H'} {(-1)}^{|E(H')| - |E(H)|} \cdot \CPsubs{H'}{G}.
    \end{align*}
\end{proof}

\begin{restatable}[{$\CPsubProb(\{H\})$ is harder than $\CPsubProb(\{H'\})$ for $H'$ minor of $H$
   ~\cite[Modification of Lemma 5.8]{DBLP:phd/dnb/Curticapean15}}]{lemma}{lemcpminor}\dglabel{lem:cpsub:minor}
    Let $H$ be a graph and  $H'$ be a minor of $H$.
    Assume that there is an algorithm that computes $\CPsubProb(\{H\})$
    for any $k$-colored tournament of order $n$ in time $O(f(n))$.
    Then there is an algorithm that computes $\CPsubProb(\{H'\})$ for any $k$-colored tournament
    of order $n$ in time $O(f(n))$.
    In particular, $\Exp{\CPsubProb(\{H\})} \geq \Exp{\CPsubProb(\{H'\})}$.
\end{restatable}
\begin{proof}
    Let $H$ be a $k$-labeled graph, $H'$ a $k'$-labeled graph with $H'$ being a minor of $H$ and
    $G'$ a $k'$-colored graph.
    By the proof of Lemma 5.8 in~\cite{DBLP:phd/dnb/Curticapean15},
    we can construct a $k$-colored graph $G$ in time $O(|V(H)|^2 \cdot |V(G)|^2)$ such that
    $\CPsubs{H'}{G'} = \CPsubs{H}{G}$ (see equation (5.4) in~\cite{DBLP:phd/dnb/Curticapean15}).\footnote{
        Note that $\NUM{}\mathrm{PartitionedSub}(H \to G)$ is the same as $\CPsubs{H}{G}$.
    }
    This immediately yields the result since $|V(H)|^2$ is constant whenever $H$ is fixed.
\end{proof}

\section{On the Complexity of Counting Colored Subtournaments}\label{appendix:colored}

\lemsubscolor
\begin{proof}
    The proof follows from Lemma 2.4 in~\cite{Yuster25} and is repeated here
    for completeness.

    Let $T$ be a $k$-labeled tournament and $G$ be a $k$-colored tournament of order $n$.
    For $S \subseteq [k]$,
    we define $G_S$ as the subtournament of $G$ that is obtained by deleting
    all vertices whose color is not in $S$. Observe that $G_S$ can be computed in time
    $O(n^2)$ and has at most $n$ vertices. By the inclusion-exclusion principle, we obtain
    \[\CFsubs{T}{G} = \sum_{\substack{S \subseteq [k]}} {(-1)}^{|S|} \cdot \tours{T}{G_S}.\]
    If $\tours{T}{G}$ can be computed time $O(f(n))$,
    then we use the above equality to compute $\CFsubs{T}{G}$ in time $O(2^{k} \cdot f(n))$
    by calling $\tourProb(\{T\})$ at most $2^{|V(T)|}$ times on the tournaments of order at most $n$.

    For a recursively enumerable set of tournaments $\mathcal{T}$,
    we use the above construction to obtain a parameterized Turing reduction
    from $\CFsubProb(\mathcal{T})$ to $\tourProb(\mathcal{T})$. Observe that, for an input $(T, G)$,
    we only have to query $\tourProb(\mathcal{T})$ at most $2^{|V(T)|}$ times
    on inputs $(T, G')$ with $|V(G')| \leq |V(G)|$. Further, all other computations can be done
    in time $O(h(k) \cdot n^2)$ for some computable function $h$.
\end{proof}

\section{Reduction from \texorpdfstring{$\cliqueProb$}{\#Clique}
to \texorpdfstring{$\CFcliqueProb$}{\#cf-Clique}}

In this section, we show how to remove colors when counting cliques.
The proof is  due to Curticapean and originates
from his PhD thesis~\cite{DBLP:phd/dnb/Curticapean15}.

\begin{lemma}[{Removing colors for cliques,~\cite[Lemma 1.11]{DBLP:phd/dnb/Curticapean15}}]\dglabel{lem:CFclique:clique:construction}
    Given an integer $k$ and a graph $G$ of order $n$, one can construct a
    $k$-colored graph $G'$
    of order $k \cdot n$ in time $O(k^2 \cdot n^2)$ such that
    \[\clique{k}{G} = \CFclique{k}{G'}.\]
\end{lemma}
\begin{proof}
    Without loss of generality, we assume that $V(G) = \setn{n}$ (otherwise, we might
    choose an arbitrary ordering on the vertices). We construct a
    $k$-colored graph $G'$ in the following way. The vertex set of $G'$ is
    $V(G) \times \setn{k}$. The coloring of $G'$ is defined as $c' \colon V(G) \to \setn{k},
    (v, i) \mapsto i$. Finally, $\{(u, i), (v, j)\}\in E(G')$ if and only if
    $u < v$, $i < j$, and $\{u, v\} \in E(G)$.
    Observe that $G'$ has order $k \cdot n$ and can be computed in time $O(k^2 \cdot n^2)$.
    It remains to show that $\clique{k}{G} = \CFclique{k}{G'}$. To this end,
    let
    \[
        \mathcal{K} \coloneq  \{A \subseteq V(G') : \text{$\indGraph{G'}{A}$ is a colorful $k$-clique }\}
    \]
    be the set of colorful $k$-cliques in $G'$.
    \begin{claim}\label{claim:clique:bij}
        The order of $\mathcal{K}$ is equal to the order of
        \[\mathcal{S} \coloneqq \{(u_1, \dots, u_k) \in {V(G)}^k :
        \text{For all $1 \leq i < j \leq k$ we have $u_i < u_j$ and $\{u_i, u_j\} \in E(G)$ }\}.\]
    \end{claim}
    \begin{proof}
        We define a bijection $C \colon \mathcal{S} \to \mathcal{K}$.
        For $S \coloneqq (u_1, \dots, u_k) \in \mathcal{S}$, we define
        $C(S) \coloneqq \{(u_i, i) : i \in \setn{k}\}$. First, observe that
        $\indGraph{G'}{C(S)}$ is always colorful. Further,
        $\indGraph{G'}{C(S)}$ is a clique since for all $1 \leq i < j \leq k$
        we have that $\{(u_i, i), (u_j, j)\}$ is an edge in $G'$. Hence,
        $C(S) \in \mathcal{K}$ and therefore $C$ is well-defined.

        Next, $C$ is injective since $C(S) = C(S')$
        immediately implies $S = S'$.
        Lastly, we show that $C$ is surjective. Let $K \coloneqq
        \{(u_i, i)  : i \in \setn{k}\} \in \mathcal{K}$ be a colorful $k$-clique in $G'$.
        By definition of $E(G')$, for every $(u_i, i) \neq (u_j, j) \in K$ we have $i \neq j$.
        Further, we obtain for all $i < j$ that $u_i < u_j$
        and $\{u_i, u_j\} \in E(G)$. Thus, $K = C(S)$ for $S = (u_1, \dots, u_k)$
        with $S \in \mathcal{S}$.
    \end{proof}
    By definition, we have $|\mathcal{K}| = \CFclique{k}{G'}$.
    Moreover, observe that  $|\mathcal{S}| = \clique{k}{G}$.
    To see this, note that each $k$-clique $\indGraph{G}{A}$
    has a unique ordering $A = \{u_1, \dots, u_k\}$ with $u_i < u_j$.
    Now, \cref{claim:clique:bij} yields $\clique{k}{G} = \CFclique{k}{G'}$.
\end{proof}

We continue with the reduction for counting colorful cliques.

\begin{lemma}[{$\CFcliqueProb_{k}$ is harder than $\cliqueProb_{k}$}]\dglabel{lem:CFclique:clique}[lem:CFclique:clique:construction]
    Assume that there is an algorithm that reads the whole input and computes $\CFcliqueProb_{k}$
    for any graph of order $n$ in time $O(n^\gamma)$.
    Then there is an algorithm that computes $\cliqueProb_{k}$ for any graph
    of order $n$ in time $O(n^\gamma)$. In particular, $\Exp{\CFcliqueProb_{k}} \geq \Exp{\cliqueProb_{k}}$.
    Further, $\cliqueProb \fpt \CFcliqueProb$.
\end{lemma}
\begin{proof}
    Let $G$ be an undirected graph and $k$ be an integer. By \cref{lem:CFclique:clique:construction},
    we can construct a graph $G'$ of order $k \cdot n$ in time $O(k^2 \cdot n^2)$ such that
    $\clique{k}{G} = \CFclique{k}{G'}$. Thus, given an algorithm that reads
    the input and computes $\CFcliqueProb_{k}$ for any graph
    of order $n$ in time $O(n^\gamma)$, then we can use
    this algorithm to solve $\cliqueProb_{k}$ in time $O(k^2 \cdot n^2 + {(k \cdot n)}^\gamma)$.
    Note that this running time is in $O(n^\gamma)$
    since $k$ is fixed and $\gamma \geq 2$ (because the algorithm reads the whole input).

    We use the same construction to obtain a parameterized Turing reduction
    from $\cliqueProb$ to $\CFcliqueProb$.
\end{proof}

Observe that we also obtain a reduction
for the decision variant.

\begin{lemma}[{$\DecCFcliqueProb_{k}$ is harder than $\DecCliqueProb_{k}$}]
    \dglabel{lem:DecCFclique:DecClique}[lem:CFclique:clique:construction]
    Assume that there is an algorithm that reads the whole input and solves $\DecCFcliqueProb_{k}$
    for any graph of order $n$ in time $O(n^\gamma)$.
    Then there is an algorithm that solves $\DecCliqueProb_{k}$ for any graph
    of order $n$ in time $O(n^\gamma)$. In particular, $\DecCFcExp{k} \geq \DeccExp{k}$
    and $\DecCliqueProb \fpt \DecCFcliqueProb$.
\end{lemma}
\begin{proof}
    The proof of this statement is completely analog to the proof of \cref{lem:CFclique:clique}.
\end{proof}

\section{The Complexity of Finding Colorful Tournaments}\label{sec:DecCFsubProb}

\subsection{Hardness via the Signature}

Let $T$ denote a tournament of order $k$ and let $R$ be a signature of $T$ with $r \coloneqq |R|$.
Given a $(k-r)$-colored graph $G$, we construct a tournament $G^\ast$
by starting with the tournament $\tourG{G}{T}$, that is naturally $(k-r)$-colored (see \cref{def:TourG}),
and then adding the subtournament $\indGraph{T}{R}$ to it. This yields
a $k$-colored tournament since each vertex in $R$ has its own color.

Now, if $A \subseteq V(G^\ast)$ with $\indGraph{G^\ast}{A}$
is colorful and isomorphic to $T$, then $A$ must contain $R$ since
each vertex of $R$ has its own color. Further, due to the definition of signature, we
cannot flip any edges in the $\tourG{G}{T}$ part of $\indGraph{G^\ast}{A}$.
This then yields that $A  \setminus R$ is a colorful clique in $G$.

\yusterred

\begin{proof}
    The proof mostly follows the proof of Lemma 2.5 in~\cite{Yuster25}.
    Let $T$ be a tournament with vertex set $V(T) = \{v_1, \dots, v_k\}$.
    Without loss of generality, we assume $R = \{v_1, \dots v_r\}$ (otherwise
    we reorder the vertices).
    Given a $(k-r)$-colored graph $G$ with coloring $c \colon V(G) \to  \setn{k-r}$,
    we first construct a $k$-colored tournament $G^\ast$
    with coloring $c^\ast \colon V(G) \to \setn{k}$  in the following way.
    We define $V(G^\ast) = \{v_1, \dots v_r\} \uplus V(G)$, where $\{v_1, \dots v_r\}$
    are new vertices that are not in $V(G)$. The coloring of $G^\ast$
    is given by $c^\ast$ which is defined as $c^\ast(v_i) = i$
    and $c^\ast(x) = r + c(x)$ for all $x \in V(G)$. Given $u, v \in V(G^\ast)$,
    we orientate the edge between $u$ and $v$ in the following way.
    For $u, v \in V(G)$ with $u \neq v$:
    \begin{itemize}
        \item If $c(u) = c(v)$ then use an arbitrary orientation.
        \item Else if $\{u, v\} \in E(G)$ then $(u,v)\in \DE{G^\ast}$
        if and only if $(c^\ast(u),c^\ast(v))\in \DE{T}$.
        \item Else if $\{u,v\} \notin E(G)$ then $(u,v)\in \DE{G^\ast}$
        if and only if $(c^\ast(v),c^\ast(u))\in \DE{T}$.
    \end{itemize}

    Otherwise, at least one of $u,v$ belongs to $\{v_1,\dots,v_r\}$, and in particular $c^\ast(u) \neq c^\ast(v)$. In this case, we orientate
    the edge $\{u, v\}$ in $G^\ast$ such that its orientation is the same
    as $\{c^\ast(u), c^\ast(v)\}$ in $T$.\footnote{i.e., $(u, v) \in \DE{G^\ast}$ if and only if
    $(c^\ast(u), c^\ast(v)) \in \DE{T}$.}
    Note that $G^\ast$ has order $n + r$ and can be computed in time $O({(n + r)}^2)$.

    We now claim that $\CFsubs{T}{G^\ast} = \CFclique{k-r}{G}$.
    To this end, let $A \subseteq V(G^\ast)$ be a set of vertices such that $\indGraph{G^\ast}{A} \cong T$
    and $\indGraph{G^\ast}{A}$ is colorful.
    Observe that $c^\ast$ restricted to $A$ is an isomorphism
    from $\indGraph{G^\ast}{A}$ to a tournament $T^\ast$ with vertex set $\setn{k}$ and
    edge set $\{(c^\ast(u), c^\ast(v)) : (u, v) \in \DE{\indGraph{G^\ast}{A}} \}$.

    In the following, we show that for $B \coloneqq A \setminus \{v_1, \dots v_r\} $ the
    subgraph $\indGraph{G}{B}$ is a colorful $(k-r)$-clique in $G$.
    First note that $c(B) = \setn{k - r}$ since otherwise $c^\ast(A) \neq \setn{k}$.
    Next, we assume that $\indGraph{G}{B}$ is not a $(k-r)$-clique
    and show that this assumption leads to a contradiction.
    According to our assumption, there are vertices $u, v \in B$
    with $\{u, v\} \notin E(G)$. By construction of $G^\ast$,
    this implies that $\{u, v\}$ in $\indGraph{G^\ast}{A}$ does not have the same orientation
    as $\{c^\ast(u), c^\ast(v)\}$ in $T$, implying that $T^\ast$ and in $T$
    have an opposite orientation on $\{c^\ast(u), c^\ast(v)\}$. However, let $\{v_i, x\}$
    be any edge with at least one endpoint in $R$. By construction,
    $\{v_i, x\}$ in $G^\ast$ has the same orientation
    as $\{c^\ast(v_i), c^\ast(x)\}$ in $T$. Thus,
    $T$ and $T^\ast$ have the same orientation on $\{v_i, x\}$.
    This means that $T^\ast$ is obtained from $T$
    by only changing the orientation of edges that are not incident to $R$.
    However, since $\indGraph{G^\ast}{A} \cong T$, we also get
    $T^\ast \cong T$ which is a contradiction to $R$ being a signature. This
    proves that $\indGraph{G}{B}$ is a colorful $(k-r)$-clique in $G$.

    In contrast, let $B \subseteq V(G)$ be a set of vertices with
    $\indGraph{G}{B}$ being a colorful $(k-r)$-clique in $G$.
    For $A \coloneqq B \uplus \{v_1, \dots v_r\}$,
    we show that $\indGraph{G^\ast}{A}$ is a colorful subtournament
    that is isomorphic to $T$. First note that $A$ is colorful
    since, for $1 \leq i \leq r$, $c^\ast(v_i) = i$ and $c^\ast(B) = \{r+1, \dots, k\}$.
    Next, we show that $(u, v) \in \DE{\indGraph{G^\ast}{A}}$
    if and only if $(c^\ast(u), c^\ast(v)) \in \DE{T}$.
    If $u, v \in B$ then $\{u, v\} \in E(G)$ which implies that
    $\indGraph{G^\ast}{A}$ and $T$ have the same orientate on
    $\{u, v\}$ and  $\{c^\ast(u),c^\ast(v) \}$.
    In contrast, if either
    $u \notin B$ or $v \notin B$ then at least one vertex of $\{u,v\}$ belongs to $\{v_1, \dots v_r\}$.
    By construction, $\indGraph{G^\ast}{A}$ and $T$ have the same orientate on
    $\{u, v\}$ and $\{c^\ast(u), c^\ast(v)\}$. This shows
    that $c^\ast$ defines an isomorphism from $\indGraph{G^\ast}{A}$
    to $T$.

    In summary, we showed that there is one-to-one relation between colorful subtournament in
    $G^\ast$ that are isomorphic to $T$ and colorful $(k-r)$-clique in $G$.
    This proves the lemma.
\end{proof}

\begin{remark}\label{7-15-3}
    One can also use the construction described above to obtain a tournament $G^\ast$
    by fixing a set of vertices $R$ that is not necessarily a signature. If we plug this into
    $\CFsubs{T}{G^\ast}$ then we can again represent it as a linear combination of $\CPindsubProb$-counts
    (like in \cref{lem:tour:colindsub}). However,
    since we fixed $\indGraph{T}{R}$ in the construction of $G^\ast$,
    we only obtain terms $\CPindsubs{H}{G}$, where $T_{H}$ is isomorphic to $T$
    and $T_{H}$ only flips edges that are non-adjacent to $R$.\footnote{i.e., $E(H)$ contains
    all edges that have at least one vertex in $R$.}
    Thus, if $R$ is a signature, then the only viable graph in the linear combination of
    \cref{lem:tour:colindsub} is $K_k$ since
    $T_{K_k} \cong T$ and $T_{K_k}$ does not flip any edges.
    Therefore, fixing the edges eliminates all terms in the linear combination of \cref{theo:tour:ae} except the term
    that is responsible for counting cliques. Hence, we can see
    \cref{lem:tour:colindsub} and \cref{theo:tour:ae} as a generalization of \cref{lem:CFsubs:CFclique}.
    Let $R \subseteq V(T)$ be a set of vertices and let $H_R$ be the $k$-labeled
    graph that is obtained by connecting all vertices in $\setn{k} \setminus R$
    with each other. Further, each vertex $R$ in $H_R$ is an isolated vertex.
    Then $R$ is a signature of $T$ whenever $\ae{T}{H_R} \neq 0$.
\end{remark}

\begin{theorem}[$\DecCFsubProb(\{T\})$ is harder than $\DecCliqueProb_{|V(T)|-\sig(T)}$]
    \dglabel{theo:sig:lower:bounds}[lem:CFsubs:CFclique,lem:DecCFclique:DecClique]
    Let $T$ be a tournament of order $k$. Assume that there
    is an algorithm that solves $\DecCFsubProb(\{T\})$
    for tournaments of order $n$ in time $O(n^\gamma)$, then there is an algorithm
    that solves
    $\DecCliqueProb_{k - \sig(T)}$ for any graph of order $n$ in time
    $O(n^\gamma)$.
    In particular, $\DecCFtExp{T} \geq \DeccExp{(k - \sig(T))}$.

    Further, given a r.e.\ set $\mathcal{T}$ of infinitely
    many tournaments such that $\{V(T) - \sig(T) : T \in \mathcal{T}\}$ is unbounded
    then $\DecCFsubProb(\mathcal{T})$ is \Dw-hard.
\end{theorem}
\begin{proof}
    Assume first that there is an algorithm $\mathbb{A}$
    that solves $\DecCFsubProb(\{T\})$ for tournaments of order $n$
    in time $O(f(n))$. We use $\mathbb{A}$ to create an algorithm $\mathbb{B}$
    that solves $\DecCFcliqueProb_{k-\sig(T)}$ on graphs of order $n$ in time
    $O(g(k) \cdot f(n))$ for some computable function $g$.

    Given a $(k-\sig(T))$-colored graph $G$ with $n$ vertices,
    we start by computing a signature $R \subseteq V(T)$ with $r \coloneqq |R| = \sig(T)$.
    Observe that this takes time $O(g'(k))$ for some computable function $g'$.
    Next, due to \cref{lem:CFsubs:CFclique}, we can compute a $k$-colored tournament $G^\ast$
    of order $(n + r)$ in time $O({(n+r)}^2)$
    such that $G$ contains a colorful $(k-r)$-clique if and only
    if $G^\ast$ contains a colorful copy of $T$, meaning that we can simply
    return $\DecCFsubs{T}{G^\ast}$. Observe that this algorithm solves $\DecCFcliqueProb_{k-r}$
    in time $O(g''(k) \cdot {(n + r)}^\gamma)$ for some computable function $g''$.
    Further, \cref{lem:CFclique:clique:construction} implies
    that $\DecCliqueProb_{k - r}$ can be solved in time $O(h(k) \cdot f(n))$. Note
    that this implies an $O(n^\gamma)$ for $\DecCliqueProb_{k - r}$ since $r \leq k$
    and $k$ is fixed.

    Lastly, we use the above construction to obtain a parameterized Turing reduction
    from $\DecCliqueProb$ to $\DecCFsubProb(\mathcal{T})$. Observe that on
    input $(G, k)$ with $n \coloneqq |V(G)|$, we first find a graph
    $T \in \mathcal{T}$ with $k' \geq k$ for $k' \coloneqq |V(T)| - \sig(T)$
    in time $h'(k)$ for some computable function $h'$. Note that the size of $k'$
    and $|V(T)|$
    is independent of $G$ and that there exists a computable function $f$ with
    $f(k) = |V(T)|$. By adding $k' - k$ apices to $G$,
    we obtain a new graph $G'$ such that $\DecClique{k}{G} = \DecClique{k'}{G'}$.
    Next, we can use the construction from above and \cref{lem:CFclique:clique:construction}
    to compute $\DecClique{k'}{G'}$ by querying $\DecCFsubProb(\{T\})$
    on a graph of order at most $k' \cdot (n + |V(T)|)$. Further, all other
    computations take time $O(h''(k) \cdot n^2)$ for some computable
    function $h''$. This shows that $\DecCFsubProb(\mathcal{T})$ is \Dw-hard.
\end{proof}

\subsection{\texorpdfstring{$\DecCFsubProb(\{T\})$}{Dec-cf-IndSub({T})} is Hard}

\begin{theorem}[$\DecCFsubProb(\{T\})$ is hard]\dglabel{maintheorem:Dec:Cf:sub}[theo:sig:lower:bounds, lem:sig:lower:bound]
    Write $\mathcal{T}$ for a recursively enumerable class of directed graphs.
    \begin{itemize}
        \item The problem $\DecCFsubProb(\mathcal{T})$ is \Dw-hard if
        $\mathcal{T}$ contains infinitely many tournaments and FPT otherwise.
        \item Let $T$ be a tournament of order $k$. Assume that there
        is an algorithm that reads the whole input and solves $\DecCFsubProb(\{T\})$
        for tournaments of order $n$ in time $O(n^\gamma)$, then there is an algorithm
        that solves $\DecCliqueProb_{\lceil \log(k)/4 \rceil}$ for all graphs of order $n$ in time
        $O(n^\gamma)$.

        Further, assuming ETH, there is a global constant $\beta > 0$
        such that no algorithm that reads the whole input solves $\DecCFsubProb(\{T\})$
        for any tournament of order $n$ in time $O(n^{\beta \log(k)})$.
        \qedhere
    \end{itemize}
\end{theorem}
\begin{proof}
    For the second part,
    let $T$ be a tournament with $k$ vertices. By \cref{lem:sig:lower:bound}
    each tournament $T$ contain a signature $R$ of size at most $k - \lceil \log(k)/4 \rceil$.
    Thus, the second part of the theorem directly follows from \cref{theo:sig:lower:bounds}.
    For the ETH result, note that we can choose $\beta = \alpha/5$, where
    $\alpha$ is the global constant from \cref{lemma:eth:clique}.

    For the first part, assume that $\mathcal{T}$ contains finitely many tournaments. Then, there is a $k$ such that
    $|V(T)| \leq k$ for all tournaments $T \in \mathcal{T}$. If $T \in \mathcal{T}$ is not a tournament
    then $\DecCFsubs{T}{G} = 0$ for all input tournaments $G$. Otherwise,
    we solve $\DecCFsubs{T}{G}$ in time $O(k^2 \cdot |V(G)|^k)$ by using a brute force algorithm.

    Otherwise, $\mathcal{T}$ contains infinitely many tournaments.
    \cref{lem:sig:lower:bound} implies for all tournaments $T \in \mathcal{T}$
    that $V(T) - \sig(T) \geq \lceil \log(|V(T)|)/4$. Hence, \cref{theo:sig:lower:bounds}
    yields that $\DecCFsubProb(\mathcal{T})$ is \Dw-hard.
\end{proof}

\begin{remark}\label{remark:no:reduction:dec:cf}
    Note that \cref{maintheorem:Dec:Cf:sub} shows that $\DecCFsubProb(\mathcal{T})$ is \Dw-hard
    as long as $\mathcal{T}$ contains infinitely many different tournaments. However,
    \cref{maintheorem:Dec:sub} shows that there are many different class of tournaments $\mathcal{T}_c$
    for which $\DecTourProb(\mathcal{T}_c)$ is FPT. Thus, there is no reduction from
     $\DecCFsubProb(\mathcal{T})$ to $\DecTourProb(\mathcal{T})$ which is contrary to
     the counting case (see \cref{lem:subs:to:color}).
\end{remark}

\tableofresults
{\color{color3}\bfseries Highlighted} are what we consider to be the main novel technical
ideas of this work.
\end{document}